\documentclass[preprint,12pt]{elsarticle}
\usepackage{amssymb}
\usepackage{amsmath,amsfonts,bm}
\usepackage{amsthm}
\usepackage{algorithm}
\usepackage{physics}
\usepackage{float}
\usepackage{enumerate}
\usepackage{caption}
\usepackage{subcaption}
\usepackage[colorlinks,linkcolor=red,anchorcolor=blue,citecolor=green]{hyperref}
\usepackage{geometry}
\usepackage{mathrsfs}
\usepackage{xcolor}
\usepackage{enumitem}

\newtheorem{proposition}{Proposition}

\journal{Journal of Computational Physics}

\begin{document}

\begin{frontmatter}

\title{Unified Gas-Kinetic Wave-Particle Method for \\ Multiscale Simulation of Vlasov-Poisson-Fokker-Planck System}

\author[a]{Zhigang Pu}
\ead{zgpuac@ust.hk}
\author[c]{Chang Liu}
\ead{liuchang@iapcm.ac.cn}
\author[a]{Yixiao Wang}
\ead{ywangyp@connect.ust.hk}
\author[a,b]{Kun Xu\corref{cor1}}
\ead{makxu@ust.hk}
\cortext[cor1]{Corresponding author}

\address[a]{{Department of Mathematics, Hong Kong University of Science and Technology},
            {Clear Water Bay, Kowloon},
            {Hong Kong},
            {China}}
\address[b]{{Shenzhen Research Institute, Hong Kong University of Science and Technology},
            {Shenzhen},
            {China}}
\address[c]{
Institute of Applied Physics and Computational Mathematics, Beijing, China}
\begin{abstract}
A unified gas-kinetic wave--particle method with Fokker--Planck collisions (UGKWP-FP) is developed for the Vlasov--Poisson--Fokker--Planck system. The collision operator is modeled by the Lenard--Bernstein operator, whose stochastic representation corresponds to the Ornstein--Uhlenbeck process in velocity space. To extend the UGKWP framework beyond the conventional Bhatnagar--Gross--Krook (BGK) collision model, the Fokker--Planck operator is decomposed into a nonstiff drift--diffusion contribution and a stiff thermalization contribution. The former is retained in the particle dynamics through a modified Ornstein--Uhlenbeck process, whereas the latter is represented by a BGK-type relaxation toward the local Maxwellian. This decomposition enables an adaptive wave--particle representation: the method follows stochastic particle dynamics in rarefied regimes and increasingly represents the rapidly equilibrating distribution by the analytical wave component as the continuum regime is approached. The modified friction coefficient is constructed to recover the original Fokker--Planck dynamics in the rarefied limit while preserving the hydrodynamic limit under strong collisions. Numerical experiments demonstrate that the proposed method captures velocity-space drift and diffusion, recovers the expected kinetic and continuum behavior across a range of Knudsen numbers, and reproduces the characteristic evolution of collisional plasma phenomena.
\end{abstract}


\end{frontmatter}

\section{Introduction}

Plasmas play a central role in physical phenomena across astrophysics, semiconductor processing, magnetic fusion, and aerospace engineering \cite{chen1984introduction}. A fundamental challenge in modeling these systems stems from their disparate spatiotemporal scales, among which the variation in collisionality is particularly important. This degree of collisionality is quantified by the Knudsen number, $\mathrm{Kn} = l_{mfp} / L$, defined as the ratio of the particle mean free path $l_{mfp}$ to the characteristic domain length $L$. In the strongly collisional regime ($\mathrm{Kn} \ll 1$), frequent collisions rapidly relax the velocity distribution function toward local thermodynamic equilibrium, validating a hydrodynamic fluid description. When \(L\) is comparable to or smaller than the mean free path, nonequilibrium effects become significant and a kinetic description is required. For plasmas dominated by long-range Coulomb interactions, the cumulative effect of small-angle scattering is commonly described by the Vlasov--Fokker--Planck (VFP) equation \cite{chen1984introduction}.  In complex applications, such as magnetic reconnection \cite{ji2022magnetic} and near-space plasma flows \cite{pu2026electromagnetic}, fluid and kinetic regimes frequently coexist within a single computational domain. Accurately and efficiently simulating these systems necessitates multiscale numerical algorithms capable of bridging disparate collisionality regimes.

Multiscale methods for collisional plasma dynamics have been developed using both Bhatnagar--Gross--Krook (BGK) relaxation models and Fokker--Planck (FP) collision operators. BGK-type models employ a simple relaxation toward a local equilibrium, which facilitates the construction and analysis of multiscale schemes. FP operators, by contrast, retain the velocity-space drift and diffusion associated with cumulative small-angle Coulomb scattering. These two classes of collision models have led to different strategies for constructing multiscale kinetic methods.

Among multiscale methods based on BGK-type collision models, Liu et al. developed the unified gas-kinetic scheme (UGKS), in which particle transport and collisional relaxation are coupled through a time-dependent interface distribution function \cite{xu2010unified,liu2016unified, liu2017unified}. This coupling enables the numerical flux to recover kinetic transport in rarefied regimes and hydrodynamic behavior in the continuum limit without resolving the collision time explicitly. The framework was subsequently extended to the unified gas-kinetic wave--particle (UGKWP) method \cite{liu2020unified,liu2021unified,pu2025unified,guo2026unified}. UGKWP decomposes the local distribution into an equilibrium wave component and a collisionless particle component according to the local ratio of the time step to the collision time. By representing the near-equilibrium contribution analytically and tracking only the nonequilibrium contribution with stochastic particles, UGKWP reduces the computational and memory costs in strongly collisional regimes while retaining a particle description in rarefied regimes. Related BGK-based multiscale formulations include the discrete unified gas-kinetic scheme (DUGKS) developed by Liu et al. \cite{liu2026discrete,liu2020discrete}, the central--upwind kinetic scheme of Xiao et al. \cite{xiao2021stochastic} and the micro--macro and implicit--explicit methods developed by Crouseilles et al. for collisional Vlasov systems \cite{crestetto2012kinetic,crouseilles2016multiscale}.

Multi-scale methods based on FP collision operators have been developed along two main directions: deterministic asymptotic-preserving (AP) discretizations for plasma kinetic equations and stochastic particle methods primarily designed for rarefied gas dynamics. In the first category, Jin and Yan proposed a penalization-based AP scheme for the nonlinear Fokker--Planck--Landau (FPL) equation \cite{jin2011class}. Building on this strategy, Gamba et al. extended the micro--macro framework to kinetic equations with the FPL collision operator \cite{gamba2019micro}. For the Vlasov--Poisson--Fokker--Planck (VPFP) system, Carrillo et al. developed a variational AP method within an implicit--explicit framework \cite{carrillo2021variational}, while Blaustein and Filbet proposed a structure- and asymptotic-preserving scheme \cite{blaustein2024structure}. More recently, Crouseilles et al. extended the UGKS framework to linear kinetic equations with FP collisions in the diffusive limit \cite{crouseilles2025generalized}. In the second category, stochastic FP particle methods have been studied in rarefied gas dynamics as computationally efficient alternatives to direct simulation Monte Carlo (DSMC) in near-continuum regimes. Gorji and Jenny developed a hybrid FP-–DSMC algorithm that combines the efficiency of continuous stochastic velocity evolution with the accuracy of DSMC in rarefied regimes \cite{gorji2015fokker}. Multiscale stochastic particle methods based directly on FP kinetic models were subsequently proposed for nonequilibrium gas flows by Fei et al. and Zhang et al.\cite{fei2017particle,cui2025multiscale}.

The UGKWP framework has been successfully applied to multiscale plasma simulations with BGK-type collision models, demonstrating good performance in both benchmark problems and practical applications \cite{liu2021unified,pu2025unified,pu2026electromagnetic,pu2026nonequilibrium}. Nevertheless, the BGK operator represents collisions through a single-rate relaxation toward a local Maxwellian and therefore does not explicitly capture the velocity-space drift and diffusion behaviors. This simplified relaxation may be insufficient for plasma phenomena governed by localized gradients and fine structures in velocity space \cite{abel2008linearized}. In the present work, we extend the UGKWP framework to a FP collision model, denoted by UGKWP-FP, to provide a more physically representative description of collisional plasma dynamics. Specifically, the Lenard--Bernstein (LB) operator is adopted as a tractable FP model that retains the essential velocity-space drift and diffusion structure. Within the resulting UGKWP-FP method, the stochastic particle component evolves according to Langevin representation of the LB operator, while the rapidly equilibrating component is represented by the hydrodynamic wave formulation inherited from UGKWP. A BGK-type relaxation is introduced only to determine the wave--particle decomposition. This construction extends the multiscale capability of UGKWP to collisional plasma dynamics with explicit velocity-space drift and diffusion.

The paper is organized as follows. Section \ref{sec:model} introduces the VPFP kinetic model and its nondimensional form. Section \ref{sec:fpbgk} presents the FP-BGK model. Section \ref{sec:ugkwp} gives the UGKWP discretization for the decomposed model. Section \ref{sec:tests} shows numerical tests. Conclusions are given in Section \ref{sec:conclusions}.

\section{VPFP kinetic model}
\label{sec:model}

\subsection{Model equations}

The VPFP equation for an electrostatic collisional plasma with a fixed ion background is written as
\begin{align}
    \frac{\partial f}{\partial t}
    + \boldsymbol{u}\cdot \nabla_{\boldsymbol{x}} f
    + \boldsymbol{a}\cdot \nabla_{\boldsymbol{u}} f
    &=
    \mathcal{Q}_{FP}(f),
    \label{eq:dimensional-fp} \\
    \nabla_{\boldsymbol{x}}^2 \phi
    &=
    -\frac{\rho_c}{\epsilon_0},
    \label{eq:poisson}
\end{align}
where $f=f(\boldsymbol{x},\boldsymbol{u},t)$ is the velocity distribution function, \(t\) is time, \(\boldsymbol{x}\in \mathbb{R}^{d_x}\) is the physical space coordinate, and \(\boldsymbol{u}\in \mathbb{R}^{d_v}\) is the microscopic particle velocity. $d_x, d_v$ denote the dimension of physical and velocity space. \(\boldsymbol{a}={q\boldsymbol{E}}/{m}\) is an external acceleration. The electric field $\boldsymbol{E}$ is derived from the electric potential $\phi$ via $\boldsymbol{E} = -\nabla_{\boldsymbol{x}} \phi$.  $\epsilon_0$ is the vacuum permittivity. $q=-e, e>0$ is the electric charge carried by an electron. The charge density $\rho_c$ is given by $\rho_c = e(n_i - n)$, where $n_i=n_0$ is the fixed ion background density, $n$ is the electron number density. \(Q_{FP}\) is chosen as the FP collision operator. In this work, a self-consistent LB-type FP collision operator is employed:
\begin{equation}
    \mathcal{Q}_{FP}(f)
    =
    \nu\nabla_{\boldsymbol{u}}\cdot
    \left[
        (\boldsymbol{u}-\boldsymbol{U})f
        +
        RT \nabla_{\boldsymbol{u}} f
    \right],
    \label{eq:fp-operator}
\end{equation}
where $\nu = p/2\mu$ is the relaxation rate, \(\mu\) is the dynamic viscosity, \(p\) is the thermodynamic pressure.  \(\boldsymbol{U}\) is the macroscopic velocity and \(R=k_B/m\) is the specific gas constant, \(k_B\) is the Boltzmann constant, \(m\) is the particle mass and $T$ is the temperature. They are defined from the moments of \(f\) as
\begin{equation}
    \rho = \int f \, \mathrm{d}\boldsymbol{u},
    \quad
    \boldsymbol{U} = \frac{1}{\rho}\int \boldsymbol{u} f \, \mathrm{d}\boldsymbol{u},
    \quad 
    T = \frac{1}{d_vR\rho}\int |\boldsymbol{u}-\boldsymbol{U}|^2 f \, \mathrm{d}\boldsymbol{u}.
    \label{eq:rho-U}
\end{equation}
Here \(\rho\) is mass density.  The local Maxwellian distribution corresponding to \((\rho,\boldsymbol{U},T)\) is
\begin{equation}
    g
    =
    \rho
    \left(\frac{\lambda}{\pi}\right)^{d_v/2}
    \exp\left(
        -\lambda|\boldsymbol{u}-\boldsymbol{U}|^2
    \right),
    \quad
    \lambda=\frac{1}{2RT}.
    \label{eq:maxwellian}
\end{equation}
The viscosity model in this work is given as
\[
\mu = \mu_{\text{ref}}\left( \frac{T}{T_{\text{ref}}} \right)^{\omega},
\]
where \(\omega=0.72\) denotes the viscosity index characterizing the temperature dependence. The reference viscosity $\mu_{\text{ref}}$ is evaluated as
\[
\mu_{\text{ref}} = \frac{5(\alpha_{\text{ref}} + 1)(\alpha_{\text{ref}} + 2)\sqrt{\pi}}{4\alpha_{\text{ref}}(5 - 2\omega_{\text{ref}})(7 - 2\omega_{\text{ref}})} \, \text{Kn}_{\text{ref}},
\]
with \(\alpha_{\text{ref}} = 1\) and \(\omega_{\text{ref}} = 0.5\). For notational convenience, we omit the subscript “ref” and denote \(\text{Kn}_{\text{ref}}\) simply as \(\text{Kn}\) in the subsequent sections.

Eq.\eqref{eq:dimensional-fp} without acceleration term admits the following
self-consistent Langevin representation in phase space \cite{risken1989fokker, fei2017particle}:
\begin{equation}
    \mathrm{d}\boldsymbol{X}(t)
    =
    \boldsymbol{V}(t)\,\mathrm{d}t,
    \qquad
    \mathrm{d}\boldsymbol{V}(t)
    =
    -\nu
    \left[
        \boldsymbol{V}(t)
        -
        \boldsymbol{U}(\boldsymbol{X},t)
    \right]\mathrm{d}t
    +
    \sqrt{
        2\nu\,RT(\boldsymbol{X},t)
    }\,
    \mathrm{d}\boldsymbol{W}(t),
    \label{eq:langevin-dimensional}
\end{equation}
where $\boldsymbol{X}(t)\in\mathbb{R}^{d_x}$ and
$\boldsymbol{V}(t)\in\mathbb{R}^{d_v}$ denote, respectively, the random
position and velocity of a particle at continuous time $t$. The local
bulk velocity $\boldsymbol{U}(\boldsymbol{X},t)$ and thermal energy
$RT(\boldsymbol{X},t)$ are determined self-consistently from the local
moments of the distribution function $f$. Here, $\boldsymbol{W}(t) \in \mathbb{R}^{d_v}$ is a $d_v$-dimensional standard Wiener process with $d_v$ independent components. Its increment $\mathrm{d}\boldsymbol{W}(t) = \boldsymbol{W}(t+\mathrm{d}t) - \boldsymbol{W}(t) \sim \mathcal{N}(\boldsymbol{0}, \mathrm{d}t \mathbb{I})$ satisfies $\langle \mathrm{d}W_i(t) \rangle = 0$ and $\langle \mathrm{d}W_i(t) \mathrm{d}W_j(t) \rangle = \mathrm{d}t \delta_{ij}$, where $i, j \in \{1, \dots, d_v\}$ denote the component indices and $\langle\cdots\rangle$ denotes an ensemble average.
When $\nu$, \(\boldsymbol{U}\) and $T$ are  held constant during a time interval, the velocity equation reduces to an
Ornstein--Uhlenbeck (OU) process and admits the exact update
\begin{equation}
    \boldsymbol{V}(t)
    =
    \boldsymbol{U}
    +
    e^{-\nu t}
    \left(\boldsymbol{V}_0-\boldsymbol{U}\right)
    +
    \sqrt{RT\left(1-e^{-2\nu t}\right)}
    \boldsymbol{\xi},
    \label{eq:ou-velocity-exact}
\end{equation}
where \(\boldsymbol{\xi}\sim \mathcal{N}(\boldsymbol{0},\mathbb{I})\).

\subsection{Nondimensional form and asymptotic behavior}

Let \(l_0\), $t_0$, \(u_0\), \(\rho_0\), $n_0$ and \(T_0\) be the reference length, time, velocity, mass density, number density and temperature. The non-dimensional variables are defined by
\begin{equation}
    \hat{\boldsymbol{x}}=\frac{\boldsymbol{x}}{l_0},
    \quad
    \hat{t}=\frac{t}{t_0},
    \quad
    \hat{\boldsymbol{u}}=\frac{\boldsymbol{u}}{u_0},
    \quad
    \hat{f}=\frac{u_0^{d_v}}{\rho_0} f,
    \quad 
    \hat{n} = \frac{n} {n_0},
    \quad
    \hat{T}=\frac{T}{T_0},
    \quad 
    \hat{\phi} = \frac{\phi}{\phi_0}.
\end{equation}
The dimensionless quantities are related through
$$
u_0 = \sqrt{2RT_0}, \quad t_0=\frac{l_0}{u_0}, \quad n_0 = \frac{\rho_0}{m_0},\quad \phi_0 = \frac{m_0 u_0^2}{e}.
$$
Here electron mass is chosen as reference mass, so that $m_0 = m_e$, and the dimensionless relaxation rate and time are
\begin{equation}
    \hat{\nu}=\nu t_0,
    \quad
    \hat{\tau}=\frac{1}{\hat{\nu}}
    =
    \frac{1}{\nu t_0}.
\end{equation}
After dropping the hats, the nondimensional kinetic equation becomes
\begin{align}
    \frac{\partial f}{\partial t}
    +
    \boldsymbol{u}\cdot \nabla_{\boldsymbol{x}} f
    +
    \nabla_{\boldsymbol{x}}\phi\cdot \nabla_{\boldsymbol{u}} f
    &=
    \mathcal{Q}_{FP}(f),
    \label{eq:nondim-fp} \\
    \lambda_D^2\nabla_{\boldsymbol{x}}^2 \phi
    &=
    n-1,
    \label{eq:poisson}
\end{align}
where
\begin{equation}
    \mathcal{Q}_{FP}(f)
    =
    \nu \nabla_{\boldsymbol{u}}\cdot
    \left[
        (\boldsymbol{u}-\boldsymbol{U})f
        +
        \frac{T}{2} \nabla_{\boldsymbol{u}}f
    \right],
    \label{eq:nondim-fp-operator}
\end{equation}
and the normalized Debye length
\begin{align*}
\lambda_D=\sqrt{\frac{\epsilon_0 m_0u_0^2}{n_0 e^2}} \Big / l_0=\sqrt{\frac{2\epsilon_0 k_B T_0}{n_0 e^2}} \Big / l_0.
\end{align*}
The thermal pressure $p=\frac{1}{2} \rho T$. The electrostatic energy of the system, $E_{p}$, is defined as:
$$
{E}_{p}=\frac{\lambda_D^2}{2} \int_L\left|\nabla_{\boldsymbol{x}}\phi\right|^2 d x,
$$
where $L$ denotes the whole computational domain.

When $\nu \ll 1$, i.e. $\tau\gg 1$, collisions are weak over the characteristic time scale.
The distribution function can remain far from the local Maxwellian, so
the nonequilibrium kinetic evolution must be resolved. When
$\tau=O(1)$, particle transport and collisional
relaxation are all important, corresponding to a transitional regime. 
When $\nu\rightarrow \infty$, i.e. \(\tau\to0\), the collision operator
becomes stiff and rapidly drives the distribution function toward the
local Maxwellian. The system then approaches a near-equilibrium
hydrodynamic regime. A conventional kinetic scheme based on operator splitting between transport and collision generally requires a time step
\(\Delta t < \tau\), and a mesh size $\Delta x < l_{mfp}$. These requirements become prohibitively expensive
as \(\tau\to0\). Therefore, an asymptotic-preserving numerical method is
needed: it remains accurate and stable with macroscopic time steps and mesh sizes
independent of the small relaxation time.

\section{FP-BGK model}
\label{sec:fpbgk}

To construct an asymptotic UGKWP method for the VPFP equation, the FP operator is decomposed into a non-stiff particle-resolved FP component and an unresolved stiff remainder. The non-stiff part is solved by the OU process. The stiff remainder is modeled by a BGK-type relaxation operator. The FP relaxation term is re-written as
\begin{equation}
    \mathcal{Q}_{FP}(f) = \nu \mathcal{L}_{FP}(f), \quad \mathcal{L}_{FP}(f)=\nabla_{\boldsymbol{u}}\cdot
    \left[
        (\boldsymbol{u}-\boldsymbol{U})f
        +
        \frac{T}{2} \nabla_{\boldsymbol{u}}f
    \right].
\end{equation}
We introduce a modified particle relaxation rate \(\nu_p\), satisfying
\begin{equation}
    0\leq \nu_p \leq \nu.
\end{equation}
Then
\begin{equation}
    \nu \mathcal{L}_{FP}(f)
    =
    \nu_p \mathcal{L}_{FP}(f)
    +
    (\nu-\nu_p)\mathcal{L}_{FP}(f).
    \label{eq:fp-splitting}
\end{equation}
The first term is retained as a FP operator. The second term is a stiff thermalization remainder and is approximated by a BGK-type relaxation,
\begin{equation}
    (\nu-\nu_p)\mathcal{L}_{FP}(f)
    \approx
    \frac{g-f}{\tau_{B}}.
    \label{eq:stiff-fp-bgk}
\end{equation}
Here \(\tau_{B}\) is the relaxation time of the BGK closure. A choice of $\tau_B$ is
\begin{equation}
    \tau_{B}=\frac{C}{\nu-\nu_p}.
    \label{eq:taur-simple}
\end{equation}
The factor \(C=1/2\) follows from viscosity matching (see more details in \ref{app:viscosity-matching}). If the Prandtl number or other transport coefficients need to be adjusted, the BGK closure can be replaced by an ES-BGK or Shakhov-type closure.

With this approximation, Eq. \eqref{eq:dimensional-fp} becomes
\begin{equation}
    \frac{\partial f}{\partial t}
    +
    \boldsymbol{u}\cdot\nabla_{\boldsymbol{x}} f
    +
    \boldsymbol{a}\cdot\nabla_{\boldsymbol{u}} f
    =
    \nu_p \mathcal{L}_{FP}(f)
    +
    \frac{g-f}{\tau_{B}}.
    \label{eq:fp-bgk-model-rhs}
\end{equation}
Equivalently,
\begin{equation}
    \frac{\partial f}{\partial t}
    +
    \boldsymbol{u}\cdot\nabla_{\boldsymbol{x}} f
    +
    \boldsymbol{a}\cdot\nabla_{\boldsymbol{u}} f
    -
    \nu_p \mathcal{L}_{FP}(f)
    =
    \frac{g-f}{\tau_{B}}.
    \label{eq:fp-bgk-model}
\end{equation}
This equation is referred to as the FP-BGK model in this work.


The decomposition is introduced at the numerical level: $\nu_p$ depends on the numerical time step $\Delta t$. Specifically, $\nu_p$ should satisfy two requirements:
\begin{enumerate}
    \item In the rarefied regime, \(\nu_p\) should approach \(\nu\), so that the original FP dynamics is recovered.
    \item In the continuum regime, \(\nu_p\Delta t\) is required to remain uniformly bounded, i.e., \(\nu_p\Delta t=O(1)\). Although the exact OU velocity update is unconditionally stable, the particle transport and collision processes are treated separately. When \(\nu_p\Delta t\gg1\), the particle velocity changes substantially within a single transport step, invalidating the free-transport approximation based on a constant particle velocity and introducing a large splitting error in the particle trajectory and flux. Bounding \(\nu_p\Delta t\) keeps the particle-resolved transport and collision time scales comparable, while the faster unresolved relaxation is transferred to the wave component under the BGK framework.
\end{enumerate}

One choice is
\begin{equation}
    \nu_p
    =
    \min\left(
        \nu,
        \frac{C_p}{\Delta t}
    \right),
    \label{eq:nup-min}
\end{equation}
where \(C_p=O(1)\) is a prescribed parameter, in this work $C_P = 1$. This gives
\begin{equation}
    \nu_p \approx \nu,
    \qquad
    \nu\Delta t\ll 1,
\end{equation}
and
\begin{equation}
    \nu_p \approx \frac{C_p}{\Delta t},
    \qquad
    \nu\Delta t\gg 1.
\end{equation}
The two limiting regimes follow directly from Eq. \eqref{eq:nup-min}. In the rarefied regime,
\begin{equation}
    \nu\Delta t\ll 1.
\end{equation}
Then \(\nu_p=\nu\), \(\nu-\nu_p=0\), and \(\tau_{B}=\infty\). The decomposed model reduces to
\begin{equation}
    \partial_t f
    +
    \boldsymbol{u}\cdot\nabla_{\boldsymbol{x}} f
    +
    \boldsymbol{a}\cdot\nabla_{\boldsymbol{u}} f
    =
    \nu \mathcal{L}_{FP}(f).
\end{equation}
Thus the original FP kinetic equation is recovered.
In the continuum regime,
\begin{equation}
    \nu\Delta t\gg 1.
\end{equation}
Then \(\nu_p\Delta t=C_p \sim O(1)\), while
\begin{equation}
    \nu-\nu_p \approx \nu.
\end{equation}
The stiff relaxation is mainly carried by the BGK closure. Therefore, the distribution function is rapidly driven towards the local Maxwellian by the wave part.

\section{UGKWP method for the FP-BGK model}
\label{sec:ugkwp}

\subsection{Overview of algorithm}

The evolution of the multiscale flow with macroscopic external forces and FP diffusion is governed by the FP-BGK model now:
\begin{equation}
    \frac{\partial f}{\partial t}
    +
    \boldsymbol{u}\cdot\nabla_{\boldsymbol{x}} f
    +
    \boldsymbol{a}\cdot\nabla_{\boldsymbol{u}} f
    -
    \nu_p \mathcal{L}_{\mathrm{FP}}(f)
    =
    \frac{g-f}{\tau_{B}}.
    \label{eq:fp-bgk-model2}
\end{equation}

Within the UGKWP framework, the BGK relaxation term governs wave–-particle decomposition, including particle thermalization and resampling, whereas the FP and force terms determine the velocity-space friction, diffusion, and acceleration of the active particles.

To establish the wave-particle decomposition, we first consider the characteristic integral solution driven by the BGK relaxation process over a time interval $[0, t]$:
\begin{equation}
    f\left(\boldsymbol{x},\boldsymbol{u}, t\right) = \frac{1}{\tau_B} \int_{0}^{t} g(\boldsymbol{x}^{\prime}, \boldsymbol{u}, t^{\prime}) e^{-\left(t-t^{\prime}\right) / \tau_B} \mathrm{d} t^{\prime} + e^{-t / \tau_B} f_{0}\left(\boldsymbol{x}-\boldsymbol{u} t \right),
    \label{eq:BGKsoln}
\end{equation}
where $f_0$ is the initial distribution at $t=0$, and $g$ is the local equilibrium along the characteristic line $\boldsymbol{x}^{\prime} = \boldsymbol{x} - \boldsymbol{u}(t- t^{\prime})$. Expanding the equilibrium state $g$ via a first-order Taylor series around $(\boldsymbol{x}, t)$,
\begin{equation}
g^{\prime} = g + \nabla_{\boldsymbol{x}} g \cdot(\boldsymbol{x}^{\prime}-\boldsymbol{x}) + \partial_t g(t^{\prime}-t),
\label{eq:taylor g}
\end{equation}
and substituting Eq. \eqref{eq:taylor g} into Eq. \eqref{eq:BGKsoln} yields the multiscale analytic distribution function:
\begin{equation}
f(\boldsymbol{x}, \boldsymbol{u}, t) = \left(1-e^{-t / \tau_B}\right) g^{+}(\boldsymbol{x}, \boldsymbol{u}, t) + e^{-t / \tau_B} f_{0}\left(\boldsymbol{x} - \boldsymbol{u}t\right),
\label{eq:multiscale BGK soln}
\end{equation}
with the modified equilibrium $g^{+}$ expressed as:
\[
g^{+}\left( \boldsymbol{x},\boldsymbol{u}, t \right) = g\left( \boldsymbol{x},\boldsymbol{u} ,t\right) + \left( \frac{te^{- t\text{/}\tau_B}}{1 - e^{- t\text{/}\tau_B}} - \tau_B \right) \boldsymbol{u}\cdot \nabla_{\boldsymbol{x}} g\left( \boldsymbol{x},\boldsymbol{u},t \right) + \left( \frac{t}{1 - e^{- t\text{/}\tau_B}} - \tau_B \right) \partial_t g\left( \boldsymbol{x},\boldsymbol{u},t \right).
\]

Equation \eqref{eq:multiscale BGK soln} reveals the wave-particle duality: a simulation particle has a probability of $e^{-t/\tau_B}$ to undergo collisionless free transport, and a probability of $(1 - e^{-t/\tau_B})$ to collide and thermalize into the equilibrium wave component $g^+$. Consequently, the cumulative distribution function for a particle's free transport time $t_{f}$ is $F(t_{f} < t) = 1 - \exp(-t/\tau_B)$, from which $t_{f}$ is sampled via
\begin{equation}
    t_{f} = - \tau_B \ln(\eta), \quad \eta \sim \mathcal{U}(0,1).
\end{equation}

For a simulation time step $\Delta t$, particles with $t_{f} \ge \Delta t$ survive as simulation particles. Meanwhile, particles with $t_{f} < \Delta t$ are thermalized into the macroscopic wave field $g^+$.

Based on the above physical picture, for each time step \(t^n\rightarrow t^{n+1}\), the UGKWP-FP algorithm is summarized as follows:
\begin{enumerate}[label=\textbf{Step \arabic*:}, leftmargin=*]
    \item \textbf{Parameter initialization.} Evaluate the original FP relaxation rate $\nu_i$, the modified particle relaxation rate $\nu_{p,i}$, and the BGK relaxation time $\tau_{B,i}$ in cell $\Omega_i$.

    \item \textbf{First-stage velocity-space evolution ($\Delta t/2$).} The velocity-space evolution is performed in two successive steps.
First, both the particles and the wave component are advanced
under the electrostatic field over a half time
step as introduced in Section \ref{sec: acc electric}. Second, the
particle-resolved FP dynamics is applied only to
particles through the OU process with collision
frequency $\nu_{p,i}$ over a half time step shown in Section \ref{sec: ou process}.

    \item \textbf{Microscopic particle transport and flux accumulation.} For each simulation particle $P_k$ in cell $\Omega_i$, sample the BGK free-transport time $t_{f,k} $. Stream the particle over $\Delta t_k = \min(t_{f,k}, \Delta t)$ and update its spatial position.
    \begin{itemize}
        \item If $t_{f,k} < \Delta t$, absorb the particle into the wave component after duration $t_{f,k}$; otherwise, retain it in the particle component for the next time level.
        \item Accumulate the free-streaming particle flux $(\mathscr{F}_{\boldsymbol{W}})^{f,p}$ across cell interfaces during trajectory tracking.
    \end{itemize}

    \item \textbf{Macroscopic update and particle resampling.}  Compute the analytic wave fluxes $(\mathscr{F}_{\boldsymbol{W}})_s^g$ and $(\mathscr{F}_{\boldsymbol{W}})_s^{f,w}$ across cell interfaces. Update the cell-averaged conservative variables $\boldsymbol{W}_i$ via the finite-volume scheme in Eq. \eqref{eq: macro upadate rule}. Compute the conservative quantities carried by surviving particles, denoted by $\boldsymbol{W}_i^p$, and assign the remaining conservative quantities to the wave component:
    \[
        \boldsymbol{W}_i^h = \boldsymbol{W}_i - \boldsymbol{W}_i^p.
    \]
    Sample the fraction $e^{-\Delta t/\tau_{B,i}}$ of the wave component as new collisionless particles for the next time step, while keeping the remaining fraction as the analytic wave distribution.

    \item \textbf{Electrostatic field update.} Calculate the new spatial charge density distribution and solve the Poisson equation to update the self-consistent electric field $\boldsymbol{E}^{n+1}$.

    \item \textbf{Second-stage velocity-space evolution ($\Delta t/2$).} Apply the remaining half-step electrostatic acceleration and OU process over $\Delta t/2$ using the updated electric field $\boldsymbol{E}^{n+1}$ and the OU process with frequency $\nu_{p,i}$.
\end{enumerate}

The timestep is determined by a CFL condition
including,
$$
\Delta t =
\mathrm{CFL}
{
\frac{\Delta x}{|u|_{max}}
},
$$
where \(u_{\max}\) is the maximum
resolved particle velocity, \(\Delta x\)
is the cell size. 

\subsection{Acceleration by electric field}
\label{sec: acc electric}

Within the operator-splitting framework, the electrostatic acceleration is executed over a half time step $\delta t = \Delta t / 2$. For each discrete simulation particle $P_k$, the velocity update governed by Newton's second law over the half time step is given by:
\begin{equation}
    \boldsymbol{v}^*_k = \boldsymbol{v}_k + \boldsymbol{a}_i \left(\frac{\Delta t}{2}\right),
    \label{eq:particle-acc-update}
\end{equation}
where $\boldsymbol{a}_i = \nabla_{\boldsymbol{x}}\phi_i$ according to Eq.\eqref{eq:nondim-fp} is the local electrostatic acceleration evaluated in cell $\Omega_i$ containing particle $P_k$. Concurrently, the macroscopic bulk velocity of the wave component $\boldsymbol{U}^h_i$ undergoes an analogous half-step shift:
\begin{equation}
    \boldsymbol{U}_i^{h,*} = \boldsymbol{U}_i^h + \boldsymbol{a}_i \left(\frac{\Delta t}{2}\right).
    \label{eq:wave-acc-update}
\end{equation}
Accordingly, the momentum and total energy densities of the wave component are updated consistently with this velocity shift.

After both the particle and the wave component have been updated under the acceleration operator, the overall cell-averaged macroscopic conservative variables $\boldsymbol{W}_i = (\rho_i, \rho_i \boldsymbol{U}_i, E_i)^T$ are re-accumulated by summing the updated particle and wave contributions:
\begin{equation}
    \boldsymbol{W}_i^* = \boldsymbol{W}_i^{p,*} + \boldsymbol{W}_i^{h,*}.
    \label{eq:total-macro-acc-update}
\end{equation}

\subsection{Particle evolution by the Ornstein-Uhlenbeck process}
\label{sec: ou process}

In the numerical implementation, the continuous stochastic process in Eq. \eqref{eq:langevin-dimensional} is discretized for each simulation particle $P_k$. Within the operator-splitting framework, the micro-relaxation steps are executed in two half-step stages over duration $\delta t = \Delta t / 2$.

During each half-step velocity update over $\delta t = \Delta t / 2$, according to the exact OU process solution in Eq. \eqref{eq:ou-velocity-exact}, the updated particle velocity $\boldsymbol{v}_k^*$ is evaluated as:
\begin{align}
    \boldsymbol{v}_k^* = \boldsymbol{U}_i + e^{-\nu_{p,i} \frac{\Delta t}{2}} \left( \boldsymbol{v}_k - \boldsymbol{U}_i \right) + \sqrt{\frac{T_i}{2} \left( 1 - e^{-2\nu_{p,i} \frac{\Delta t}{2}} \right)} \, \boldsymbol{\xi}_k, \label{eq:particle-v-update}
\end{align}
where $\boldsymbol{\xi}_k \sim \mathcal{N}(\boldsymbol{0}, \mathbb{I})$ is a vector of independent standard normal random variables sampled for particle $P_k$, $\boldsymbol{v}_k$ represents the velocity prior to the sub-step, and $\nu_{p,i}$ is the cell-averaged particle friction coefficient.

It is worth noting that the stochastic velocity update in Eq. \eqref{eq:particle-v-update} does not inherently conserve the macroscopic momentum and energy of the particle part within an individual cell $\Omega_i$. Because the OU process models a continuous stochastic relaxation toward the local equilibrium state, the cumulative momentum and kinetic energy carried by the discrete particles undergo stochastic fluctuations and systematic shifts during this thermalization step.

To guarantee exact macroscopic conservation laws for the coupled system, the total conservative quantities carried by the updated particle component $\boldsymbol{W}_i^{p*}$, are re-accumulated from the particle distribution following the OU step. The conservative quantities allocated to the hydrodynamic wave component $\boldsymbol{W}_i^h$ are subsequently determined via a balance relation:
\begin{equation}
    \boldsymbol{W}_i^{h *}= \boldsymbol{W}_i - \boldsymbol{W}_i^{p *}.
    \label{eq:wave-conservative-balance}
\end{equation}
By directly absorbing microscopic momentum and energy variations into the analytical wave distribution, this re-allocation mechanism strictly enforces the exact conservation of total mass, momentum, and energy.

\subsection{Analytic wave fluxes evaluation and macroscopic variables update}
As for the evolution of macroscopic quantities, the discretized evolution equation of macroscopic variables is
\begin{equation}
    \boldsymbol{W}_{i}^{*} = \boldsymbol{W}_{i}^n - \sum_{s}^{}{\frac{\Delta t}{\left| \Omega_{i} \right|}\left| l_{s} \right|(\mathscr{F}_{\boldsymbol{W}})_s^{g}} - \sum_{s}^{}{\frac{\Delta t}{\left| \Omega_{i} \right|}\left| l_{s} \right|(\mathscr{F}_{\boldsymbol{W}})_s^{f,w}} + \frac{1}{\left| \Omega_{i} \right|}(\mathscr{F}_{\boldsymbol{W}})^{f,p}  ,
    \label{eq: macro upadate rule}
\end{equation}
where $(\mathscr{F}_{\boldsymbol{W}})_s^g$ is the equilibrium flux, $(\mathscr{F}_{\boldsymbol{W}})_s^{f,w}$ and $(\mathscr{F}_{\boldsymbol{W}})^{f,p}$ are the free transport flux contributed by wave and particles.

The numerical flux of the macroscopic conservative variable can be decomposed into the equilibrium and free-streaming fluxes according to Eq.~\eqref{eq:BGKsoln}. The equilibrium flux is
\begin{equation}
(\mathscr{F}_{\boldsymbol{W}})_s^g = \frac{1}{\Delta t}\int_{t^n}^{t^{n+1}}\boldsymbol{u}\cdot \boldsymbol{n}_s
\left[\frac{1}{\tau_B} \int_{0}^{t} g(\boldsymbol{x}^{\prime}, \boldsymbol{u},t^{'}) e^{-\left(t-t^{\prime}\right) / \tau_B} \mathrm{d} t^{\prime}\right]
\boldsymbol{\Psi} \mathrm{d}\boldsymbol{u}\mathrm{d}t,
\label{eq:macroscopic eq flux}
\end{equation}
where $\boldsymbol{\Psi}=\left(1,\boldsymbol{u}, \frac{1}{2}|\boldsymbol{u}|^2\right)^T$. The free streaming flux is
\begin{equation}
(\mathscr{F}_{\boldsymbol{W}})_s^f = \frac{1}{\Delta t}\int_{t^n}^{t^{n+1}}\boldsymbol{u}\cdot \boldsymbol{n}_s
\left[e^{-t / \tau_B} f_{0}\left(\boldsymbol{x}-\boldsymbol{u} t \right)\right]
\boldsymbol{\Psi} \mathrm{d}\boldsymbol{u}\mathrm{d}t.
\label{eq:macroscopic fr flux}
\end{equation}
The equilibrium flux can be calculated as
\begin{equation}
(\mathscr{F}_{\boldsymbol{W}})_s^g = \frac{1}{\Delta t} \int\boldsymbol{u} \cdot \boldsymbol{n}_{s}\left( q_{1}g_{0} + q_{2}\boldsymbol{u} \cdot \nabla{g}_{0} + q_{3}\partial_t g_{0} \right)\boldsymbol{\Psi} d\boldsymbol{u}.
\label{eq:eqflux-numerical}
\end{equation}

The analytical free-streaming flux $( \mathscr{F}_{\boldsymbol{W}})_s^{f,w}$ associated with the wave representation of the free-streaming part is given as
\begin{equation}
	( \mathscr{F}_{\boldsymbol{W}})_s^{f,w} = \frac{1}{\Delta t} \int\boldsymbol{u} \cdot \boldsymbol{n}_s\left( (q_{4}-\Delta te^{- \Delta t\text{/}\tau_B})g_{0}^{h} + (q_{5}+ \frac{\Delta t^{2}}{2}e^{- \Delta t\text{/}\tau_B})\boldsymbol{u} \cdot \nabla {g}_{0}^{h} \right)\boldsymbol{\Psi}d\boldsymbol{u},
	\label{eq:fw}
\end{equation}
where
$$
\begin{aligned}
q_1 &= \Delta t - \tau_B\left(1 - e^{-\Delta t/\tau_B}\right), \\
q_2 &= 2\tau_B^2\left(1 - e^{-\Delta t/\tau_B}\right) - \tau_B\Delta t - \tau_B\Delta t e^{-\Delta t/\tau_B}, \\
q_3 &= \frac{\Delta t^2}{2} - \tau_B\Delta t + \tau_B^2\left(1 - e^{-\Delta t/\tau_B}\right), \\
q_4 &= \tau_B\left(1 - e^{-\Delta t/\tau_B}\right), \\
q_5 &= \tau_B\Delta t e^{-\Delta t/\tau_B} - \tau_B^2\left(1 - e^{-\Delta t/\tau_B}\right).
\end{aligned}
$$ \(g_0^h\) denotes the equilibrium distribution function of the wave component $\boldsymbol{W}^h$. $g_{0\boldsymbol{x}}^h$ is the spatial slope of \(g_0^h\). 
The net particle contribution $(\mathscr{F}_{\boldsymbol{W}})^{f,p}$ is computed as
$$
(\mathscr{F}_{\boldsymbol{W}})^{f,p}=  \sum_{k \in P_{\partial\Omega_{i}^{+}}}^{}\boldsymbol{W}_{P_{k}} - \sum_{k \in P_{\partial\Omega_{i}^{-}}}^{}\boldsymbol{W}_{P_{k}},
$$
where $\boldsymbol{W}_{P_{k}} = \left( m_{k},m_{k}\boldsymbol{v}_{k},\frac{1}{2}m_{k}|\boldsymbol{v}_{k}|^{2} \right)$, $P_{\partial\Omega_{i}^{-}}$ and \(P_{\partial\Omega_i^+}\) are the index set of the particles streaming out and in of cell $\Omega_{i}$ during a time step.

\section{Asymptotic behaviors}
\label{sec:asymptotic-behavior}

The asymptotic behavior of the present method is determined jointly by
the FP--BGK decomposition and the wave--particle evolution. In the
following analysis, the spatial mesh size $\Delta x$ and time step
$\Delta t$ are held fixed.

\begin{proposition}
With
\begin{equation}
    \nu_p=\min\left(\nu,\frac{C_p}{\Delta t}\right),
    \qquad
    \tau_B=\frac{1}{2(\nu-\nu_p)},
    \label{eq:asymptotic-parameters}
\end{equation}
the UGKWP-FP method possesses the following asymptotic properties:
\begin{enumerate}
    \item In the collisionless limit $\nu\rightarrow0$, the method
    degenerates to a particle method for the Vlasov--Poisson system.

    \item In the continuum limit $\nu\rightarrow\infty$, the particle
    contribution vanishes exponentially and the method approaches a
    deterministic gas-kinetic scheme for the Navier--Stokes--Poisson
    system. The corresponding time step and mesh size are not restricted
    by the original FP relaxation time $\tau$ or the mean free path.
\end{enumerate}
\end{proposition}

\begin{proof}
We first examine the limiting behavior of the FP--BGK decomposition and
then the corresponding wave--particle discretization.

\paragraph{1. Collisionless limit}

When $\nu\rightarrow0$, one has $\nu\Delta t\ll1$ and therefore
\begin{equation}
    \nu_p=\nu,
    \qquad
    \nu-\nu_p=0,
    \qquad
    \tau_B=\infty.
\end{equation}
Consequently, the BGK remainder disappears and the decomposed equation
reduces to the original VPFP equation,
\begin{equation}
    \frac{\partial f}{\partial t}
    +\boldsymbol{u}\cdot\nabla_{\boldsymbol{x}}f
    +\boldsymbol{a}\cdot\nabla_{\boldsymbol{u}}f
    =
    \nu\mathcal{L}_{FP}(f).
\end{equation}
Taking the further limit $\nu\rightarrow0$ gives
\begin{equation}
    \frac{\partial f}{\partial t}
    +\boldsymbol{u}\cdot\nabla_{\boldsymbol{x}}f
    +\boldsymbol{a}\cdot\nabla_{\boldsymbol{u}}f
    =0,
    \qquad
    \lambda_D^2\nabla_{\boldsymbol{x}}^2\phi=n-1,
\end{equation}
which is the Vlasov--Poisson system.

At the algorithmic level, the BGK free-transport time satisfies
\begin{equation}
    t_f=-\tau_B\ln\eta\longrightarrow\infty,
    \qquad \eta\in(0,1).
\end{equation}
Thus every particle survives over the entire numerical time step,
\begin{equation}
    \min(t_f,\Delta t)=\Delta t,
\end{equation}
and no particle is absorbed into the wave component.

Moreover, the OU update over each substep $\delta t={\Delta t}/{2}$ is
\begin{equation}
    \boldsymbol{v}^{\,*}
    =
    \boldsymbol{U}
    +e^{-\nu_p\delta t}
    (\boldsymbol{v}-\boldsymbol{U})
    +
    \sqrt{\frac{T}{2}
    \left(1-e^{-2\nu_p\delta t}\right)}
    \boldsymbol{\xi}.
\end{equation}
Since $\nu_p=\nu\rightarrow0$,
\begin{equation}
    e^{-\nu_p\delta t}\rightarrow1,
    \qquad
    1-e^{-2\nu_p\delta t}\rightarrow0,
\end{equation}
and hence
\begin{equation}
    \boldsymbol{v}^{\,*}\rightarrow\boldsymbol{v}.
\end{equation}
Therefore, the FP collision step does not modify the particle velocity
in this limit. The particles are affected only by collisionless spatial
transport and electrostatic acceleration, and the UGKWP-FP method
reduces to a particle solver for the Vlasov--Poisson system.

\paragraph{2. Continuum limit}

When $\nu\rightarrow\infty$ with fixed $\Delta t$, one has
$\nu\Delta t\gg1$, and
\begin{equation}
    \nu_p=\frac{C_p}{\Delta t},
    \qquad
    \nu_p\Delta t=C_p=O(1).
\end{equation}
The BGK relaxation time becomes
\begin{equation}
    \tau_B
    =
    \frac{1}{2(\nu-\nu_p)}
    =
    \frac{1}{2\nu}
    \left[
        1+O\left(\frac{\nu_p}{\nu}\right)
    \right]
    \rightarrow0.
    \label{eq:tauB-continuum}
\end{equation}
Here, the second equality follows from the series
expansion.\footnote{
Since $\nu_p/\nu\to0$,
$(\nu-\nu_p)^{-1}
=\nu^{-1}(1-\nu_p/\nu)^{-1}
=\nu^{-1}[1+\nu_p/\nu
+O((\nu_p/\nu)^2)]$.
}
Thus, the unresolved stiff relaxation is asymptotically dominated by
the BGK component, while the particle-resolved FP frequency remains
bounded. The decomposed kinetic equation can be written as
\begin{equation}
        \partial_t f
        +\boldsymbol{u}\cdot\nabla_{\boldsymbol{x}} f
        +\boldsymbol{a}\cdot\nabla_{\boldsymbol{u}} f
    -
    \nu_p\mathcal{L}_{FP}(f)
    =
    \frac{g-f}{\tau_B}.
    \label{eq:fp-bgk-continuum-analysis}
\end{equation}

For the wave--particle decomposition, the fraction of particles that
survives one complete time step is
\begin{equation}
    P(t_f\geq\Delta t)
    =
    e^{-\Delta t/\tau_B}.
\end{equation}
Using Eq.~\eqref{eq:tauB-continuum},
\begin{equation}
    e^{-\Delta t/\tau_B}
    =
    \exp\left[-2(\nu-\nu_p)\Delta t\right]
    \longrightarrow0.
\end{equation}
Accordingly, the sampled particle mass and the corresponding
free-streaming particle flux satisfy
\begin{equation}
    M^p
    =
    O\left(e^{-\Delta t/\tau_B}\right) \rightarrow 0,
    \qquad
    (\mathscr{F}_{\boldsymbol{W}})^{f,p}
    =
    O\left(e^{-\Delta t/\tau_B}\right) \rightarrow 0.
    \label{eq:particle-flux-vanish}
\end{equation}
The OU evolution acts only on these active particles, and therefore its
explicit particle contribution also vanishes exponentially in the
continuum limit.

The macroscopic evolution is consequently determined by the analytic
wave flux. Expanding the coefficients of the equilibrium and
wave-based free-streaming fluxes for $\tau_B\rightarrow0$ in Eq. \eqref{eq:eqflux-numerical} and \eqref{eq:fw} gives \cite{liu2020unified}
\begin{align}
    (\mathscr{F}_{\boldsymbol{W}})^{analytic}_s
    &=
    (\mathscr{F}_{\boldsymbol{W}})^g_s
    +
    (\mathscr{F}_{\boldsymbol{W}})^{f,w}_s
    \nonumber\\
    &=
    \int
    \boldsymbol{u}\cdot\boldsymbol{n}_s
    \left[
        g_{0}
        -
        \tau_B
        \left(
        \boldsymbol{u}\cdot\nabla_{\boldsymbol{x}}g_{0}
        +
        \partial_t g_{0}
        \right)
        +
        \frac{\Delta t}{2}\partial_t g_{0}
    \right]
    \boldsymbol{\Psi}\,
    \mathrm{d}\boldsymbol{u}
    \nonumber\\
    &\quad
    +
    O(\tau_B^2)
    +
    O\left(e^{-\Delta t/\tau_B}\right).
    \label{eq:continuum-analytic-flux}
\end{align}
The first term in Eq.~\eqref{eq:continuum-analytic-flux} gives the Euler
flux, the $O(\tau_B)$ term gives the first-order Chapman--Enskog
correction, and the term involving $\Delta t/2$ provides the
second-order temporal evolution of the interface equilibrium state.
Hence the analytic flux is consistent with the Navier--Stokes flux.

In particular, the viscosity associated with the limiting BGK wave
flux is
\begin{equation}
    \mu_{\mathrm{num}}
    =
    p\tau_B
    =
    \frac{p}{2(\nu-\nu_p)}
    =
    \frac{p}{2\nu}
    \left[
        1+O\left(\frac{\nu_p}{\nu}\right)
    \right].
\end{equation}
Since the original FP model uses $\nu=p/(2\mu)$,
\begin{equation}
    \mu_{\mathrm{num}}
    =
    \mu
    \left[
        1+O\left(\frac{\nu_p}{\nu}\right)
    \right]
    \longrightarrow\mu.
\end{equation}
Therefore, the limiting wave flux recovers the viscosity of the
original FP model.

Combining Eq.~\eqref{eq:particle-flux-vanish} with
Eq.~\eqref{eq:continuum-analytic-flux}, the macroscopic update
converges to a deterministic discretization of the
Navier--Stokes--Poisson equations.
Therefore, provided that the spatial reconstruction and the
force-splitting procedure are second-order accurate, the UGKWP-FP
method becomes a second-order gas-kinetic scheme for the
Navier--Stokes--Poisson system in the continuum regime. 
\end{proof}

\section{Numerical studies}
\label{sec:tests}

\subsection{Homogeneous relaxation}

This benchmark case serves to investigate the differences in velocity-space relaxation dynamics between the BGK and FP collision models. As a baseline test, it provides essential physical insights that aid in understanding the relaxation behavior in subsequent, more complex cases. The current simulation is conducted using Particle-in-Cell FP (PIC-FP) and Particle-in-Cell BGK (PIC-BGK) methods in the absence of external acceleration fields. Implementations of PIC-FP and PIC-BGK are given in \ref{sec: pic}.

The system is initialized with a spatially uniform, counter-streaming two-beam velocity distribution. The initial phase-space distribution function $f(x, \boldsymbol{v}, 0)$ is expressed as:
\begin{equation}
    f(x, \boldsymbol{u}, 0) = \rho_0 f_0(\boldsymbol{u}),
    \label{eq:init_dist}
\end{equation}
where $\rho_0=1$ represents the background mass density, and $f_{0}(\boldsymbol{u})$ denotes the normalized two-stream Maxwellian velocity distribution. The velocity distribution $f_{0}(\boldsymbol{u})$ is composed of two symmetric counter-propagating beams:
\begin{equation}
    f_0(\boldsymbol{u}) = \frac{1}{2 (2\pi u_{th}^2)^{3/2}} \left[ \exp\left( -\frac{(u_x + U_0)^2 + u_y^2 + u_z^2}{2u_{th}^2} \right) + \exp\left( -\frac{(u_x - U_0)^2 + u_y^2 + u_z^2}{2u_{th}^2} \right) \right],
\end{equation}
where $U_0 = 4.0$ is the beam drift speed, and $u_{th} = 0.5$ is the thermal velocity. The Knudsen number $\text{Kn}=0.1$, with $\tau\approx 0.1$. The particle number is 1000 in each cell. The time step $\Delta = 0.01$. The electric field is not considered in this case.

Fig.~\ref{fig:homo-relaxation} illustrates the temporal evolution of the normalized velocity distribution function of the whole domain $f(u)$ obtained by the PIC-FP and PIC-BGK methods, starting from the identical bimodal initial condition. Significant disparities in the relaxation pathways are observed between the two collision operators, owing to their different mathematical formulations.

At early relaxation stages (Fig.~\ref{fig:homo-b}), the FP operator demonstrates a noticeably faster attenuation of the peak amplitudes compared to the BGK operator. Since the FP operator describes continuous velocity-space drift and diffusion, the two initial Maxwellian beams undergo rapid diffusive broadening, causing the peaks to flatten smoothly while spreading into the intermediate velocity regime. Conversely, the BGK model decays more slowly at $u=\pm 4.0$ without significant diffusive broadening, and a distinct central elevation emerges around $u = 0$, directly reflecting the early contribution of the equilibrium target distribution.

This contrast becomes more prominent at intermediate times (Fig.~\ref{fig:homo-c}). The FP model evolves into a broad, unimodal-like continuous plateau bridging the two initial beams. In contrast, the BGK distribution exhibits a multi-peak profile, where the residual of the initial counter-streaming beams coexist with the growing Maxwellian core centered at $u = 0$. 

Eventually, as the system approaches thermal equilibrium at late times (Fig.~\ref{fig:homo-d}), the profiles predicted by both operators converge to the same standard Maxwellian distribution. These results highlight that while both BGK and FP operators correctly recover the identical late-time equilibrium state, their transient relaxation mechanisms differ fundamentally. Due to its second-order diffusive nature, the FP operator is highly sensitive to large velocity-space gradients, driving a rapid dissipation and smoothing of steep non-equilibrium structures. In contrast, the BGK operator relaxes the distribution function towards equilibrium at a uniform rate that is insensitive to local velocity-space gradients, thereby preserving such steep gradient features over a significantly longer duration during non-equilibrium transients.

\begin{figure}
    \centering
    \begin{subfigure}[b]{0.48\textwidth}
    \centering
    \includegraphics[width=1.0\linewidth]{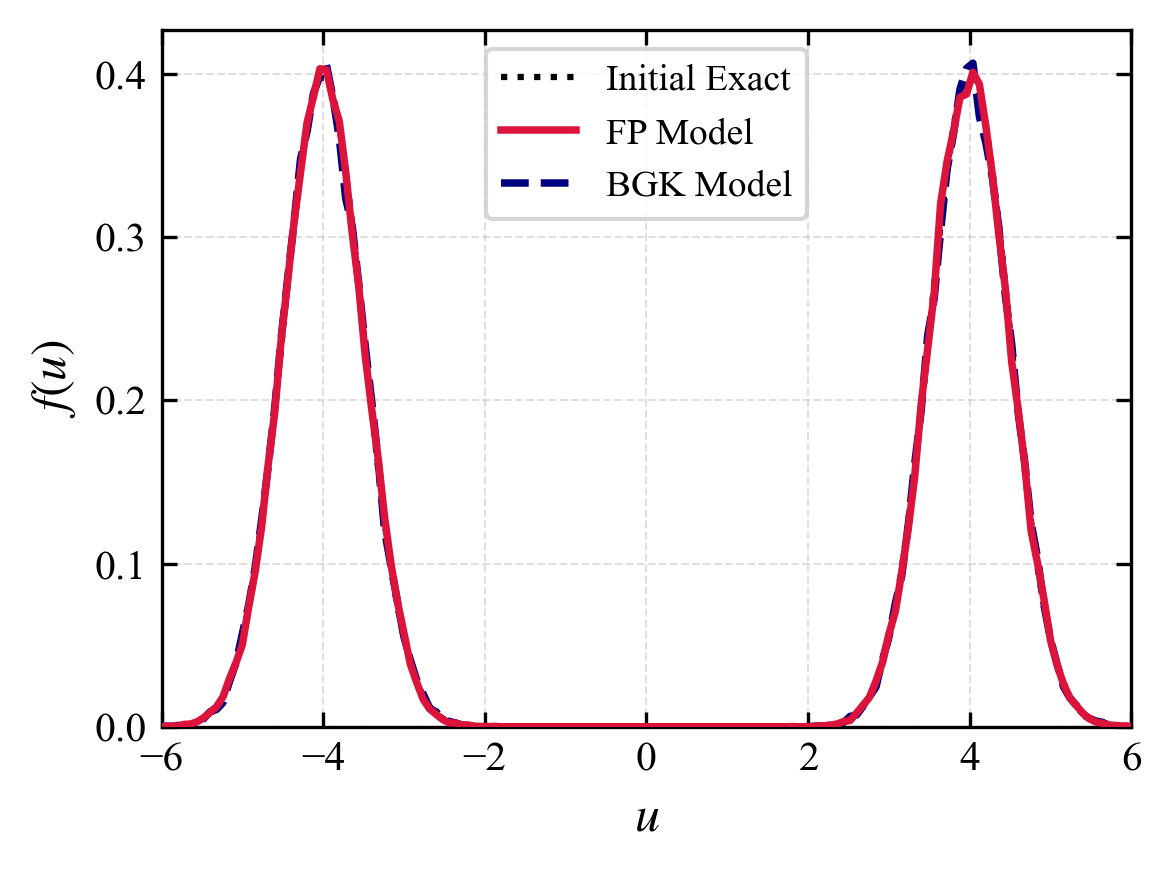}
    \caption{}
    \label{fig:homo-a}
    \end{subfigure}
    \hfill
    \begin{subfigure}[b]{0.48\textwidth}
    \centering
    \includegraphics[width=1.0\linewidth]{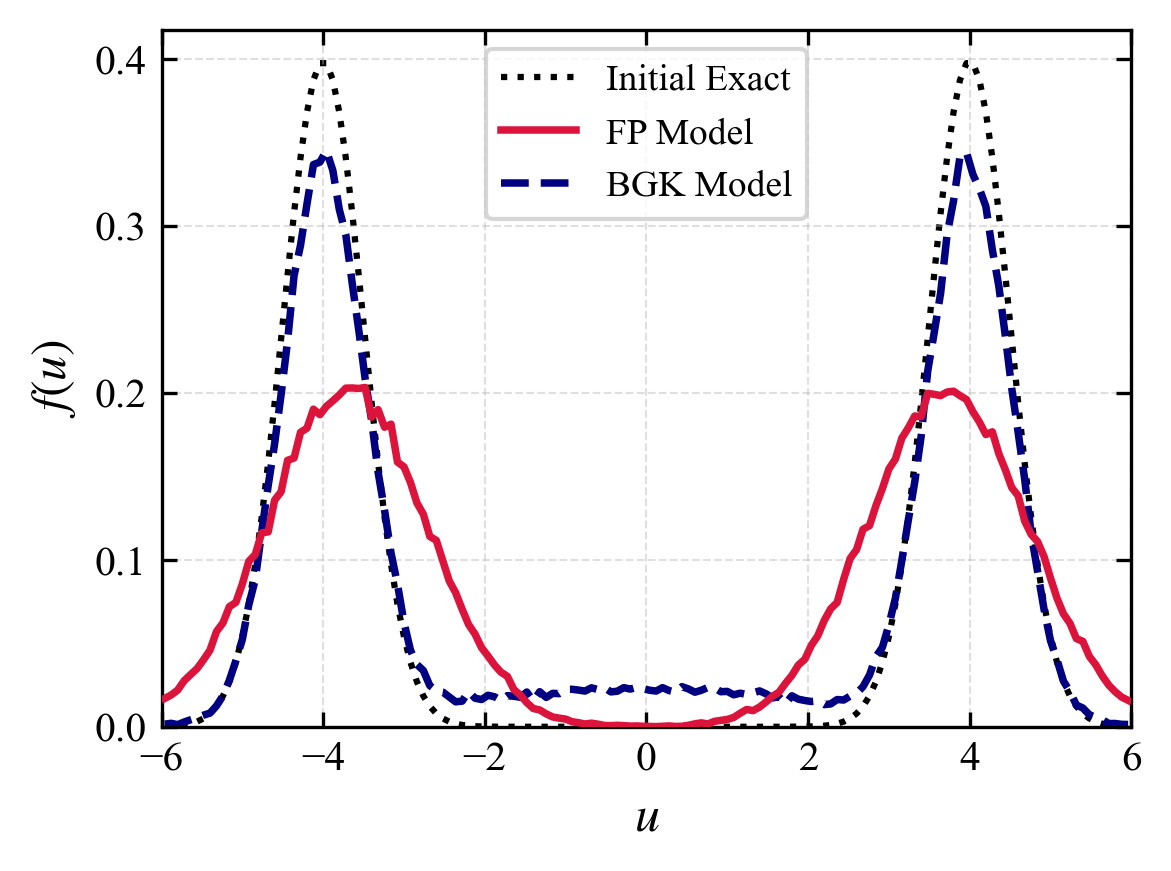}
    \caption{}
    \label{fig:homo-b}
    \end{subfigure}
    \hfill
    \begin{subfigure}[b]{0.48\textwidth}
    \centering
    \includegraphics[width=1.0\linewidth]{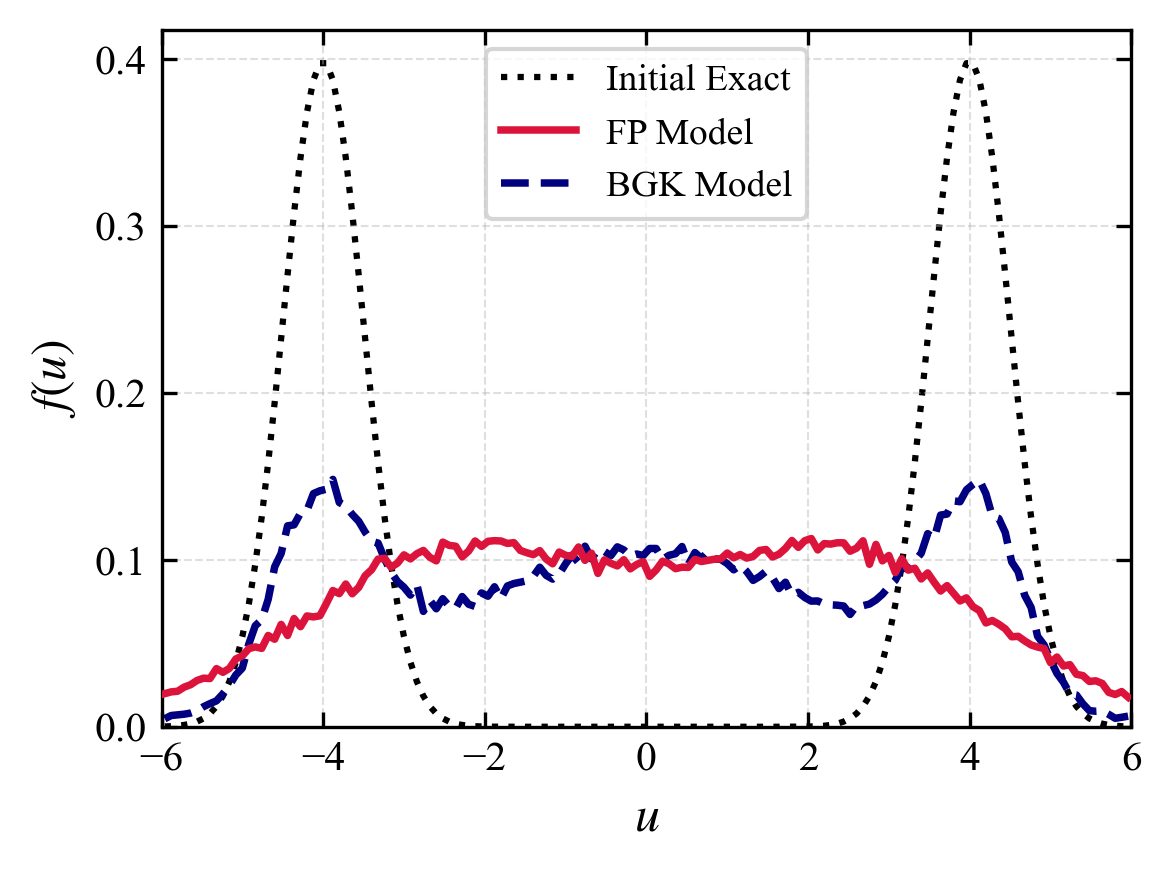}
    \caption{}
    \label{fig:homo-c}
    \end{subfigure}
  	\hfill
    \begin{subfigure}[b]{0.48\textwidth}
    \centering
    \includegraphics[width=1.0\linewidth]{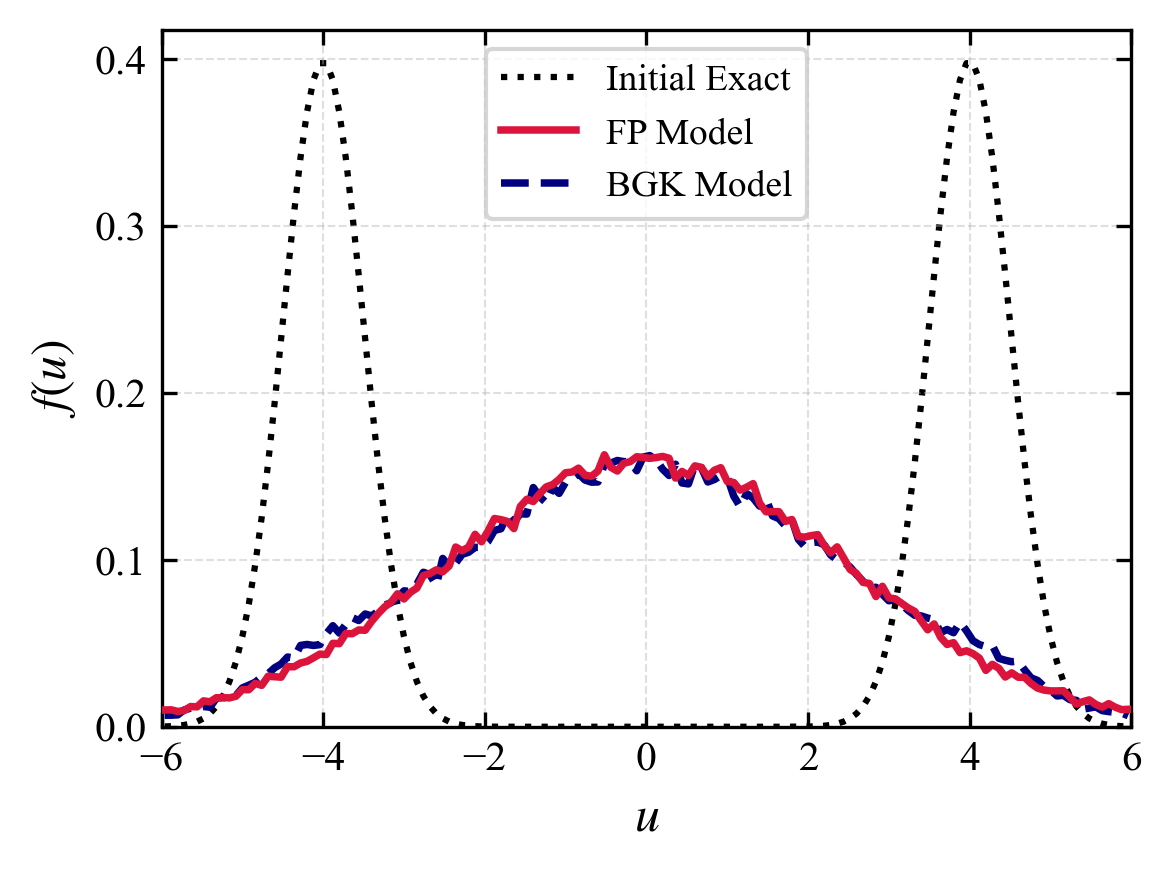}
    \caption{}
    \label{fig:homo-d}
    \end{subfigure}
    \caption{Comparison of domain-averaged velocity distribution functions $f(u)$ between PIC-FP and PIC-BGK models during homogeneous relaxation at different time instances: (a) $t/\tau = 0$, (b) $t/\tau \approx 0.1$, (c) $t/\tau \approx 0.7$, and (d) $t/\tau \approx 2$. The distributions are computed by   macroparticles across all spatial grid cells and normalized by the total mass so that $\int_{-\infty}^{\infty} f(u) \, \mathrm{d}u = 1$.}
    \label{fig:homo-relaxation}
\end{figure}

\subsection{Nonlinear Landau damping}

This benchmark investigates the dynamic differences between the FP and BGK collision operators across various flow regimes and demonstrates the asymptotic-preserving properties of the proposed UGKWP-FP scheme. We examine a wide range of Knudsen numbers, spanning from $\text{Kn}=\infty$, $\text{Kn}=0.1$ to $\text{Kn}=0.001$. This comprehensive range evaluates the scheme's performance across diverse flow regimes without requiring explicit resolution of the mean free path.

The system is initialized with a spatially perturbed Maxwellian distribution in phase space:
$$f(x, \boldsymbol{u}, 0) = \rho_0 \left[1 + \alpha \cos(k x)\right] f_0(\boldsymbol{u}),$$
where $\rho_0 = 1.0$ is the uniform background density, $\alpha$ is the perturbation amplitude, and $k$ is the wavenumber. The term $f_0(\boldsymbol{u})$ denotes the unperturbed 3D Maxwellian velocity distribution with a thermal velocity of $u_{th} = 1.0$:
$$f_0(\boldsymbol{u}) = \frac{1}{(2\pi u_{th}^2)^{3/2}} \exp\left(-\frac{u_x^2 + u_y^2 + u_z^2}{2 u_{th}^2}\right).$$

To suppress initial statistical fluctuations, the phase-space distribution is initialized using quasi-random numbers based on low-discrepancy bit-reversal sequences \cite{hammersley2013monte}. The computational domain is defined as $x \in [0, L]$ with length $L = 2\pi/k$, and is uniformly discretized by 128 spatial grid points. The wavenumber and initial perturbation amplitude are set to $k = 0.5$ and $\alpha = 0.5$, respectively. The particle number is initialized as 1000 in each cell and will adaptively change according to the local Knudsen number. The CFL number is taken as 0.5 for UGKWP-FP.

\begin{figure}
    \centering
    \begin{subfigure}[b]{0.48\textwidth}
    \centering
    \includegraphics[width=1.0\linewidth]{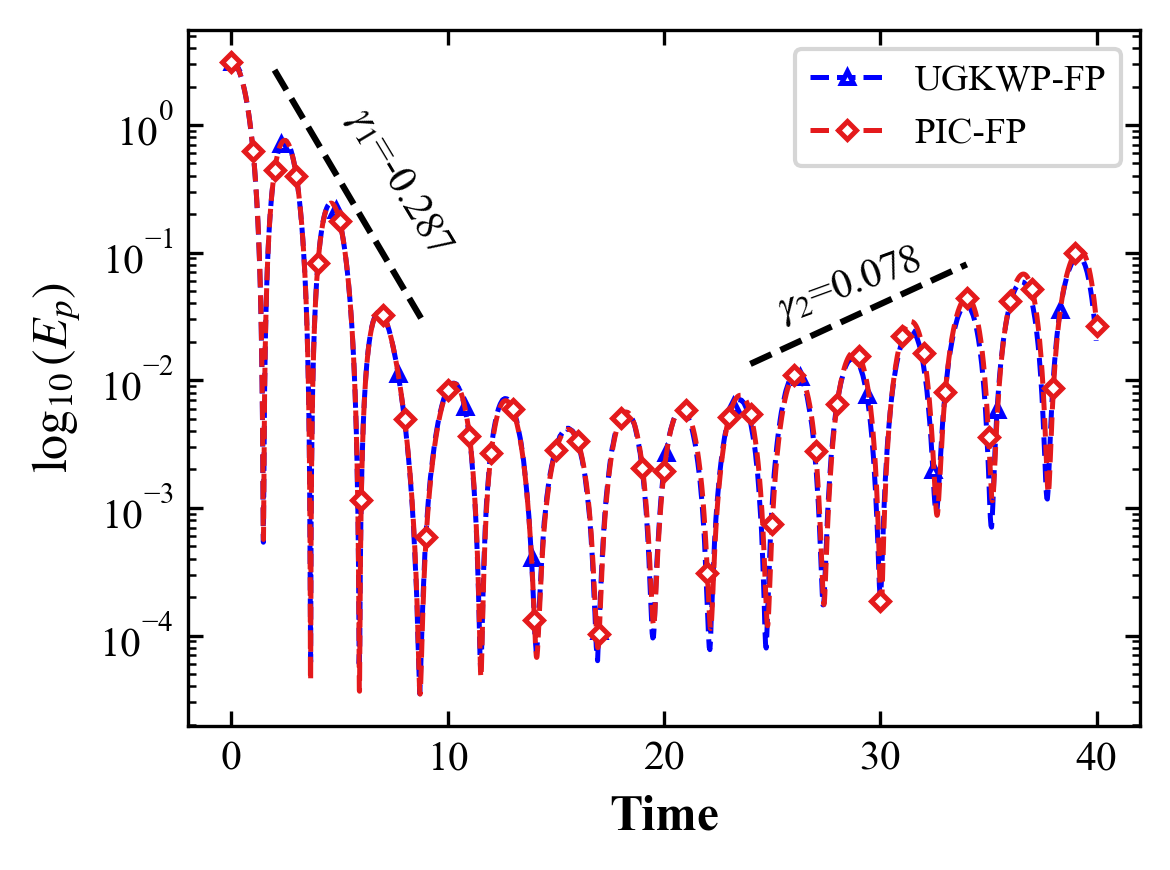}
    \caption{}
    \label{fig:nld-kninf}
    \end{subfigure}
    \begin{subfigure}[b]{0.48\textwidth}
    \centering
    \includegraphics[width=1.0\linewidth]{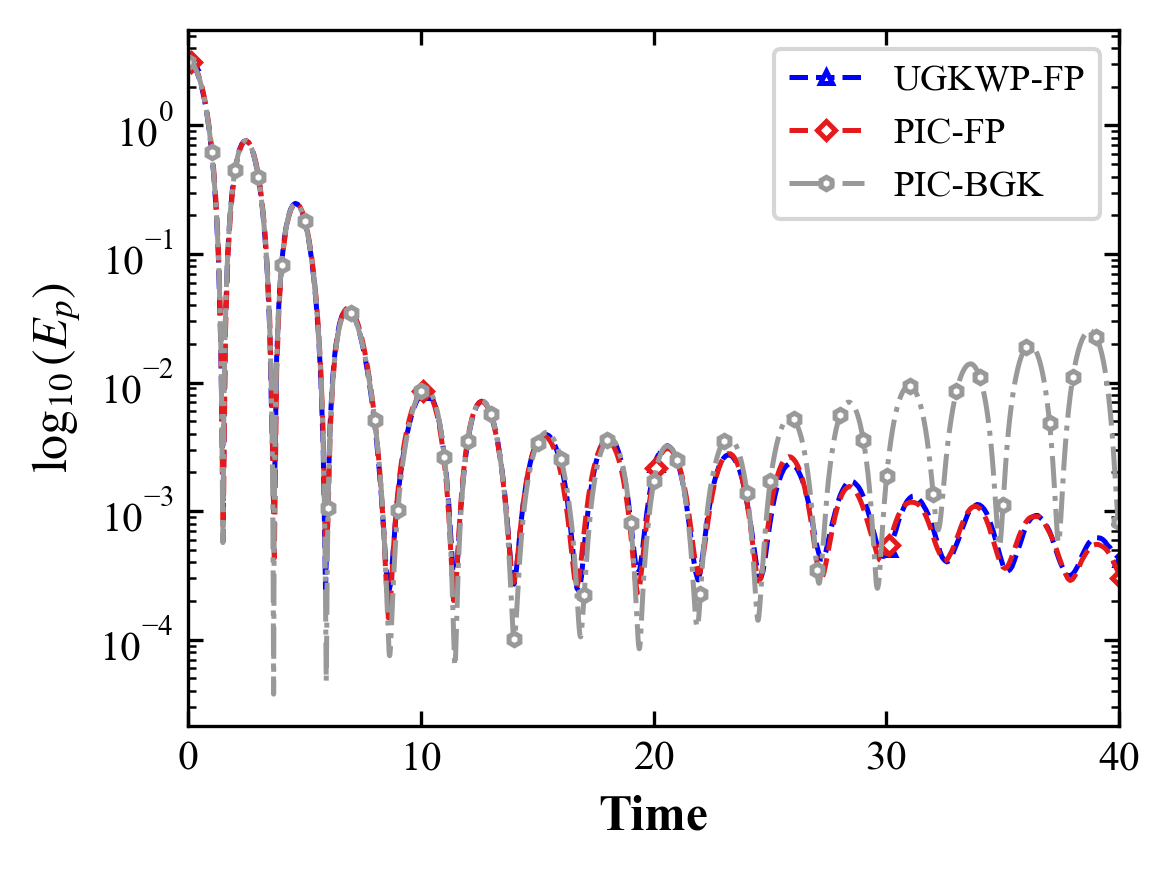}
    \caption{}
    \label{fig:nld-kn100-E}
    \end{subfigure}
    \caption{Temporal evolution of the electric field energy for the nonlinear Landau damping in the highly rarefied regime at (a) $\text{Kn}=\infty$, (b) $\text{Kn}=100$. Results are compared among the UGKWP-FP, PIC-FP, and PIC-BGK methods.}
\end{figure}

Figure \ref{fig:nld-kninf} illustrates the temporal evolution of the electric field energy in the completely collisionless limit ($\text{Kn} = \infty$). In the absence of collisions, the system is entirely governed by the Vlasov-Poisson dynamics, where wave-particle resonance dictates the initial energy evolution and the subsequent nonlinear particle trapping. As depicted in the figure, the electrostatic energy profile obtained from the UGKWP-FP method is in good agreement with that of the reference PIC-FP scheme and the theoretical initial damping rate ($\gamma_1 = -0.287$) and the secondary growth rate ($\gamma_2 = 0.078$). This alignment demonstrates that the proposed UGKWP-FP algorithm can recover the collisionless kinetic limit.

\begin{figure}
    \centering
    \begin{subfigure}[b]{0.325\textwidth}
    \centering
    \includegraphics[width=1.0\linewidth]{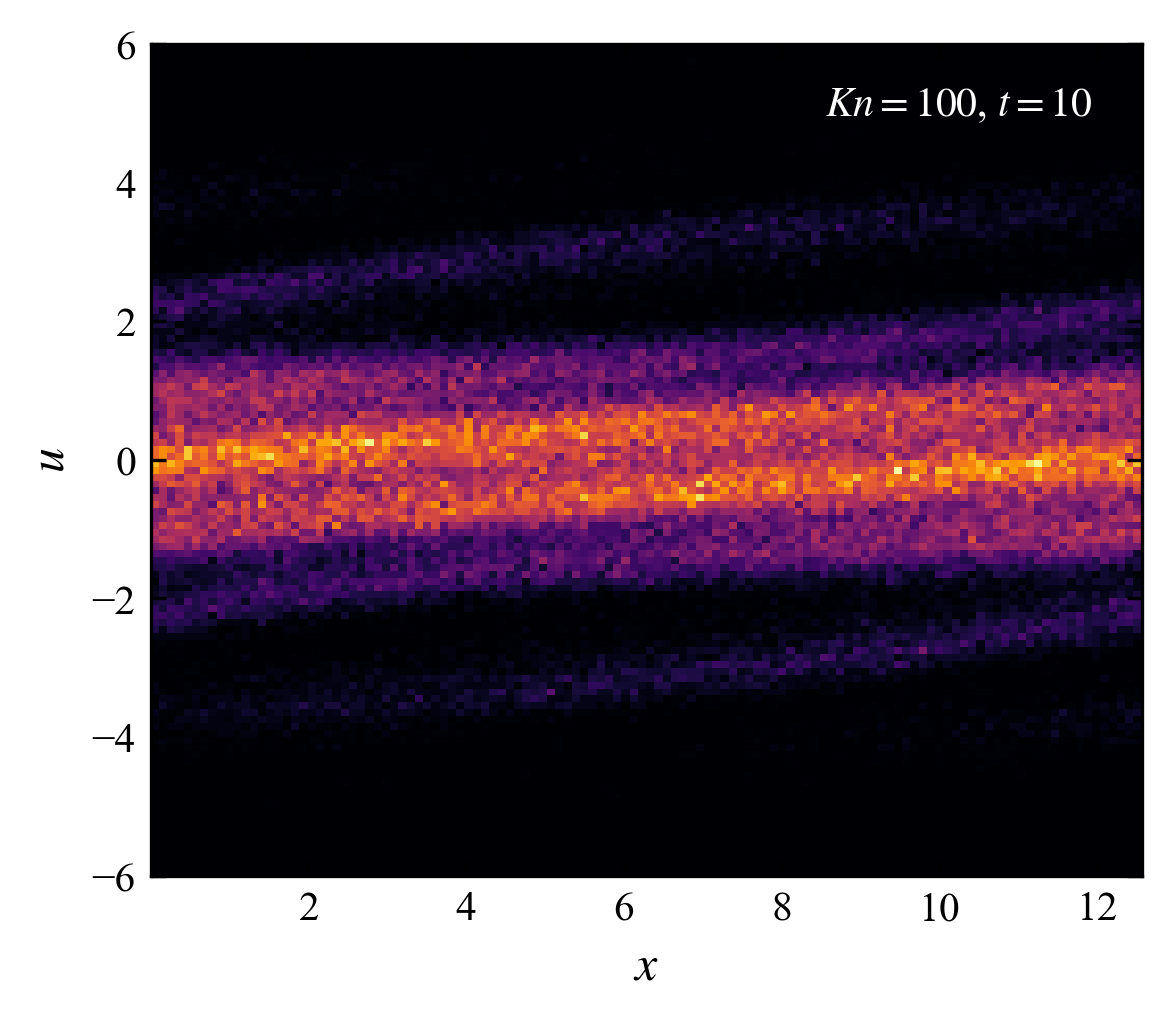}
    \end{subfigure}
    \hfill
    \begin{subfigure}[b]{0.325\textwidth}
    \centering
    \includegraphics[width=1.0\linewidth]{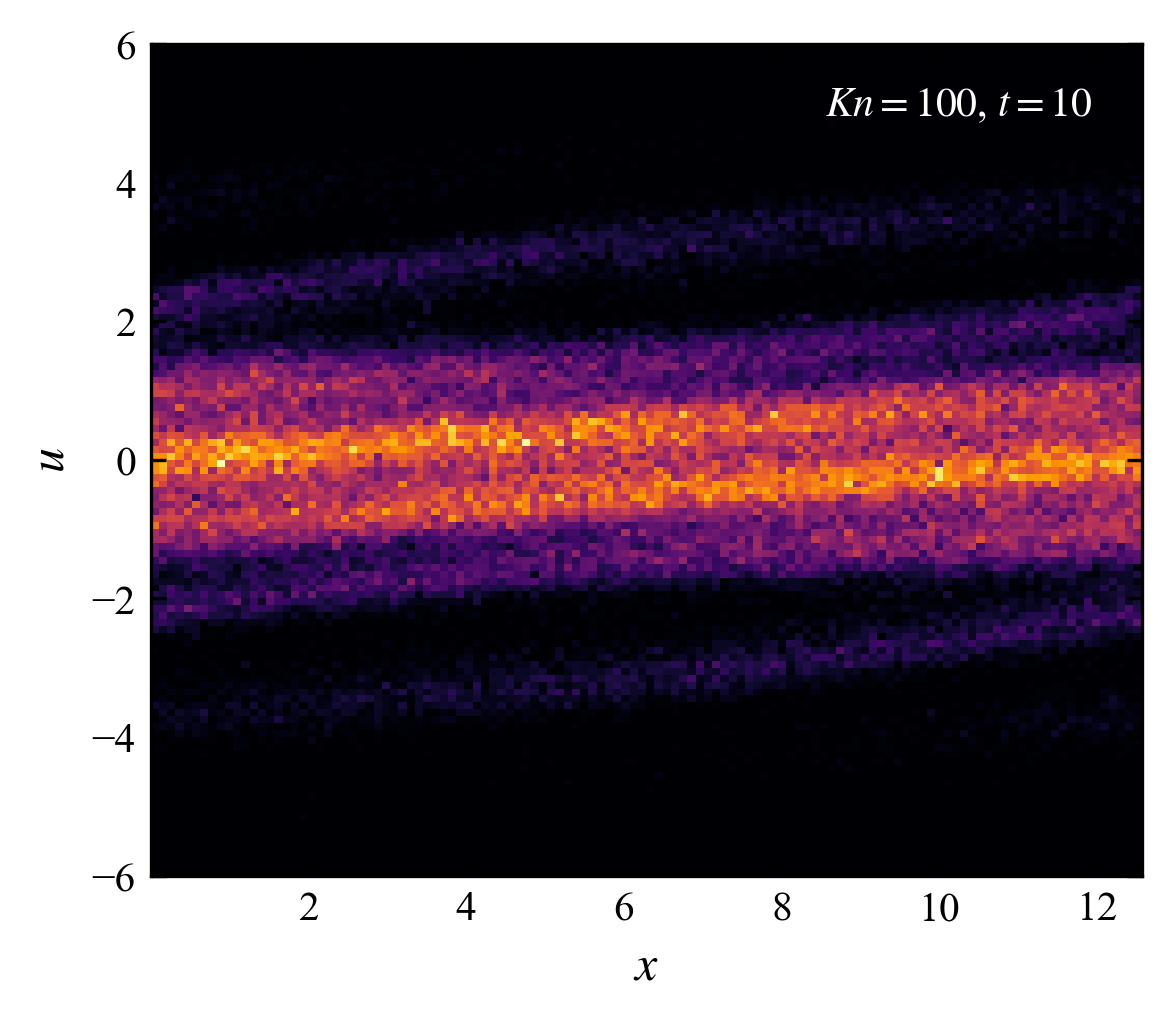}
    \end{subfigure}
    \hfill
    \begin{subfigure}[b]{0.325\textwidth}
    \centering
    \includegraphics[width=1.0\linewidth]{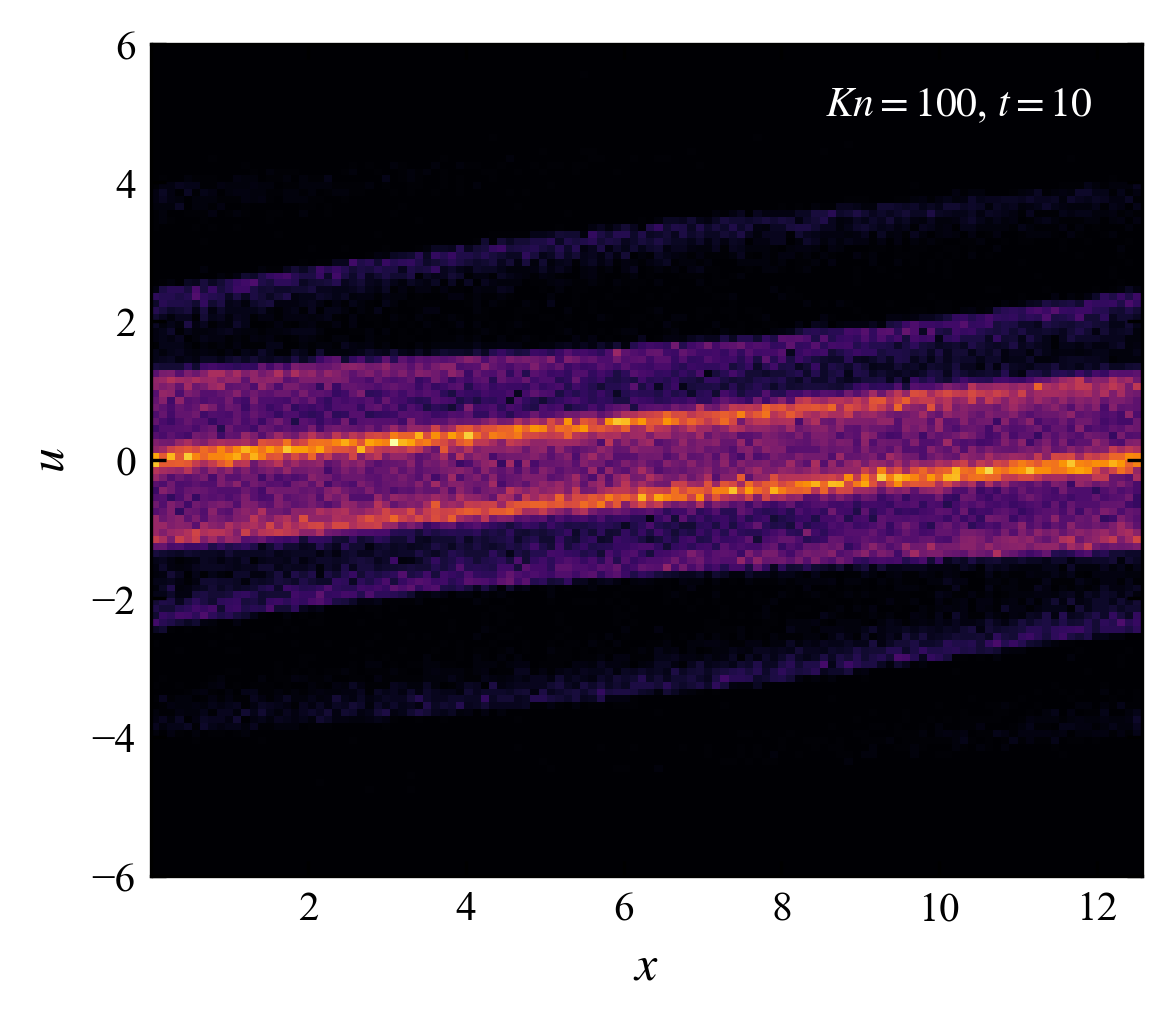}
    \end{subfigure}
    \vfill
    \centering
    \begin{subfigure}[b]{0.325\textwidth}
    \centering
    \includegraphics[width=1.0\linewidth]{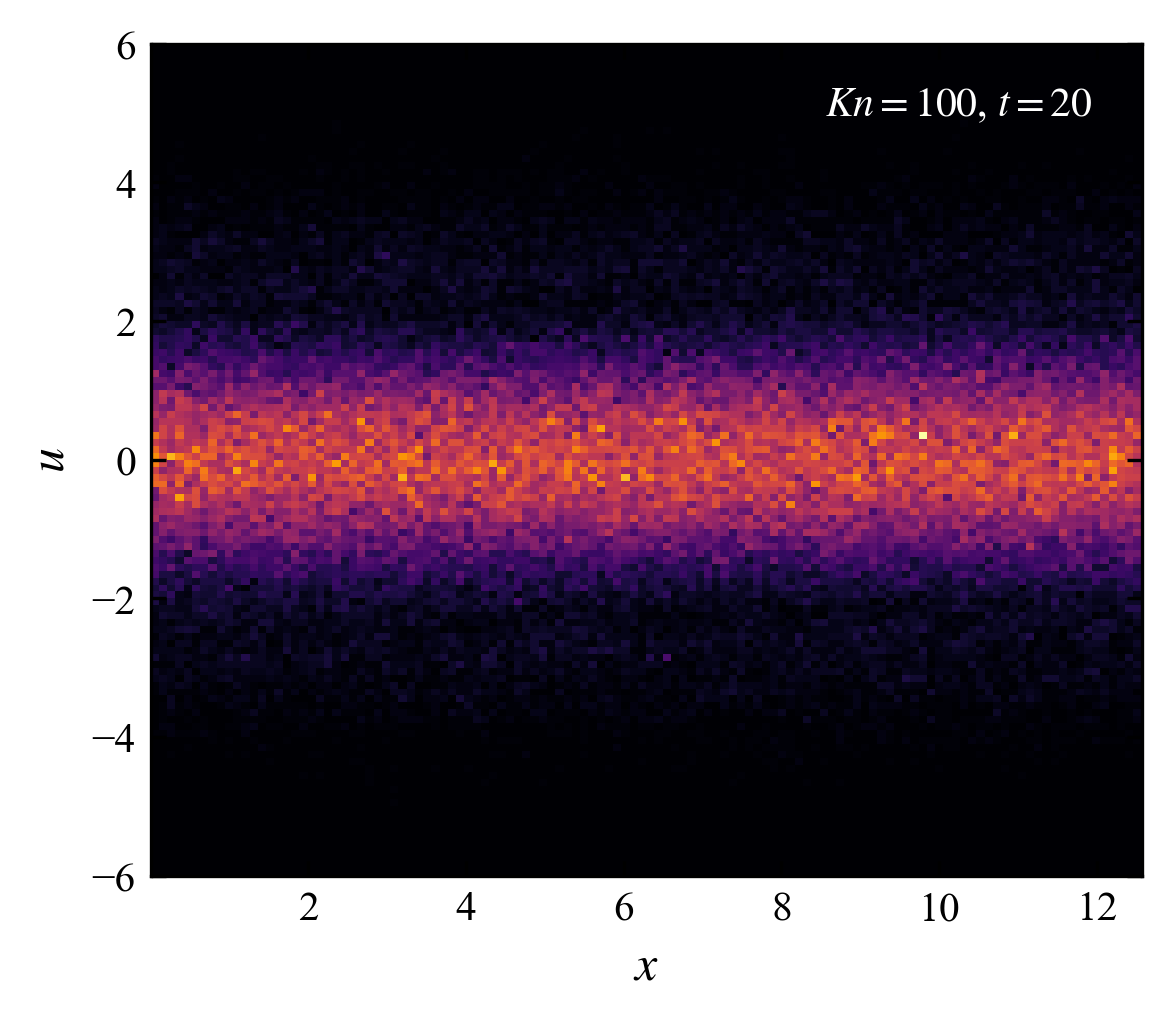}
    \end{subfigure}
    \hfill
    \begin{subfigure}[b]{0.325\textwidth}
    \centering
    \includegraphics[width=1.0\linewidth]{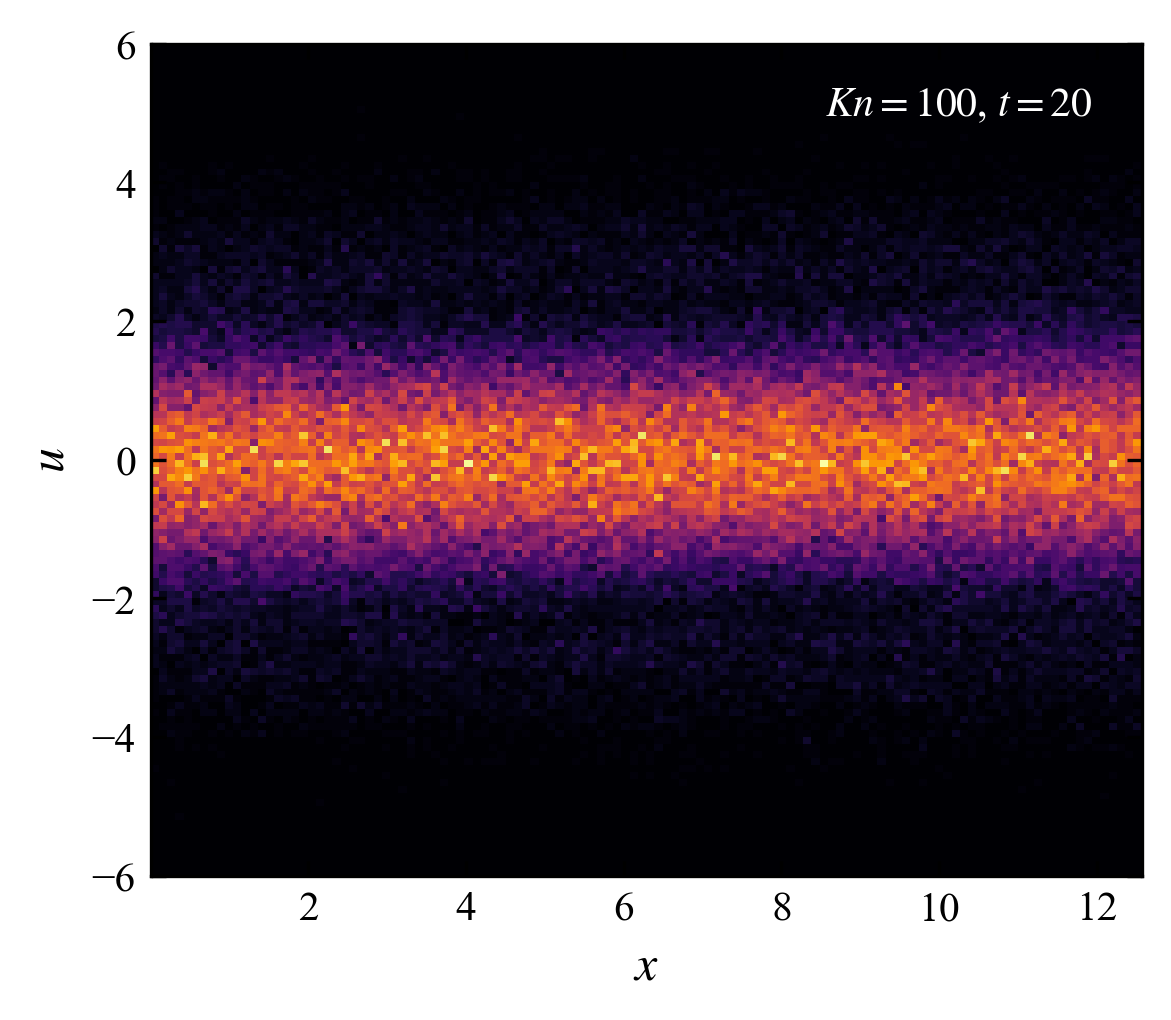}
    \end{subfigure}
    \hfill
    \begin{subfigure}[b]{0.325\textwidth}
    \centering
    \includegraphics[width=1.0\linewidth]{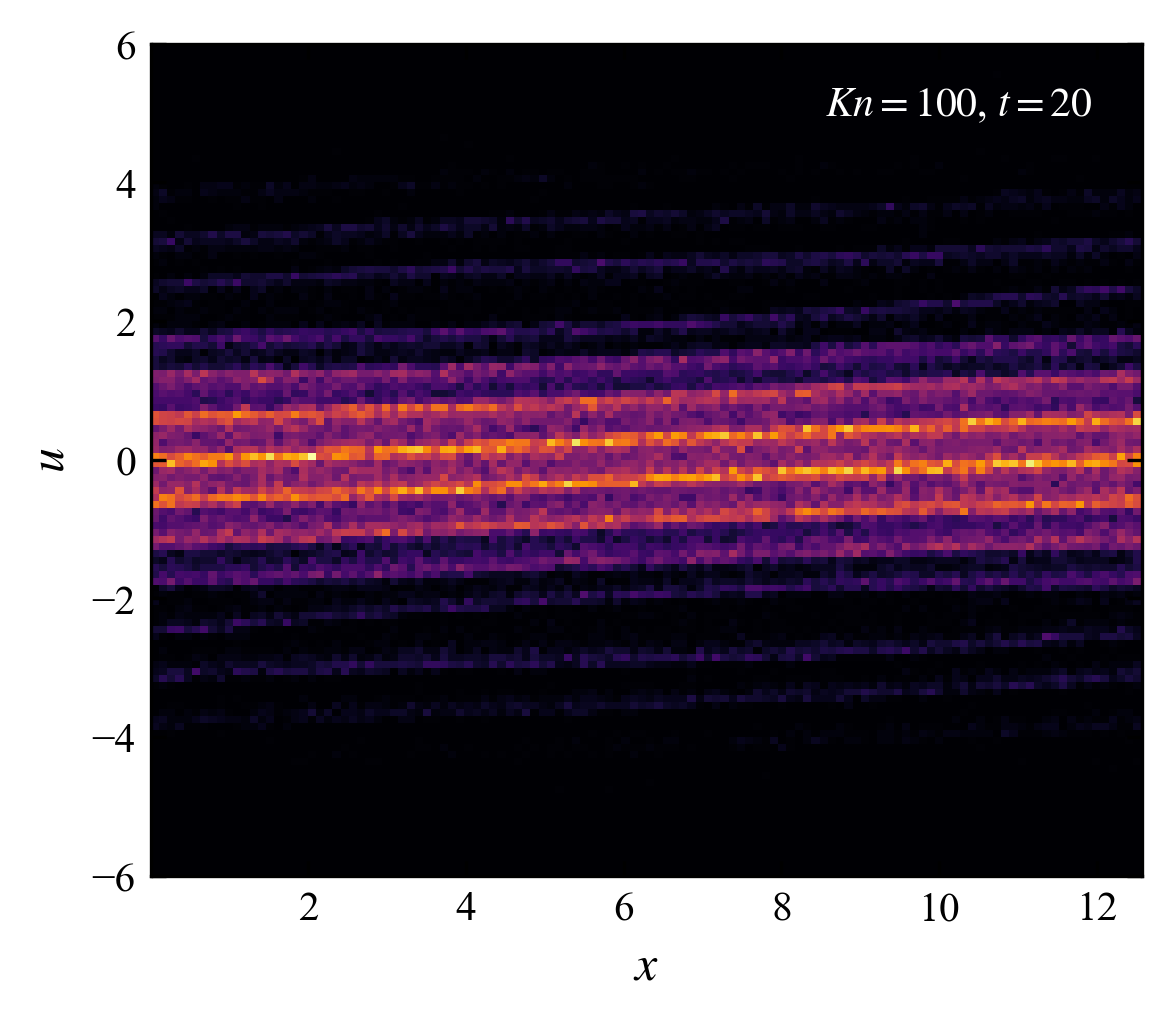}
    \end{subfigure}
    \vfill
    \centering
    \begin{subfigure}[b]{0.325\textwidth}
    \centering
    \includegraphics[width=1.0\linewidth]{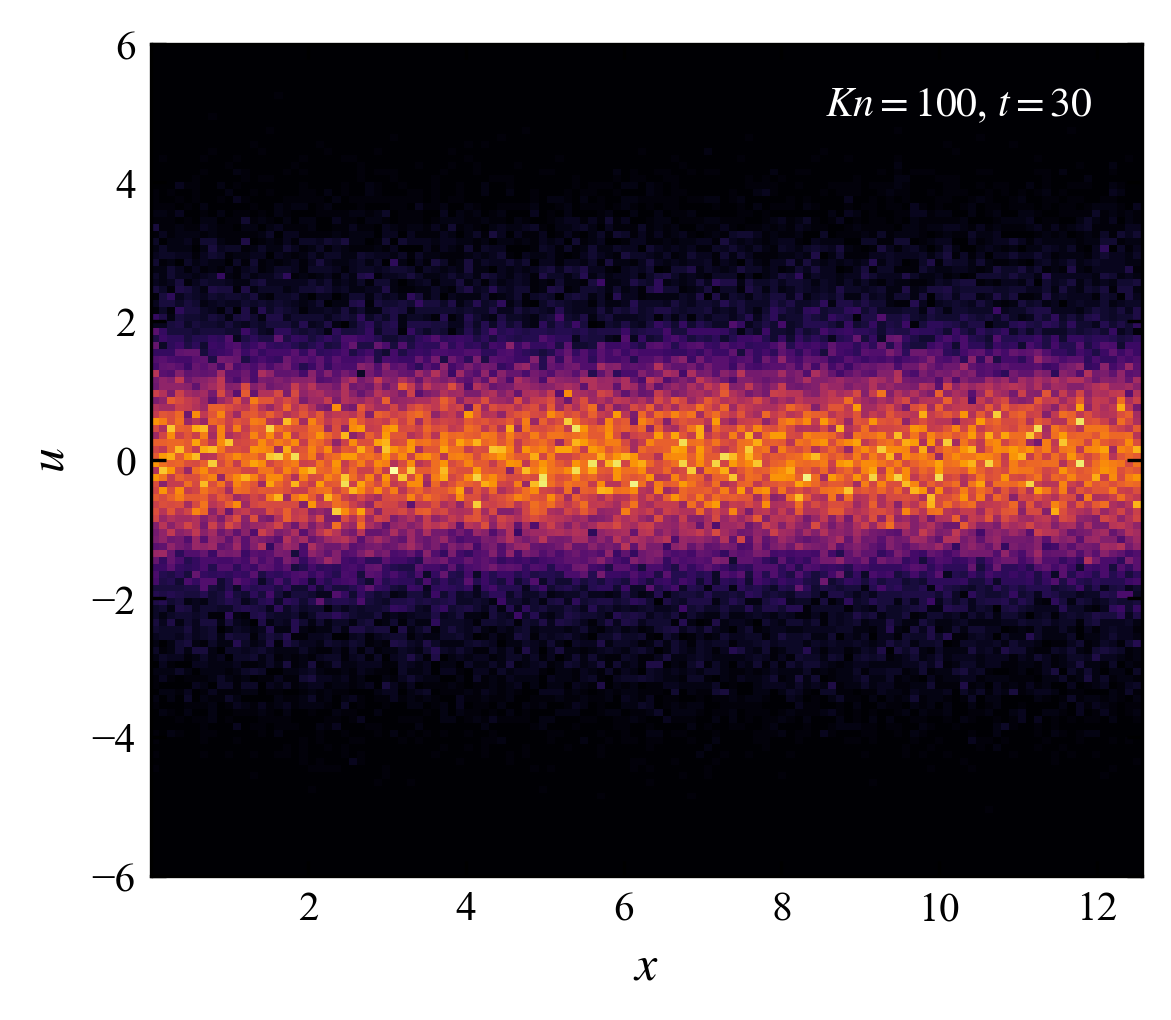}
    \end{subfigure}
    \hfill
    \begin{subfigure}[b]{0.325\textwidth}
    \centering
    \includegraphics[width=1.0\linewidth]{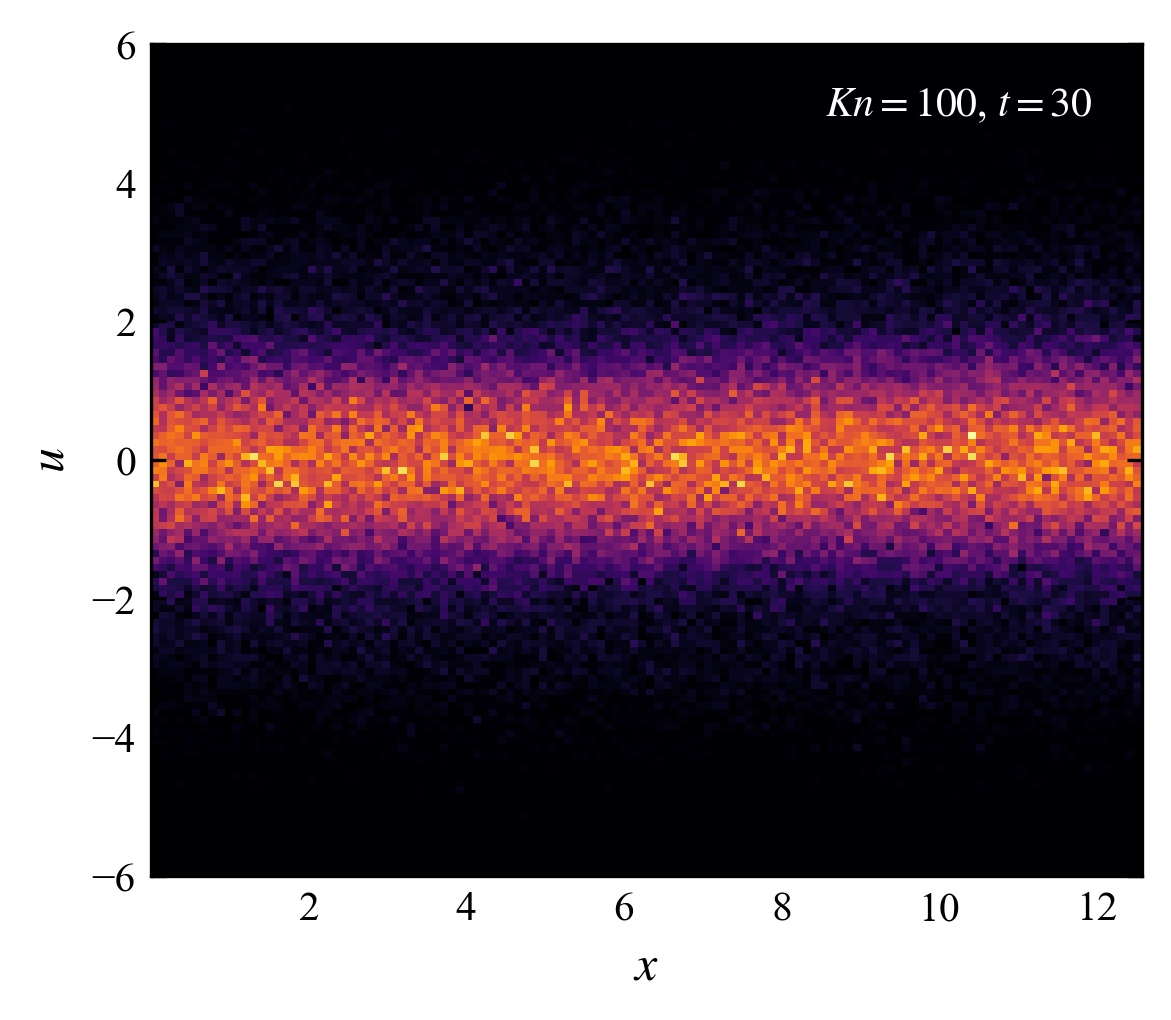}
    \end{subfigure}
    \hfill
    \begin{subfigure}[b]{0.325\textwidth}
    \centering
    \includegraphics[width=1.0\linewidth]{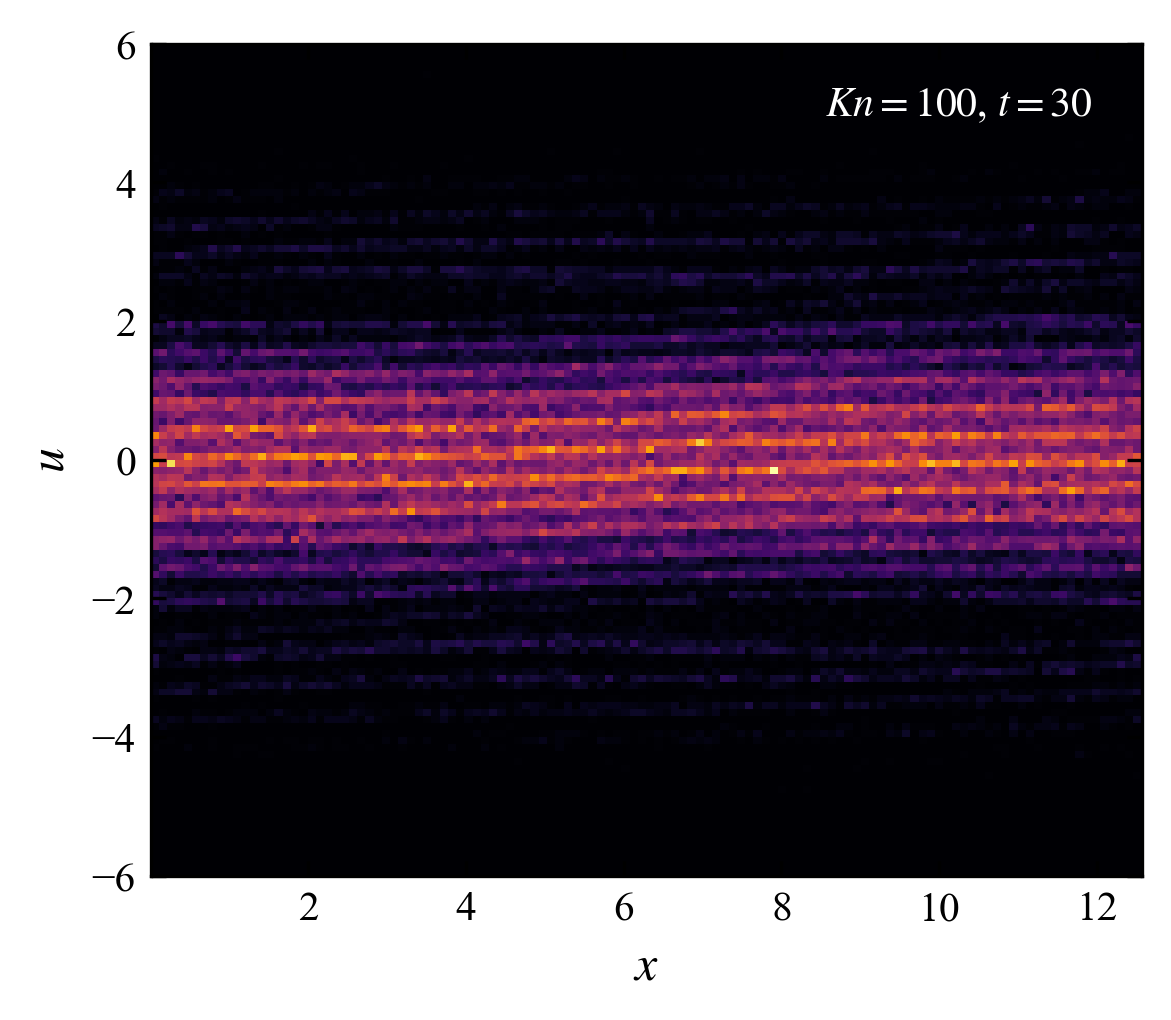}
    \end{subfigure}
    \vfill
    \centering
    \begin{subfigure}[b]{0.325\textwidth}
    \centering
    \includegraphics[width=1.0\linewidth]{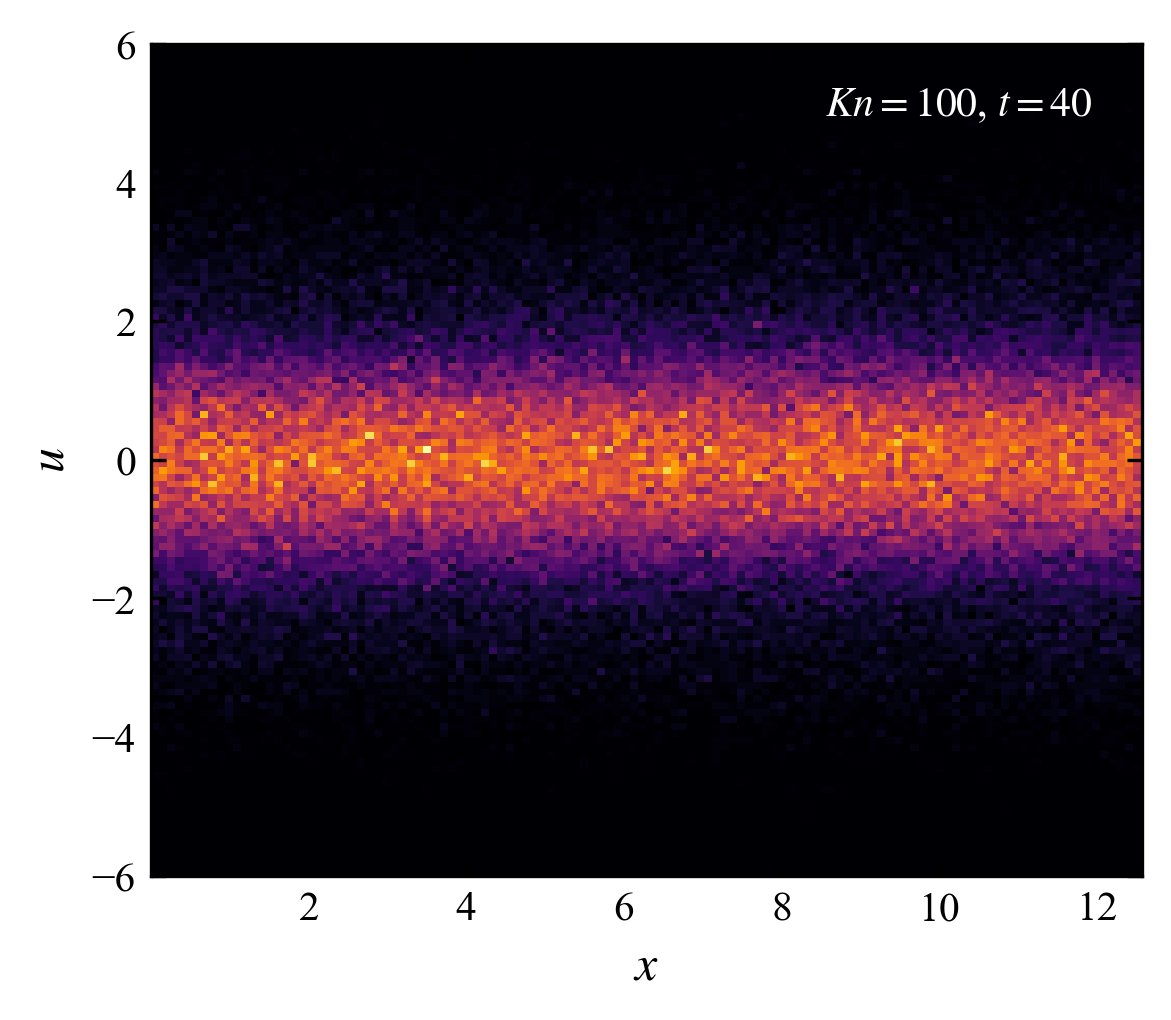}
    \end{subfigure}
    \hfill
    \begin{subfigure}[b]{0.325\textwidth}
    \centering
    \includegraphics[width=1.0\linewidth]{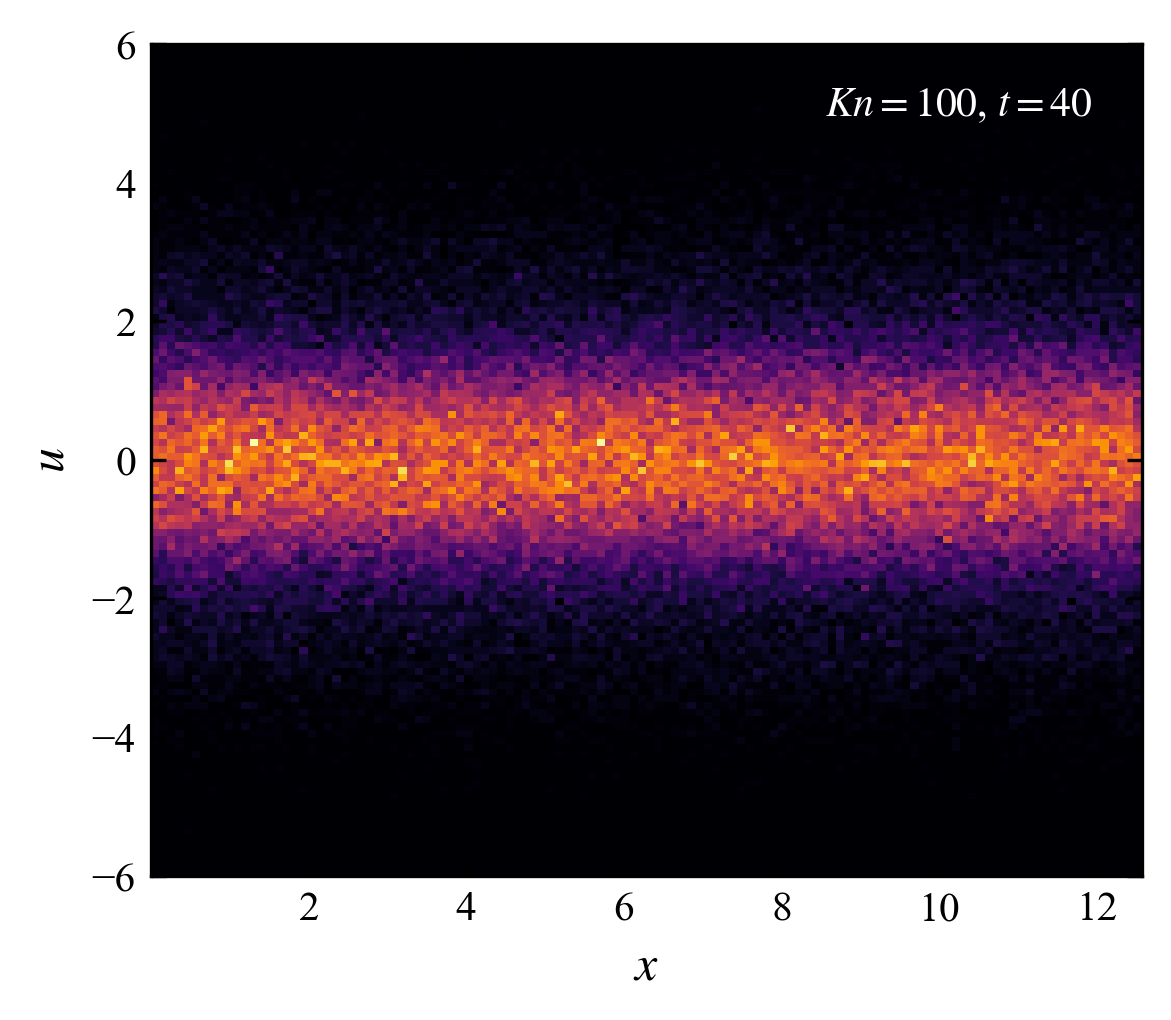}
    \end{subfigure}
    \hfill
    \begin{subfigure}[b]{0.325\textwidth}
    \centering
    \includegraphics[width=1.0\linewidth]{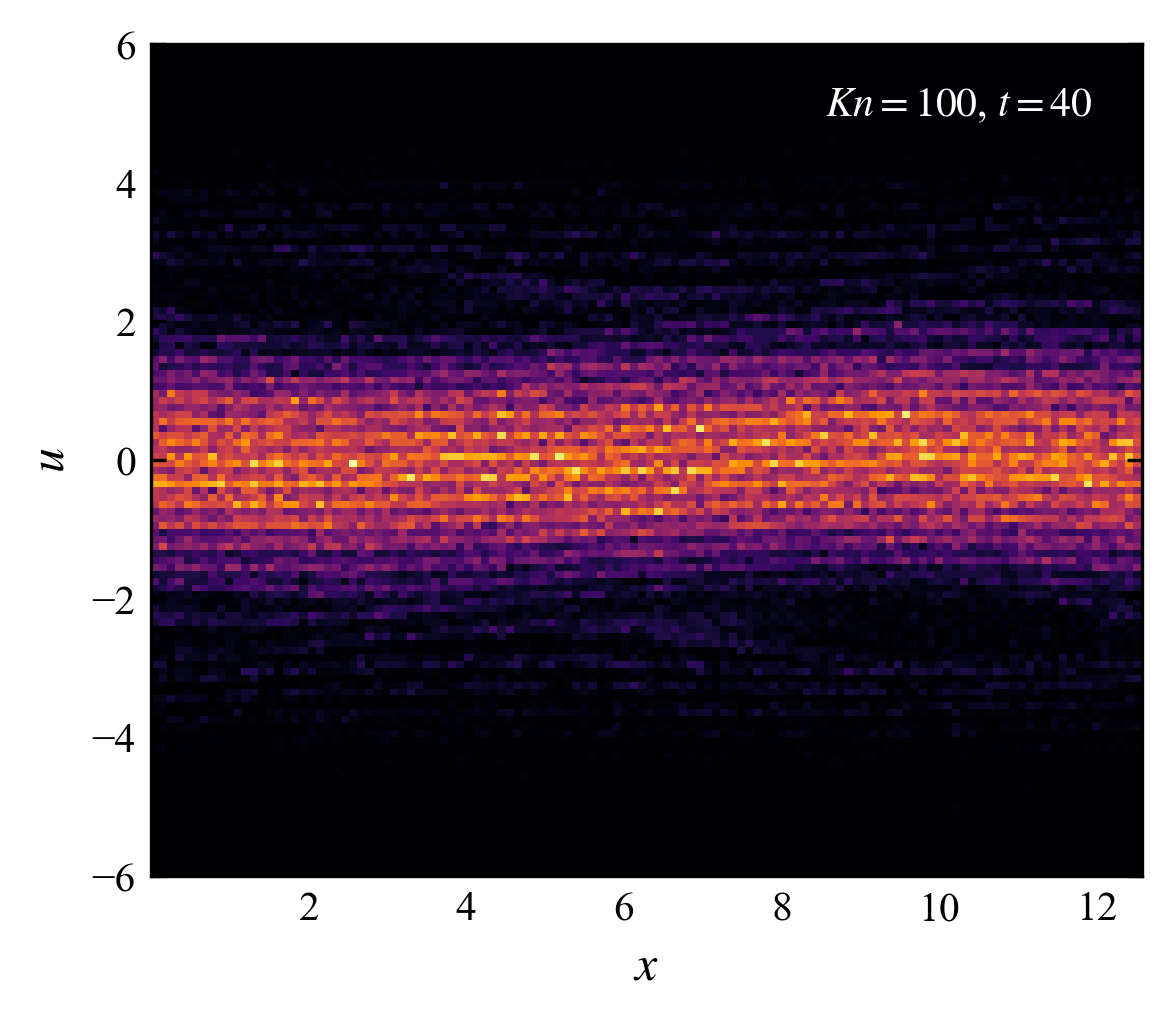}
    \end{subfigure}
	\caption{Snapshots of the phase-space distribution function for the nonlinear Landau damping in the highly rarefied regime ($\text{Kn}=100$). The rows, from top to bottom, correspond to different simulation times at $t=10$, $20$, $30$, and $40$. The columns, from left to right, display the results obtained using the PIC-FP, UGKWP-FP, and PIC-BGK methods, respectively.}
    \label{fig:nld-phase-kn100}
\end{figure}

To further investigate the influence of weak collisions, we examine the highly rarefied regime with $\text{Kn} = 100$, corresponding to a relaxation time of $\tau \approx 100$ and a relaxation rate of $\nu \approx 0.01$. Figure~\ref{fig:nld-kn100-E} and Figure~\ref{fig:nld-phase-kn100} present the temporal evolution of electrostatic energy and the phase-space distribution, respectively. The results obtained by the proposed UGKWP-FP algorithm exhibit close agreement with the benchmark PIC-FP solution.

During the early evolution ($t < 20$), both the FP and BGK models yield nearly identical electric field energy profiles, capturing the characteristic wave-particle interactions. However, discrepancies emerge in the late-stage evolution ($t > 20$). While the FP model exhibits a continuous decay in the electric field energy, the BGK model predicts a notable resurgence and growth in energy amplitude. This discrepancy stems from the distinct mechanisms through which the FP and BGK terms act on velocity space. 

As shown in Fig.~\ref{fig:nld-phase-kn100}, the phase space of the FP model becomes smoothed by $t = 10$ and relaxes toward an equilibrium state for $t \ge 20$. Conversely, the BGK model allows fine phase-space structures to persist throughout $t = 20 \text{--} 40$, thereby driving the secondary energy growth observed in Fig.~\ref{fig:nld-kn100-E}. These findings highlight that, even under highly rarefied conditions ($\text{Kn}=100$), the diffusion mechanism inherent to the FP model plays a different role compared to the BGK model in suppressing phase-space filamentation and driving the system toward physical thermalization.

\begin{figure}
    \centering
    \begin{subfigure}[b]{0.48\textwidth}
    \centering
    \includegraphics[width=1.0\linewidth]{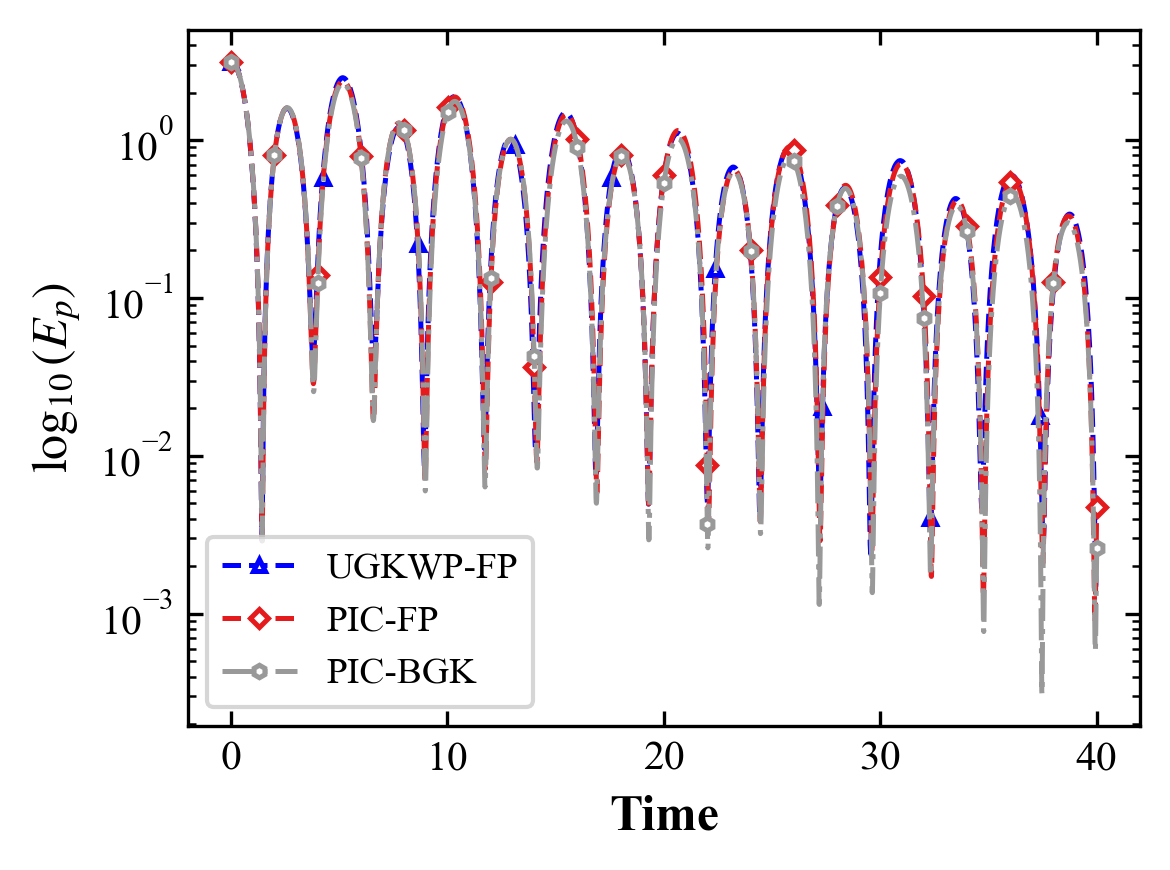}
    \caption{}
    \label{fig:nld-kn0.1-a}
    \end{subfigure}
    \hfill
    \begin{subfigure}[b]{0.48\textwidth}
    \centering
    \includegraphics[width=1.0\linewidth]{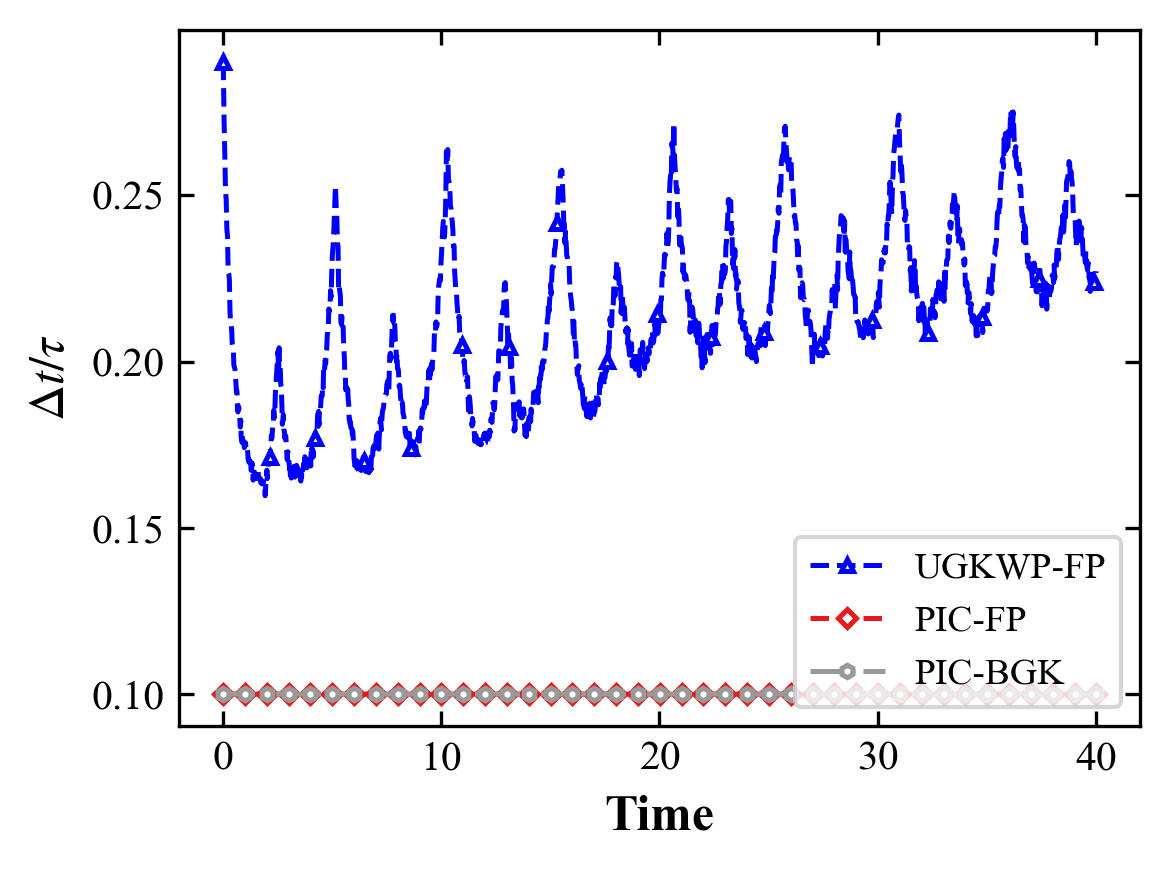}
    \caption{}
    \label{fig:nld-kn0.1-b}
    \end{subfigure}
    \caption{Temporal evolution of (a) the electric field energy and (b) the normalized time step $\Delta t/\tau$ for nonlinear Landau damping at $\text{Kn} = 0.1$. Results are compared among the UGKWP-FP, PIC-FP, and PIC-BGK methods.}
    \label{fig:nld-kn0.1}
\end{figure}

Figure \ref{fig:nld-kn0.1} presents the temporal evolution of the electric field energy and the normalized time step for nonlinear Landau damping in the transitional regime at $\text{Kn} = 0.1$. In this regime, the macroscopic time step $\Delta t$ and the microscopic collision time $\tau$ are on the same order of magnitude. As shown in Fig.~\ref{fig:nld-kn0.1}(a), the energy evolution predicted by the UGKWP-FP method is in good agreement with the reference PIC-FP solution. Notably, the discrepancy between the BGK and FP collision models becomes quite subtle in this transitional regime. This suggests that as collisionality increases, the macroscopic relaxation effects described by the two collision operators gradually converge. 

Figure \ref{fig:nld-kn0.001} presents the simulation results for nonlinear Landau damping in the near-continuum regime at $\text{Kn}=0.001$. As shown in Fig.~\ref{fig:nld-kn0.001-a}, strong collisional relaxation
 heavily suppress the collisionless Landau damping mechanism. Consequently, while the electric field energy exhibits sustained oscillations,  its overall magnitude remains nearly undissipated over time. The energy evolution profiles predicted by the proposed UGKWP-FP scheme agree well with those obtained from the reference PIC-FP and PIC-BGK models. This agreement confirms that the UGKWP-FP method correctly recovers the physics in the continuum limit.

Furthermore, Fig.~\ref{fig:nld-kn0.001-b} highlights the computational efficiency and the asymptotic-preserving property of the proposed scheme by comparing the temporal evolution of the normalized time step, $\Delta t/\tau$. For the PIC-FP and PIC-BGK solvers, the time steps are strictly constrained by the stiff microscopic collision time $\tau \sim 0.001$, forcing $\Delta t \le \tau$. Besides, the mesh size is also restricted by the mean free path $\Delta x \le l_{mfp} \sim 0.001$. In contrast, the UGKWP-FP method dynamically adopts a time step that is significantly larger than the mean collision time $\Delta t / \tau \approx 20$. Given that the employed grid size is also vastly larger than the mean free path ($\Delta x = 0.1 \gg l_{mfp} \sim 0.001$), this comparison comprehensively demonstrates the capability of the UGKWP-FP scheme to efficiently simulate near-continuum flows without explicitly resolving the underlying kinetic scales.

\begin{figure}
    \centering
    \begin{subfigure}[b]{0.48\textwidth}
    \centering
    \includegraphics[width=1.0\linewidth]{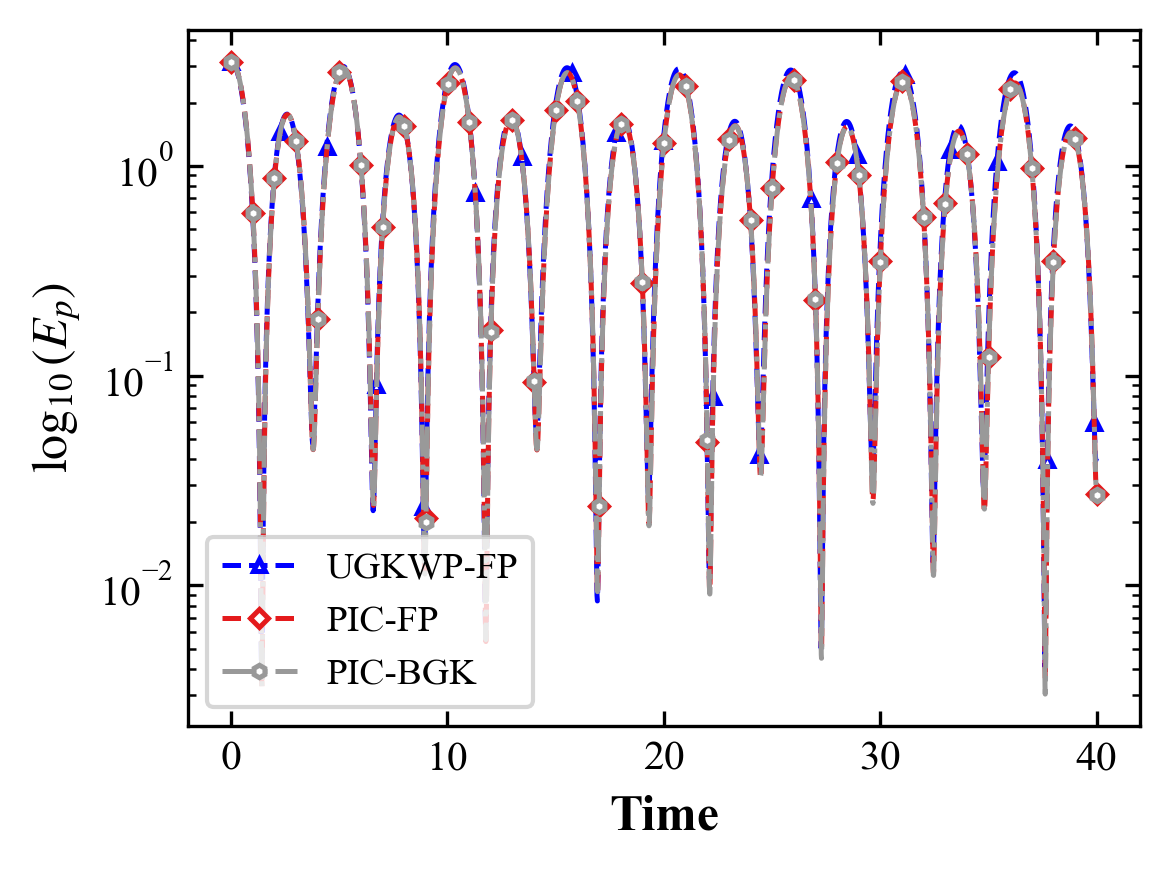}
    \caption{}
    \label{fig:nld-kn0.001-a}
    \end{subfigure}
    \hfill
    \begin{subfigure}[b]{0.48\textwidth}
    \centering
    \includegraphics[width=1.0\linewidth]{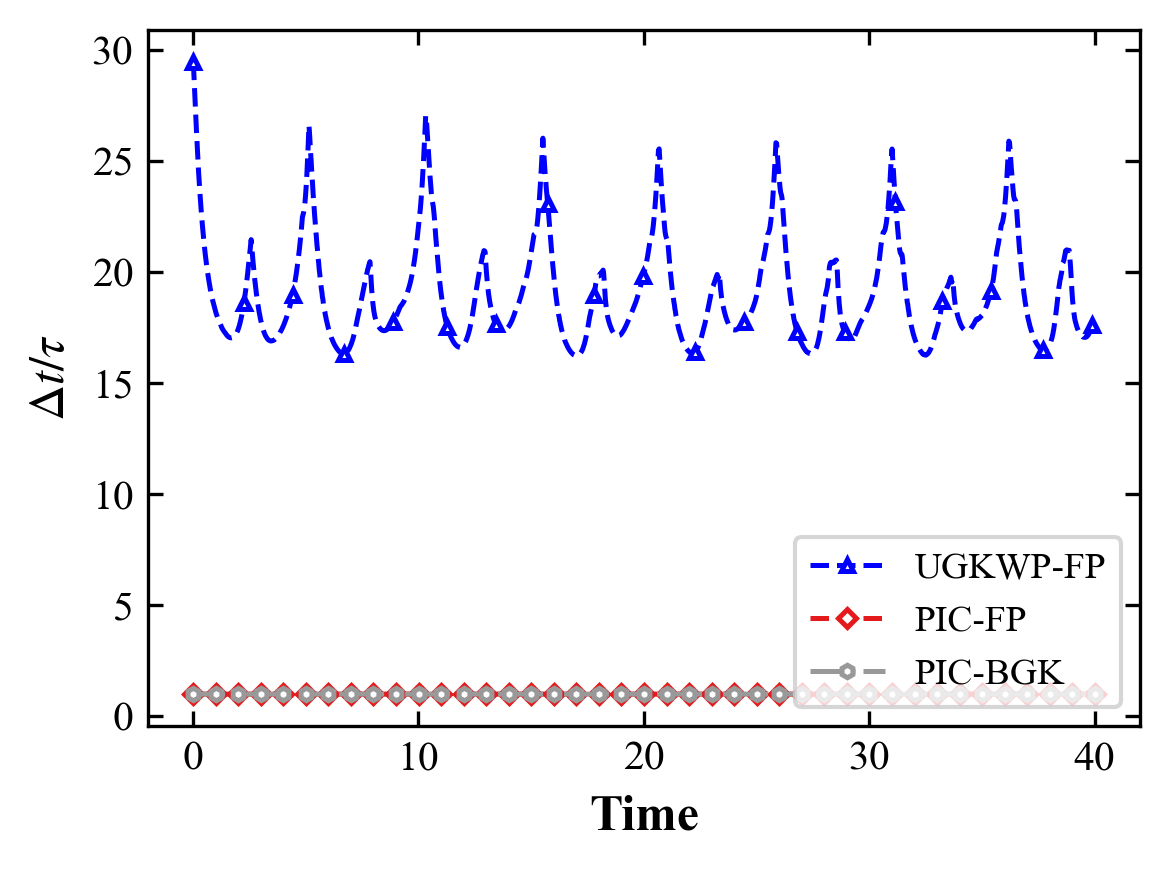}
    \caption{}
    \label{fig:nld-kn0.001-b}
    \end{subfigure}
    \caption{Temporal evolution of (a) the electric field energy and (b) the normalized time step $\Delta t/\tau$ for nonlinear Landau damping in the near-continuum regime at $\text{Kn} = 0.001$. Results are compared among the UGKWP-FP, PIC-FP, and PIC-BGK methods.}
    \label{fig:nld-kn0.001}
\end{figure}

\subsection{Bump-on-tail instability}

We further investigate the difference of the FP and BGK collision operators on the bump-on-tail instability case. This benchmark is particularly sensitive to numerical velocity-space diffusion, which can smooth the positive slope of the distribution near the resonant phase velocity and thereby weaken the instability \cite{Drummond1962}. The initial phase-space distribution function consists of a spatially perturbed superposition of a stationary bulk Maxwellian and a drifting beam Maxwellian:
\begin{equation}
f(x, \boldsymbol{u}, 0) = \rho_0 \left[1 + \alpha \cos(k x)\right] \left( n_p f_{M,p}(\boldsymbol{u}) + n_b f_{M,b}(\boldsymbol{u}) \right),
\end{equation}
where $\rho_0=1$ is the background mass density, $\alpha=0.04$ is the perturbation amplitude, and $k=0.3$ is the wavenumber. The bulk plasma fraction is set to $n_p = 0.9$ with a thermal velocity of $u_{th,p} = 1.0$, while the secondary beam fraction is $n_b = 0.1$ with a reduced thermal velocity $u_{th,b} = 0.5$ and a macroscopic drift velocity $U_b = 4.5$ along the $x$-direction. The normalized 3D Maxwellian distributions for the bulk and the beam are respectively given by:
\begin{equation}
f_{M,p}(\boldsymbol{u}) = \frac{1}{(2\pi u_{th,p}^2)^{3/2}} \exp\left(-\frac{u_x^2 + u_y^2 + u_z^2}{2 u_{th,p}^2}\right),
\end{equation}
\begin{equation}
f_{M,b}(\boldsymbol{u}) = \frac{1}{(2\pi u_{th,b}^2)^{3/2}} \exp\left(-\frac{(u_x - U_b)^2 + u_y^2 + u_z^2}{2 u_{th,b}^2}\right).
\end{equation}
The particle number is initialized as 1000 in each cell and will adaptively change according to the local Knudsen number. The CFL number is taken as 0.5 for UGKWP-FP.

\begin{figure}
    \centering
    \begin{subfigure}[b]{0.48\textwidth}
    \centering
    \includegraphics[width=1.0\linewidth]{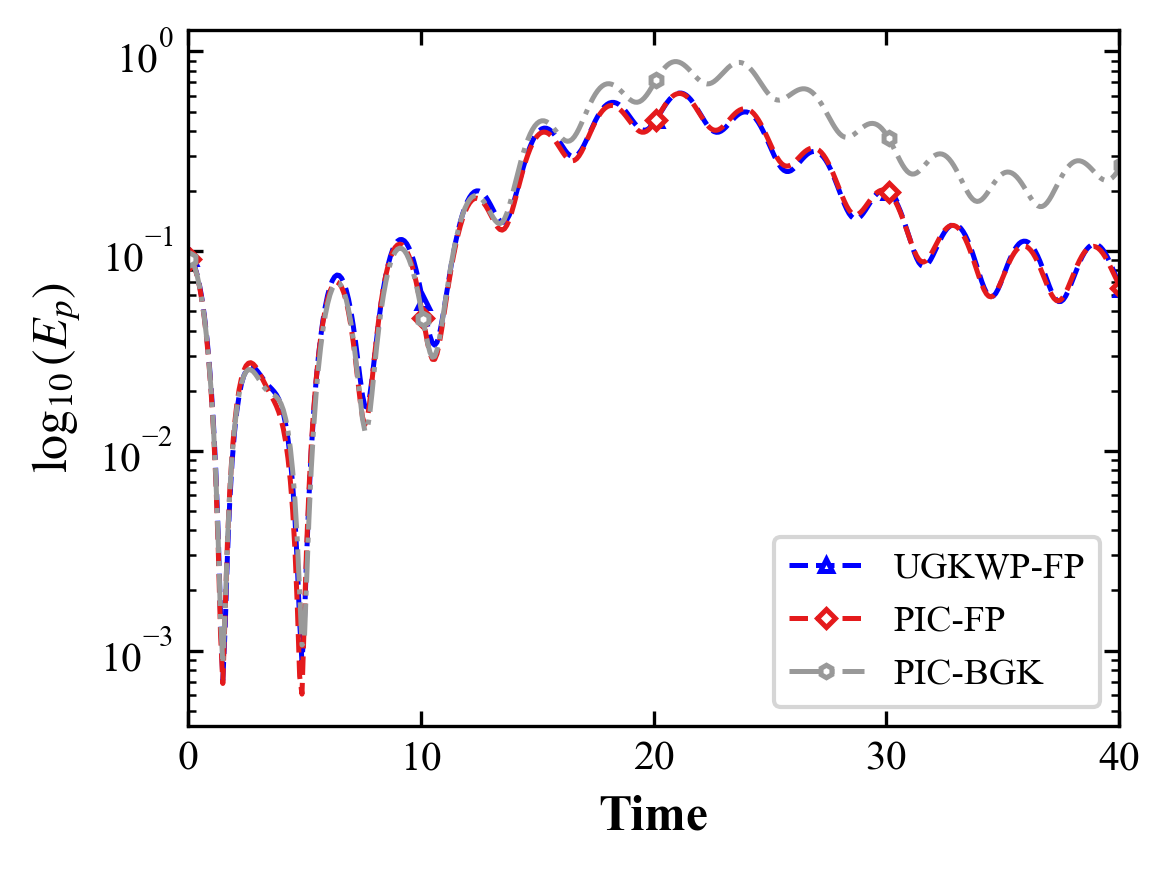}
    \caption{}
    \label{fig:bti-kn100-E}
    \end{subfigure}
        \begin{subfigure}[b]{0.48\textwidth}
    \centering
    \includegraphics[width=1.0\linewidth]{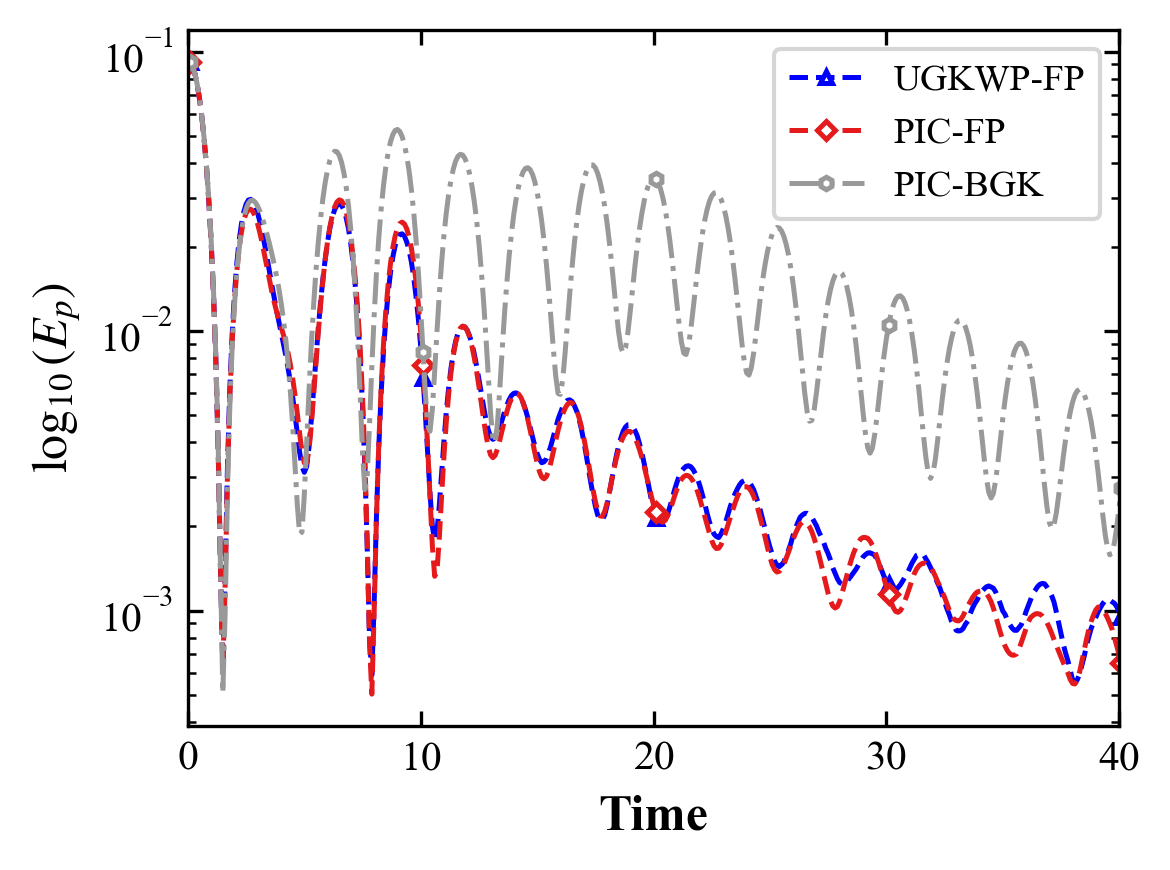}
    \caption{}
    \label{fig:bti-kn10-E}
    \end{subfigure}
    \caption{Temporal evolution of the electric field energy for the bump-on-tail instability in the rarefied regime at (a) $\text{Kn}=100$ and (b) $\text{Kn}=10$. Results are compared among the UGKWP-FP, PIC-FP, and PIC-BGK methods.}
\end{figure}

\begin{figure}
    \centering
    \begin{subfigure}[b]{0.325\textwidth}
    \centering
    \includegraphics[width=1.0\linewidth]{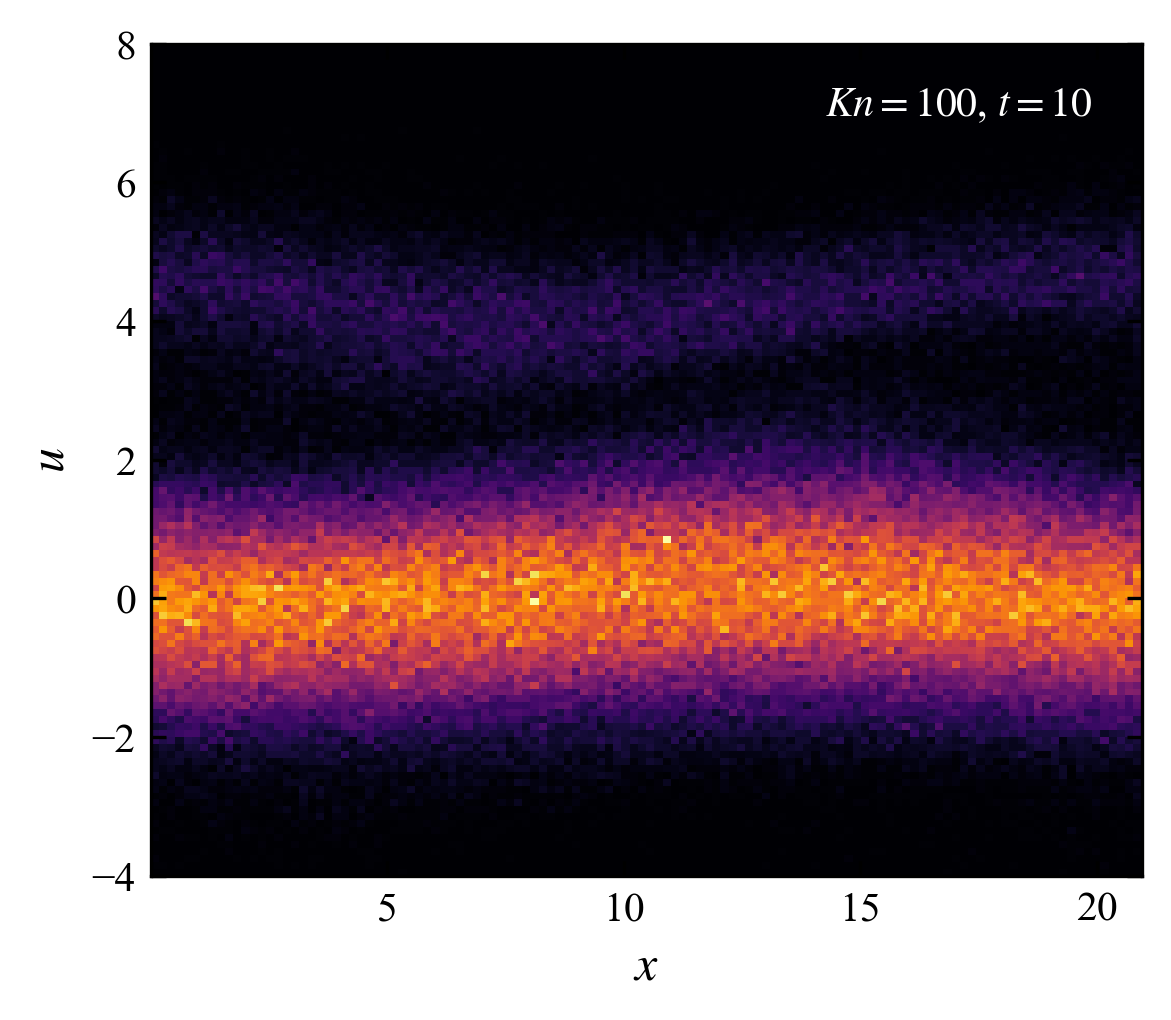}
    \end{subfigure}
    \hfill
    \begin{subfigure}[b]{0.325\textwidth}
    \centering
    \includegraphics[width=1.0\linewidth]{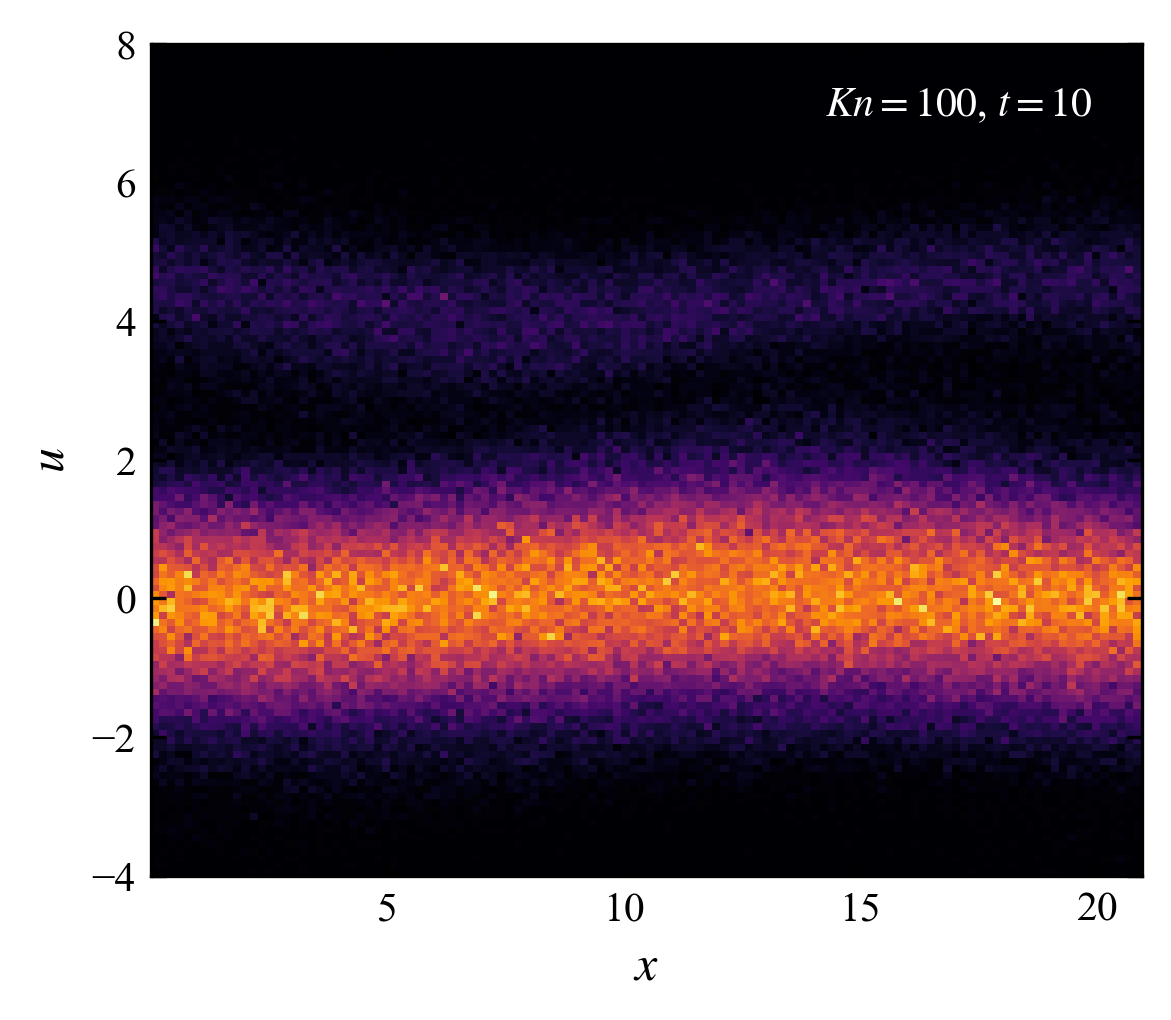}
    \end{subfigure}
    \hfill
    \begin{subfigure}[b]{0.325\textwidth}
    \centering
    \includegraphics[width=1.0\linewidth]{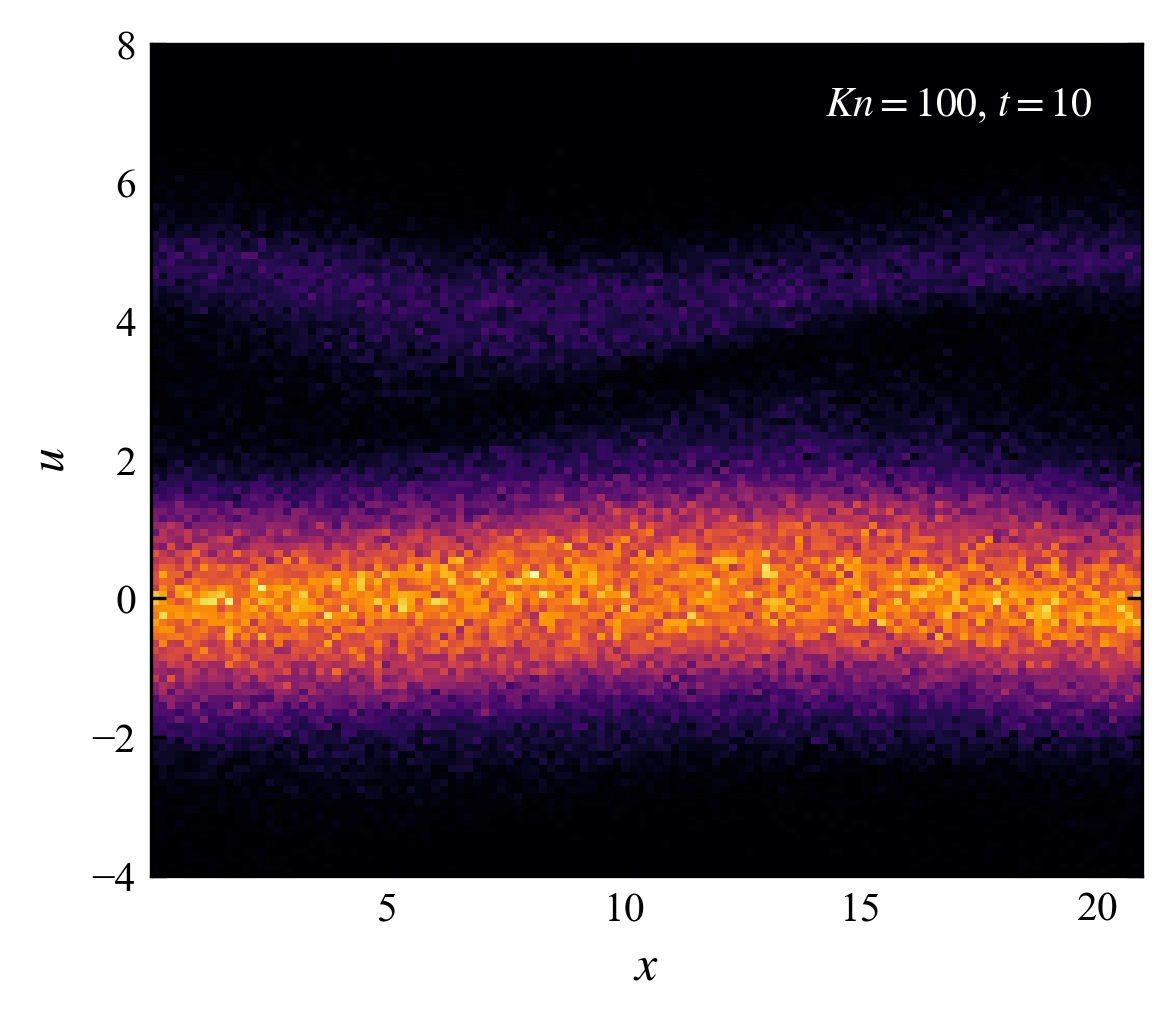}
    \end{subfigure}
    \vfill
    \centering
    \begin{subfigure}[b]{0.325\textwidth}
    \centering
    \includegraphics[width=1.0\linewidth]{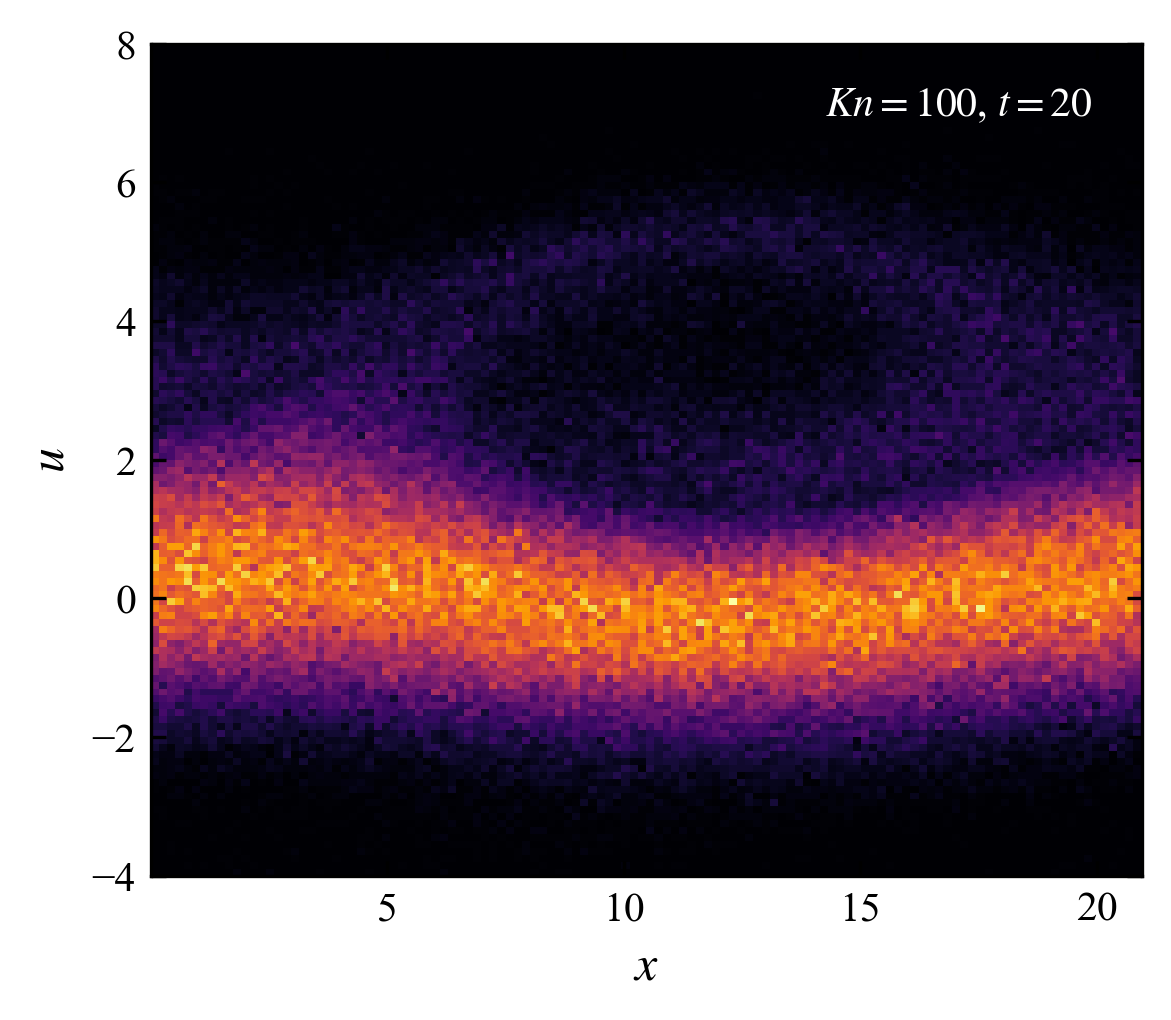}
    \end{subfigure}
    \hfill
    \begin{subfigure}[b]{0.325\textwidth}
    \centering
    \includegraphics[width=1.0\linewidth]{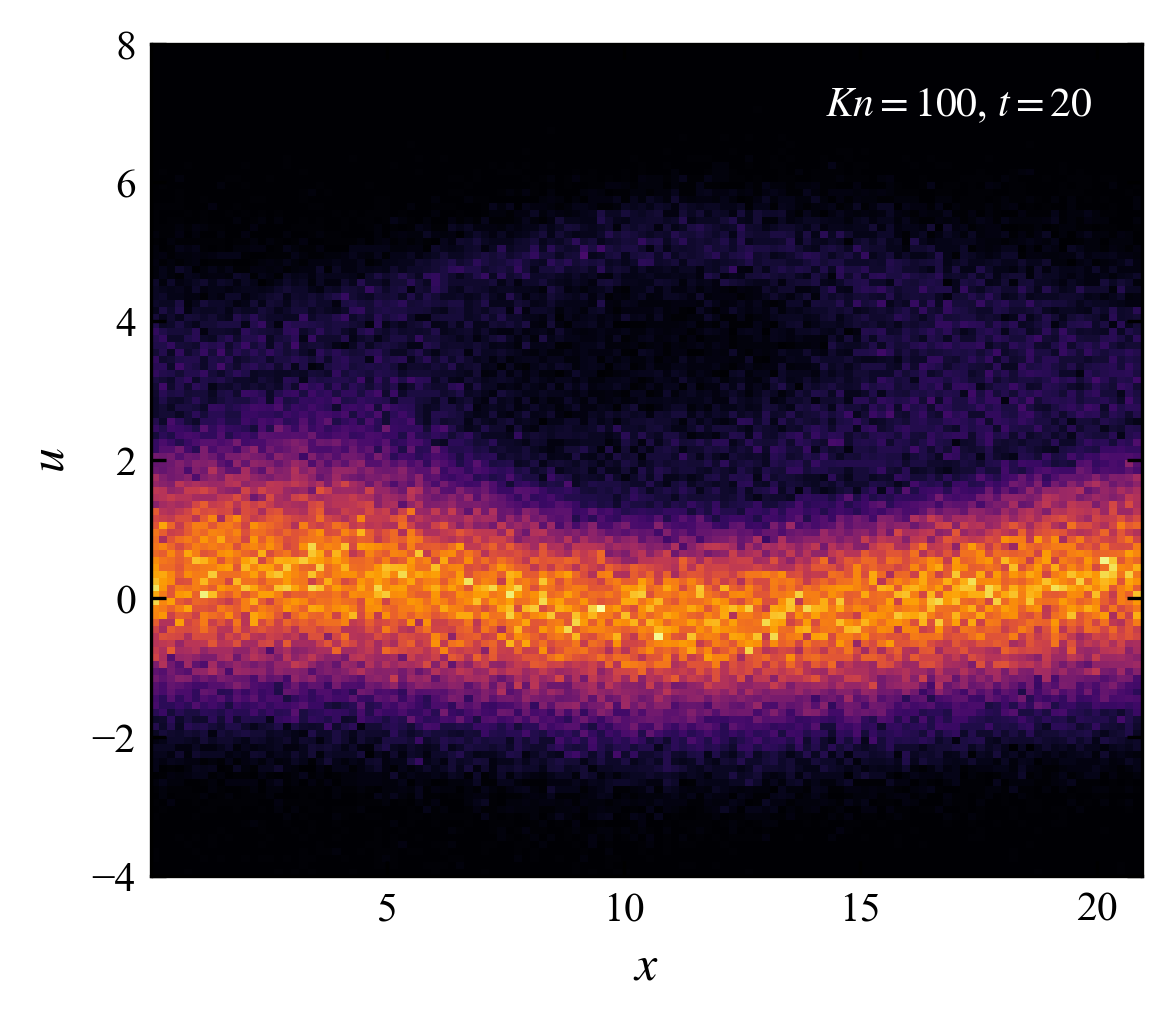}
    \end{subfigure}
    \hfill
    \begin{subfigure}[b]{0.325\textwidth}
    \centering
    \includegraphics[width=1.0\linewidth]{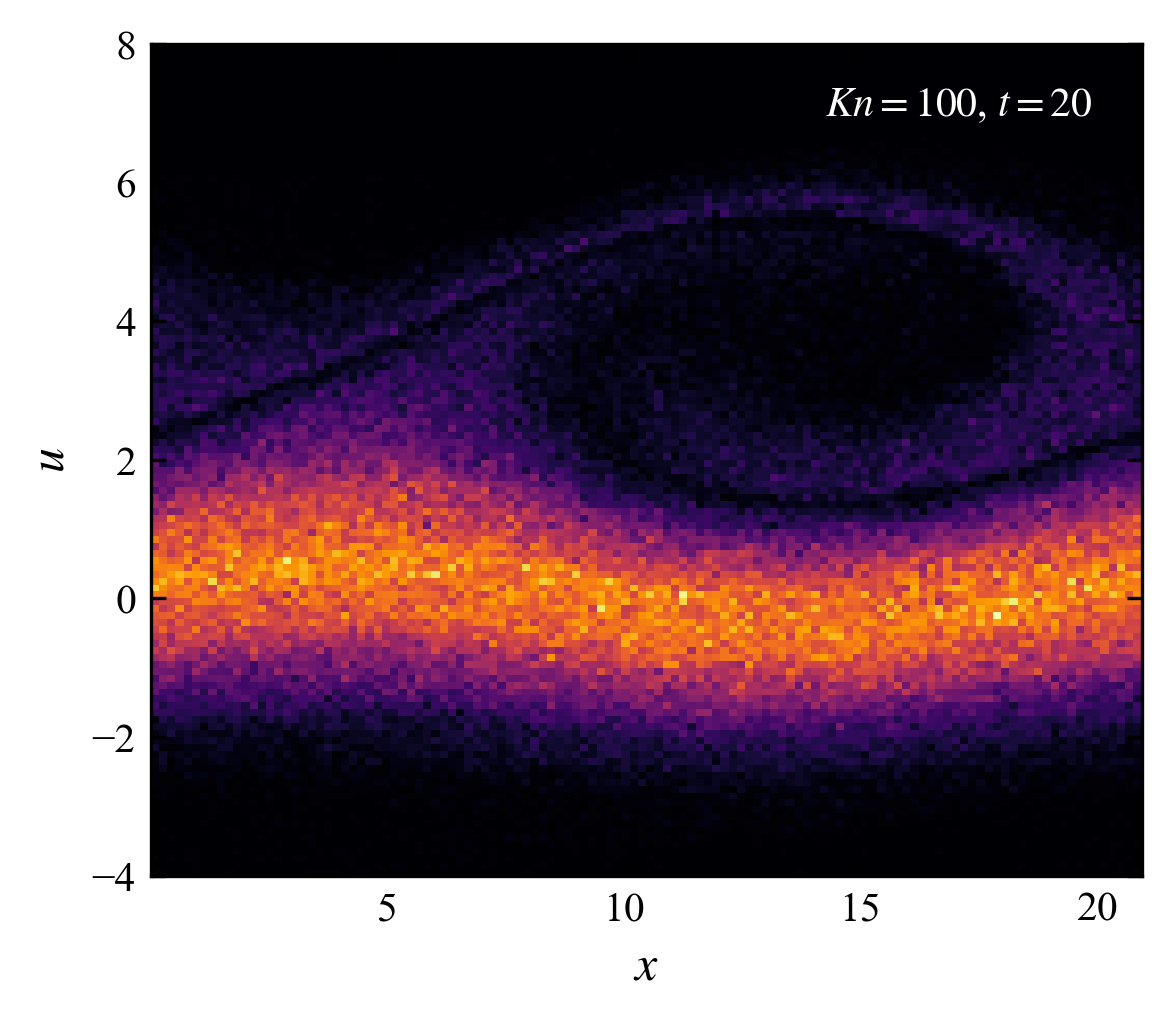}
    \end{subfigure}
    \vfill
    \centering
    \begin{subfigure}[b]{0.325\textwidth}
    \centering
    \includegraphics[width=1.0\linewidth]{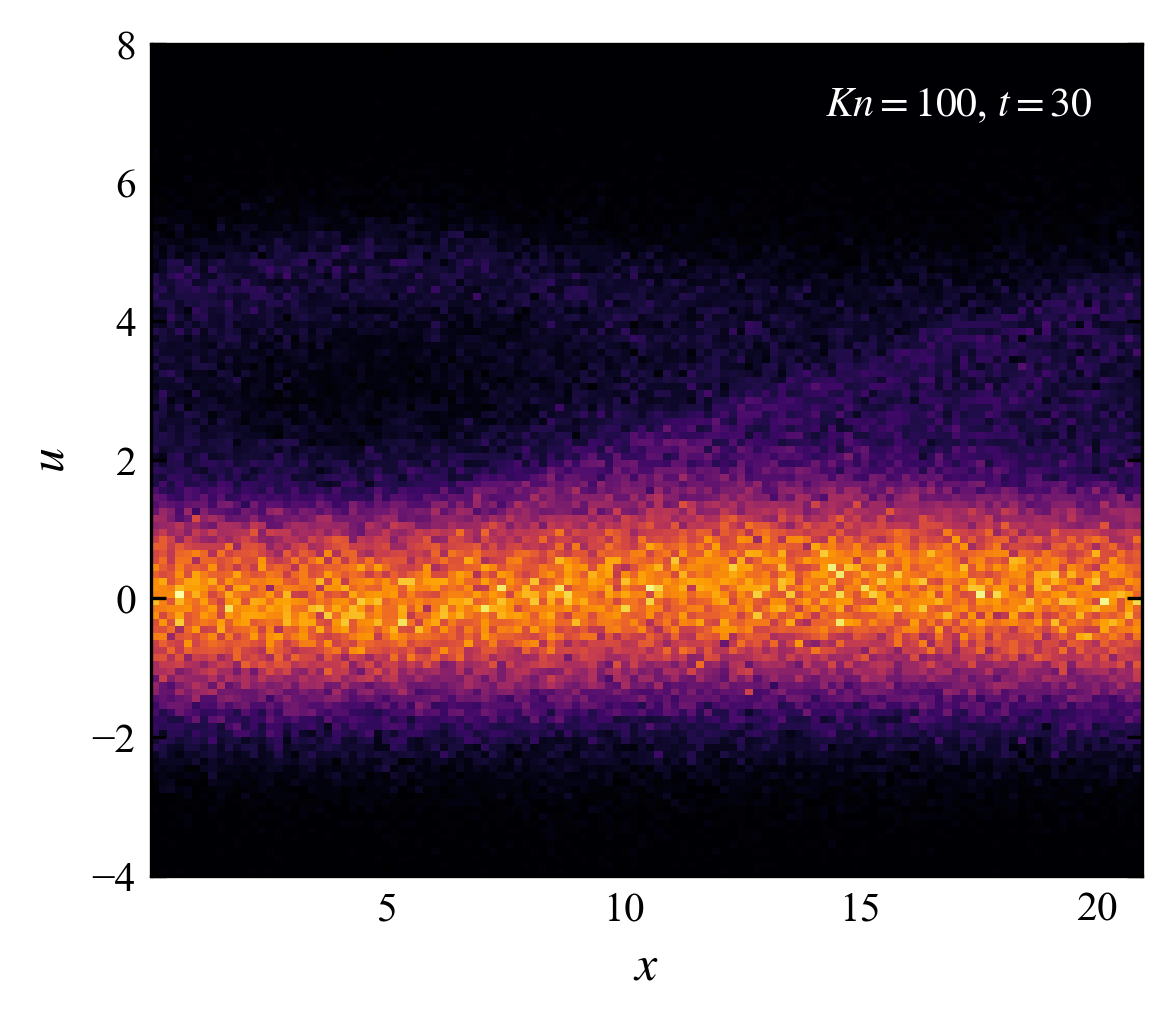}
    \end{subfigure}
    \hfill
    \begin{subfigure}[b]{0.325\textwidth}
    \centering
    \includegraphics[width=1.0\linewidth]{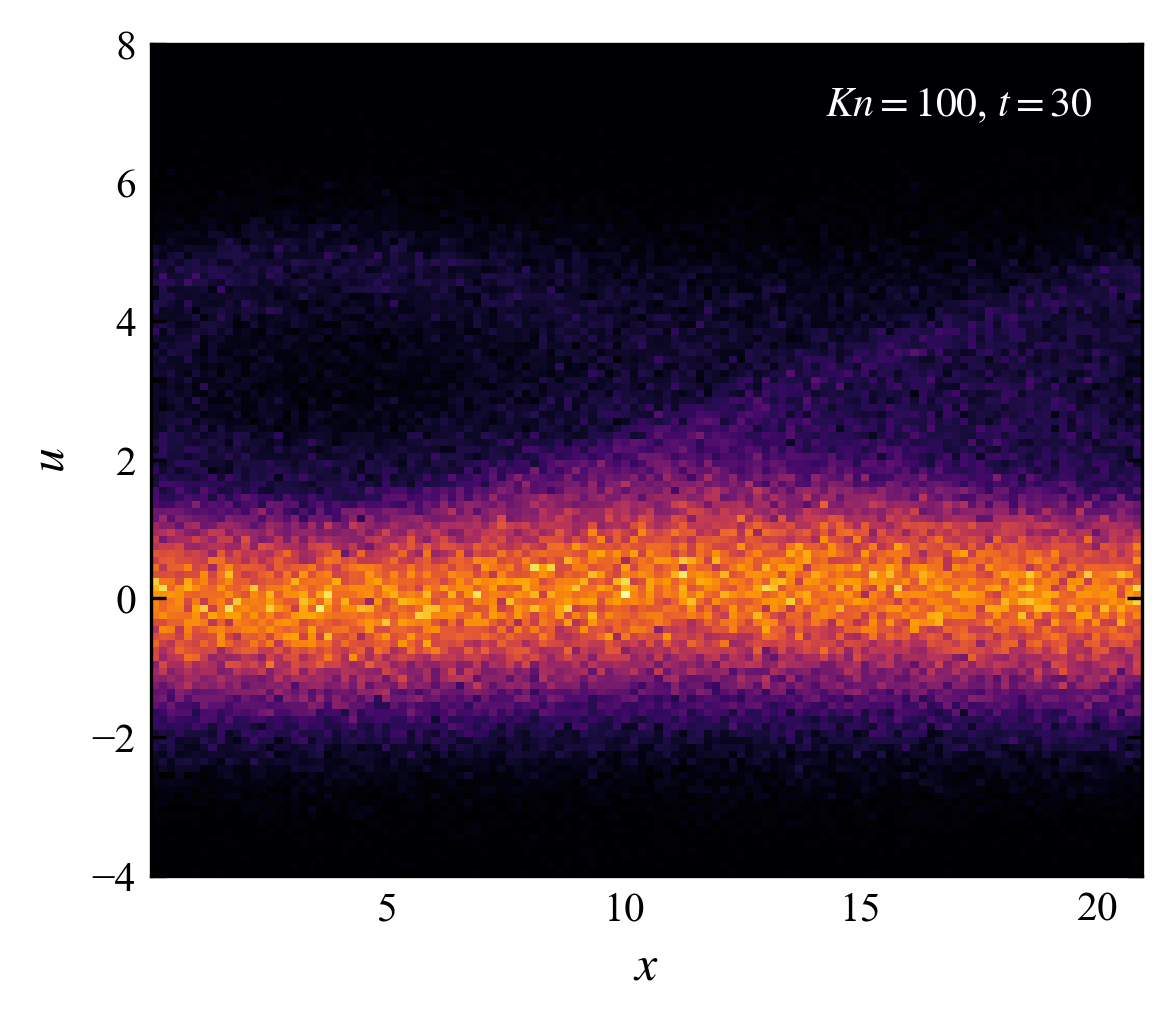}
    \end{subfigure}
    \hfill
    \begin{subfigure}[b]{0.325\textwidth}
    \centering
    \includegraphics[width=1.0\linewidth]{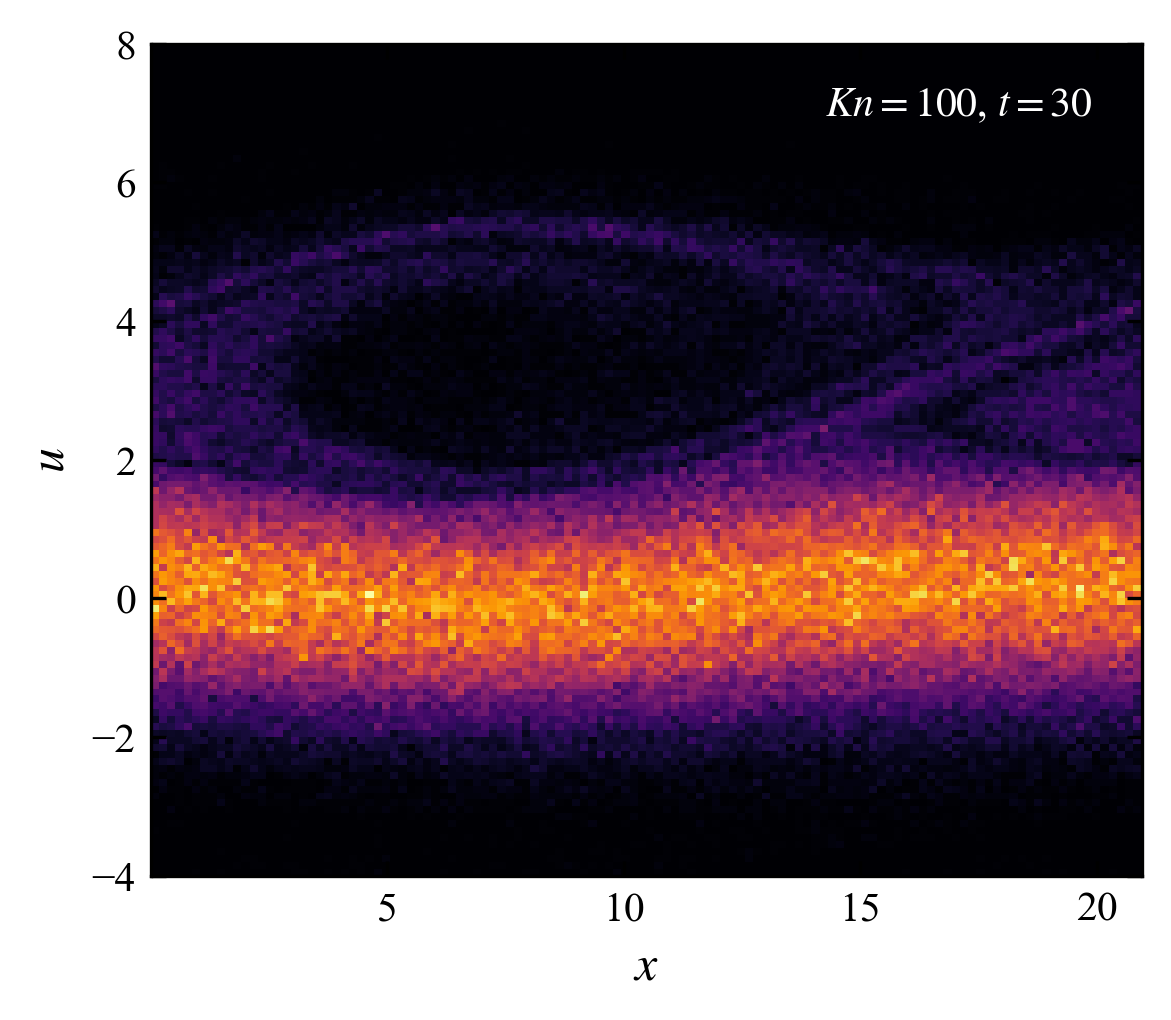}
    \end{subfigure}
    \vfill
    \centering
    \begin{subfigure}[b]{0.325\textwidth}
    \centering
    \includegraphics[width=1.0\linewidth]{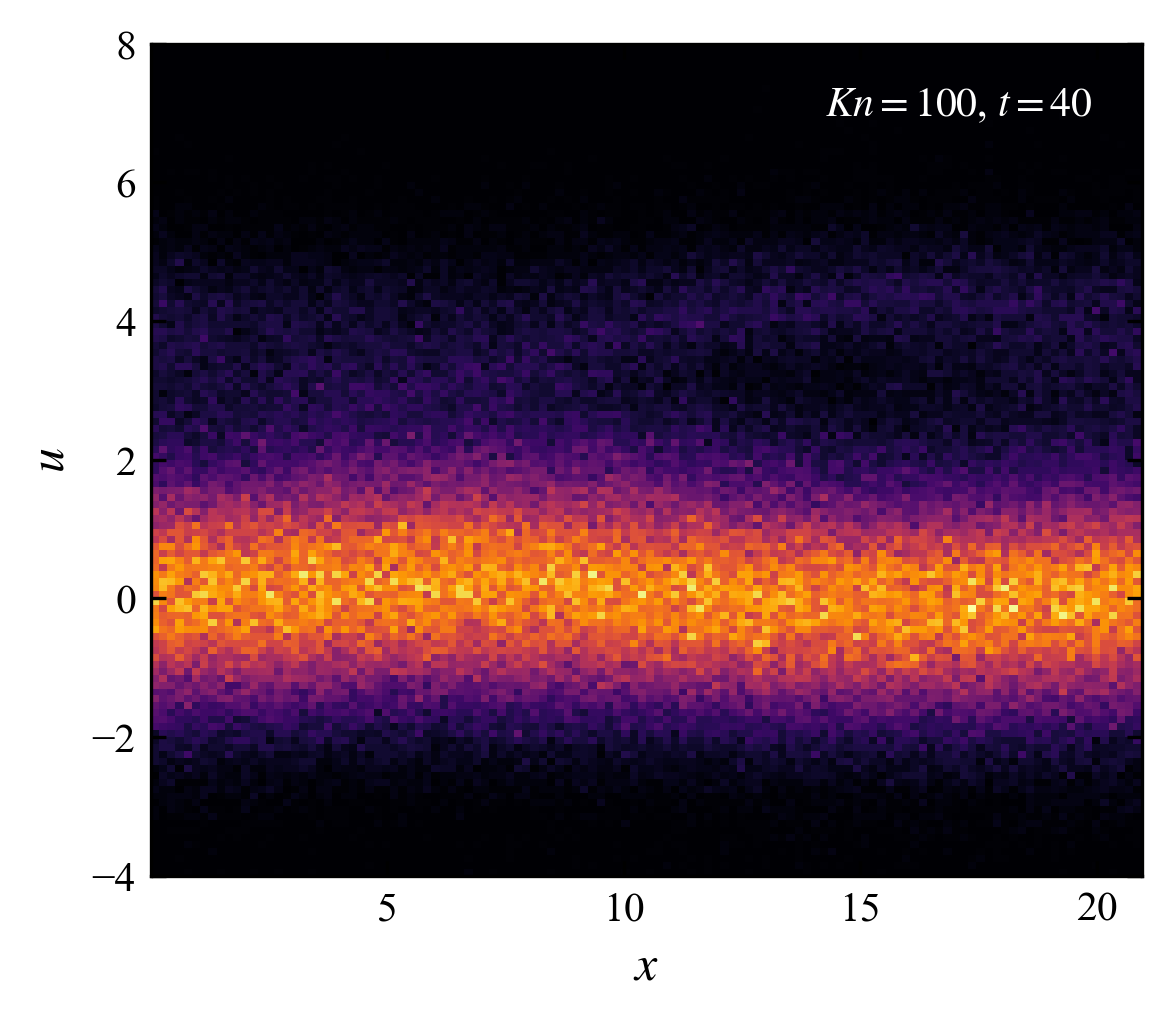}
    \end{subfigure}
    \hfill
    \begin{subfigure}[b]{0.325\textwidth}
    \centering
    \includegraphics[width=1.0\linewidth]{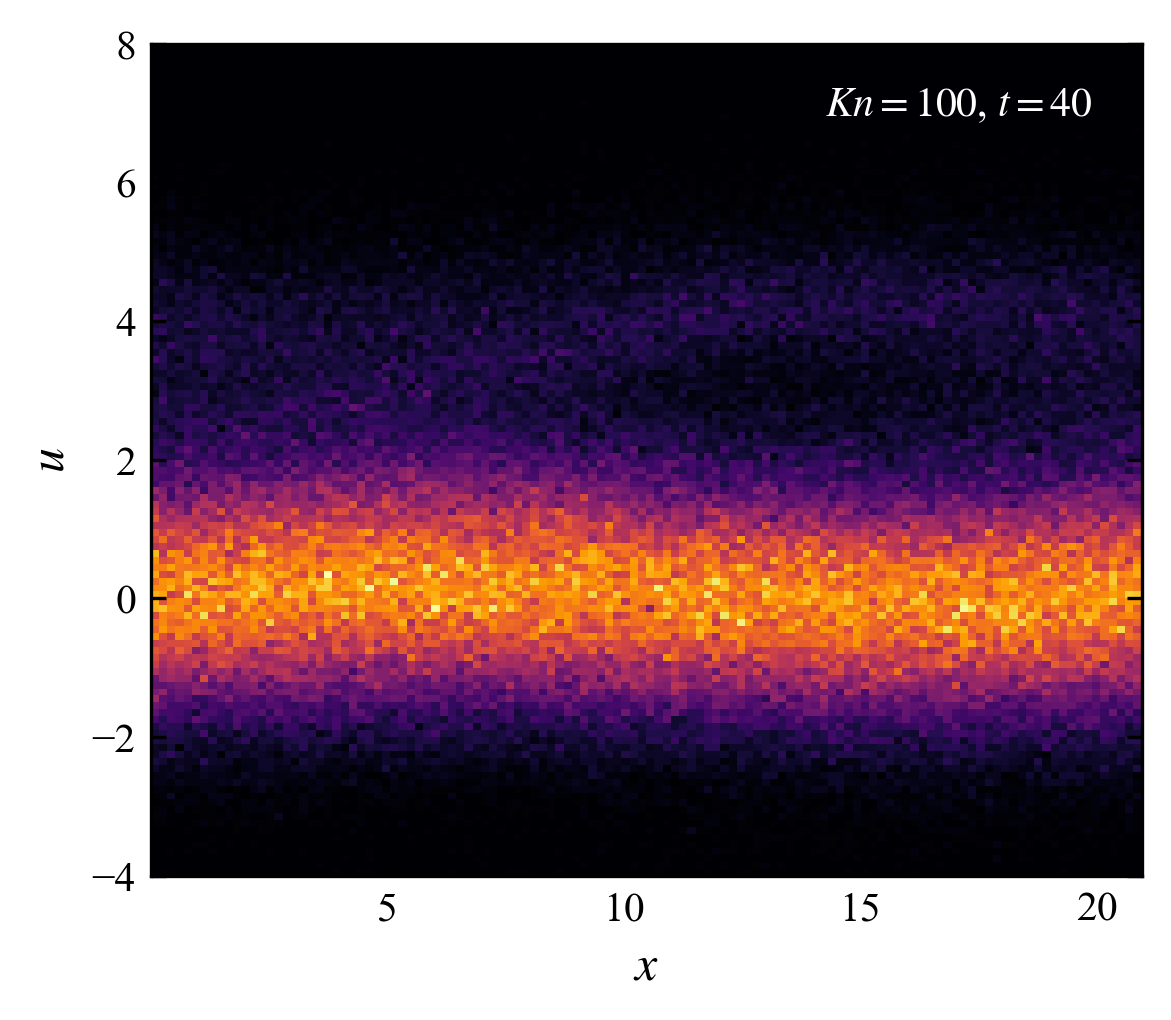}
    \end{subfigure}
    \hfill
    \begin{subfigure}[b]{0.325\textwidth}
    \centering
    \includegraphics[width=1.0\linewidth]{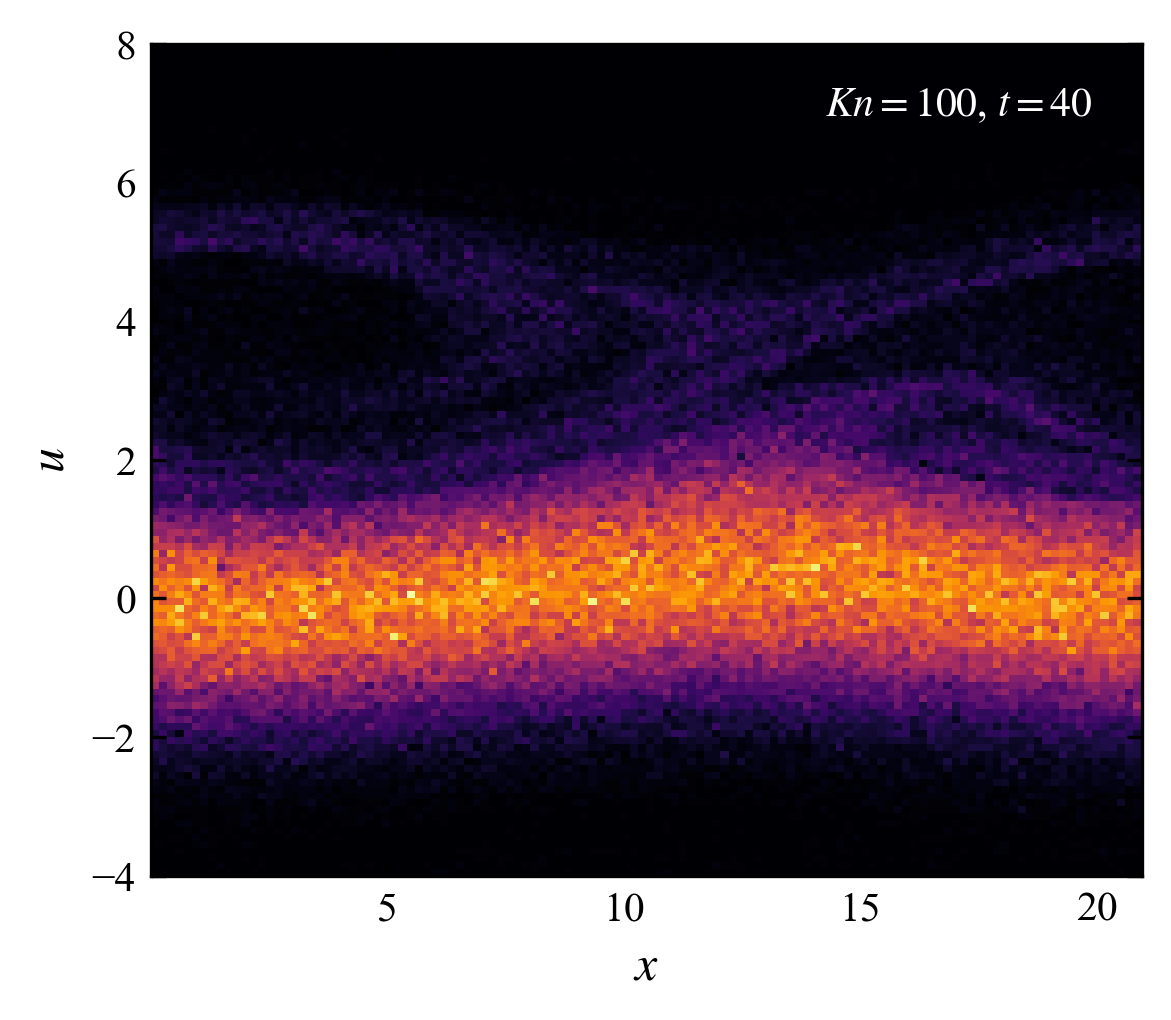}
    \end{subfigure}
	\caption{Snapshots of the phase-space distribution function for the bump-on-tail instability in the highly rarefied regime ($\text{Kn}=100$). The rows, from top to bottom, correspond to different simulation times at $t=10$, $20$, $30$, and $40$. The columns, from left to right, display the results obtained using the PIC-FP, UGKWP-FP, and PIC-BGK methods, respectively.}
    \label{fig:bti-phase-kn100}
\end{figure}

The temporal evolution of the electric field energy and the corresponding phase-space diagrams at $\text{Kn}=100$ are presented in Fig.~\ref{fig:bti-kn100-E} and Fig.~\ref{fig:bti-phase-kn100}, respectively. As demonstrated in both the energy profiles and the microscopic phase-space structures, the proposed UGKWP-FP scheme exhibits good agreement with the reference PIC-FP solution. 

Meanwhile, noticeable discrepancies emerge between the FP and BGK collision models during the nonlinear stage of the instability (after $t \approx 15$). As observed in the rightmost column of Fig.~\ref{fig:bti-phase-kn100}, the PIC-BGK method preserves sharp, high-gradient velocity structures and distinct phase-space vortices. In contrast, the FP operator, governed by continuous velocity-space diffusion, efficiently targets and smooths out these localized steep gradients, causing the fine structures to damp out more rapidly. This fundamental difference in velocity-space dissipation leads to the distinction in the electric energy evolution. 

Furthermore, unlike the purely collisionless limit where nonlinear particle trapping typically sustains a relatively stable energy plateau after saturation, the electric field energy in this case exhibits a gradual decay following the initial saturation peaks. At $\text{Kn}=100$, the macroscopic relaxation time is on the order of $\tau \approx 100$, corresponding to a very low collision frequency of $\nu \approx 0.01$. Despite this sparsity of collisions, their cumulative effect over time gradually degrades the particle trapping mechanism. Because the FP operator diffuses localized velocity fragments more aggressively than the BGK relaxation process, it disrupts the trapped particle orbits within the phase-space vortices more effectively. Consequently, the FP model yields a lower saturation amplitude compared to the BGK model.

\begin{figure}
    \centering
    \begin{subfigure}[b]{0.325\textwidth}
    \centering
    \includegraphics[width=1.0\linewidth]{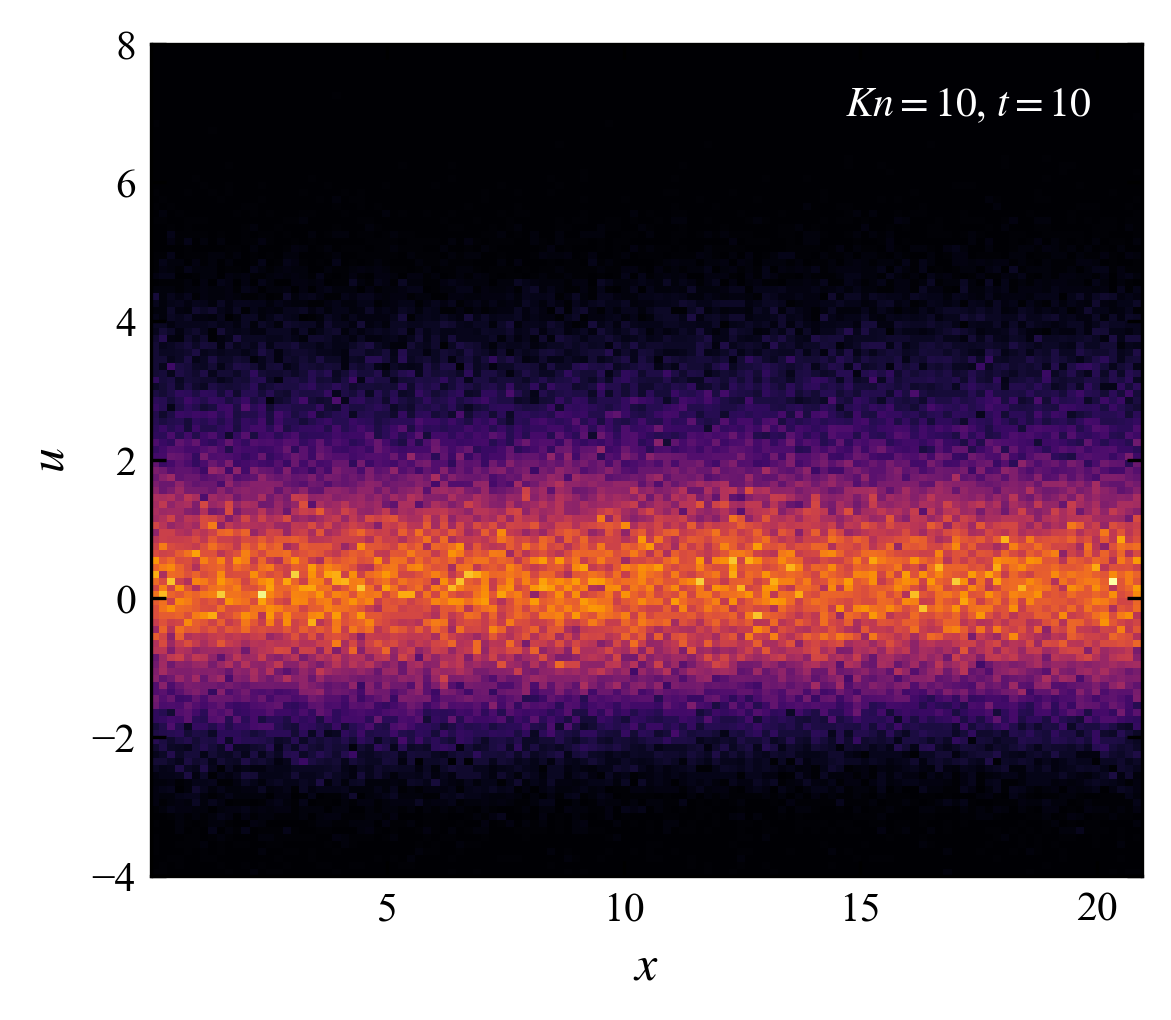}
    \end{subfigure}
    \hfill
    \begin{subfigure}[b]{0.325\textwidth}
    \centering
    \includegraphics[width=1.0\linewidth]{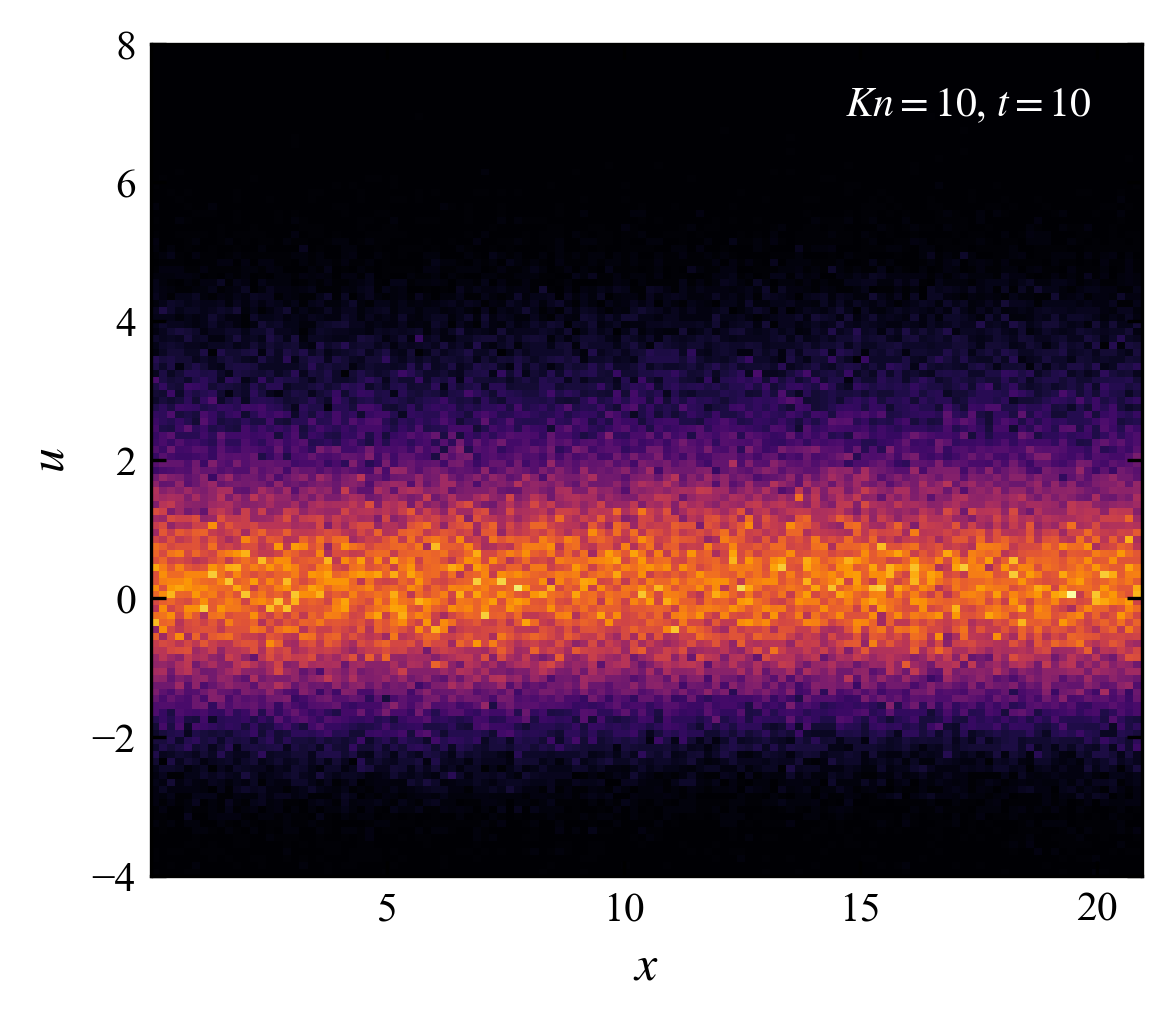}
    \end{subfigure}
    \hfill
    \begin{subfigure}[b]{0.325\textwidth}
    \centering
    \includegraphics[width=1.0\linewidth]{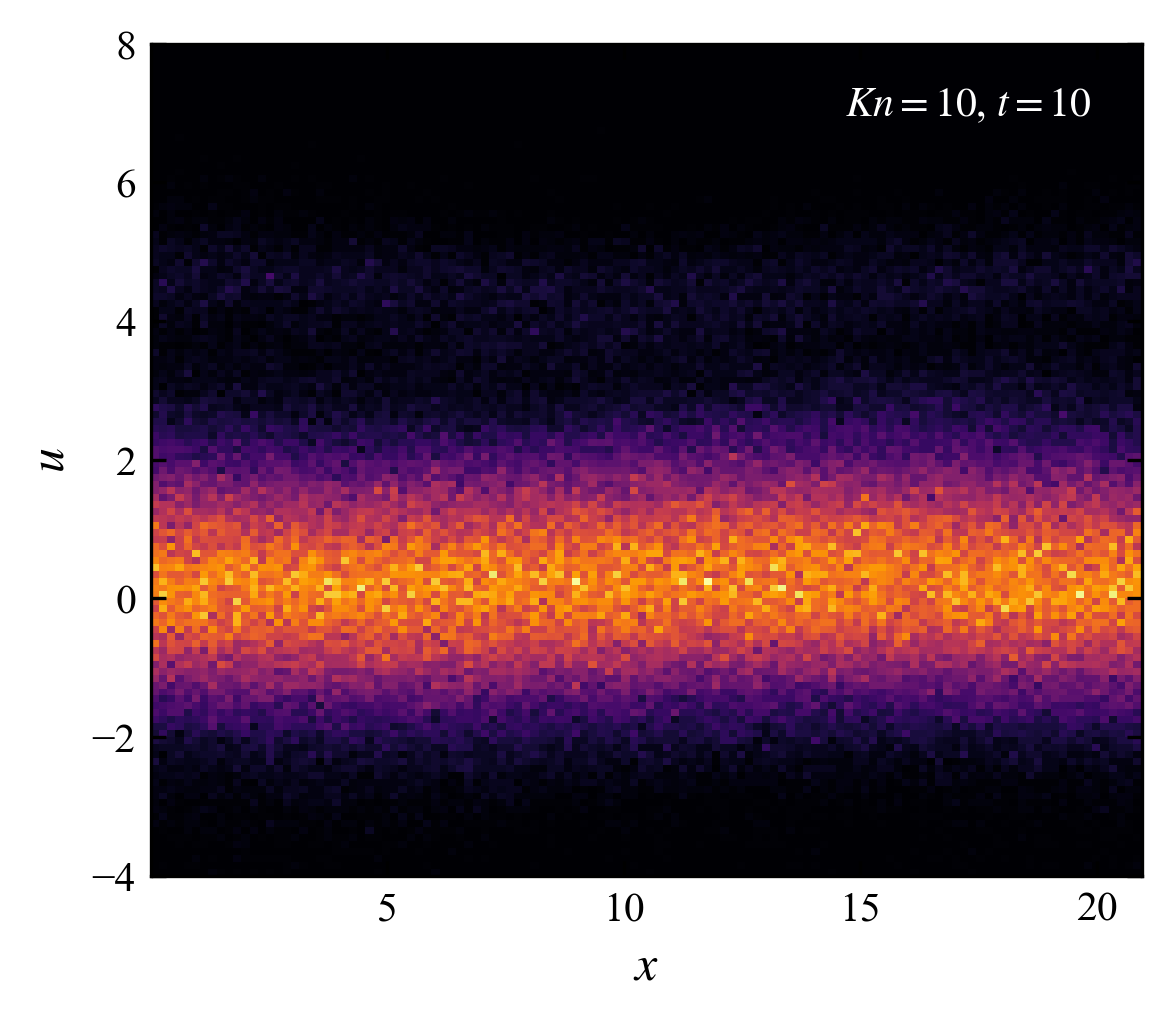}
    \end{subfigure}
    \vfill
    \centering
    \begin{subfigure}[b]{0.325\textwidth}
    \centering
    \includegraphics[width=1.0\linewidth]{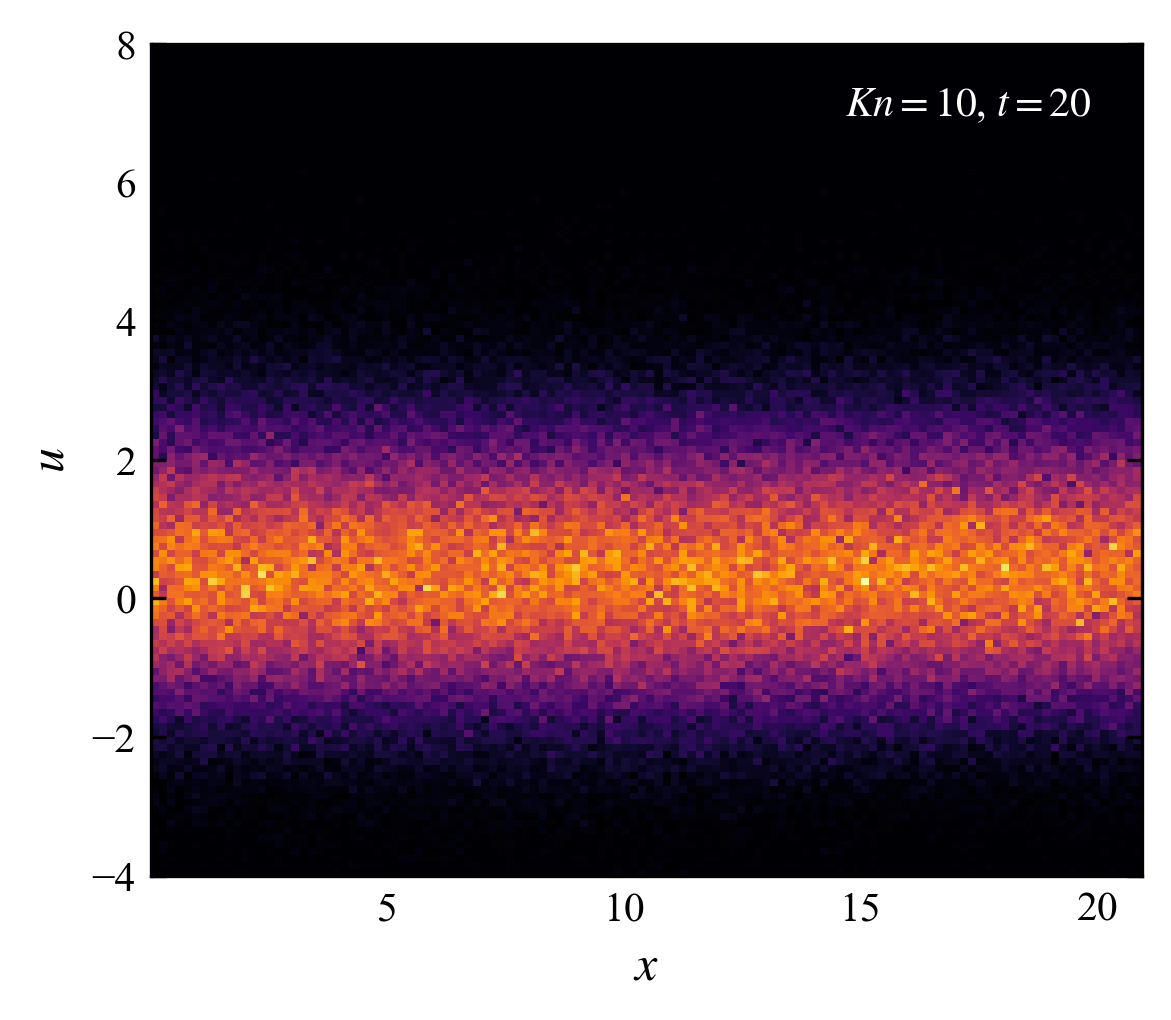}
    \end{subfigure}
    \hfill
    \begin{subfigure}[b]{0.325\textwidth}
    \centering
    \includegraphics[width=1.0\linewidth]{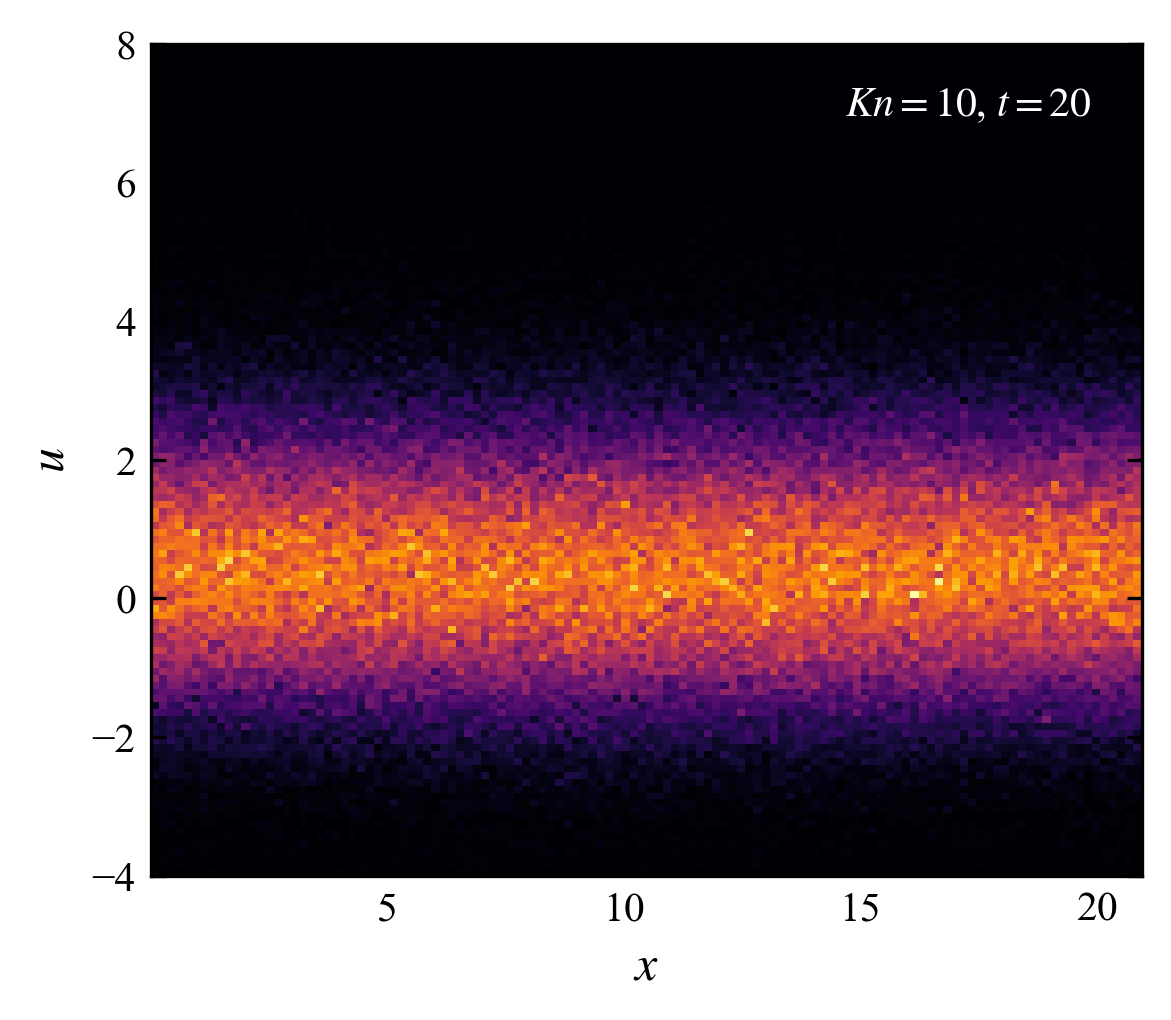}
    \end{subfigure}
    \hfill
    \begin{subfigure}[b]{0.325\textwidth}
    \centering
    \includegraphics[width=1.0\linewidth]{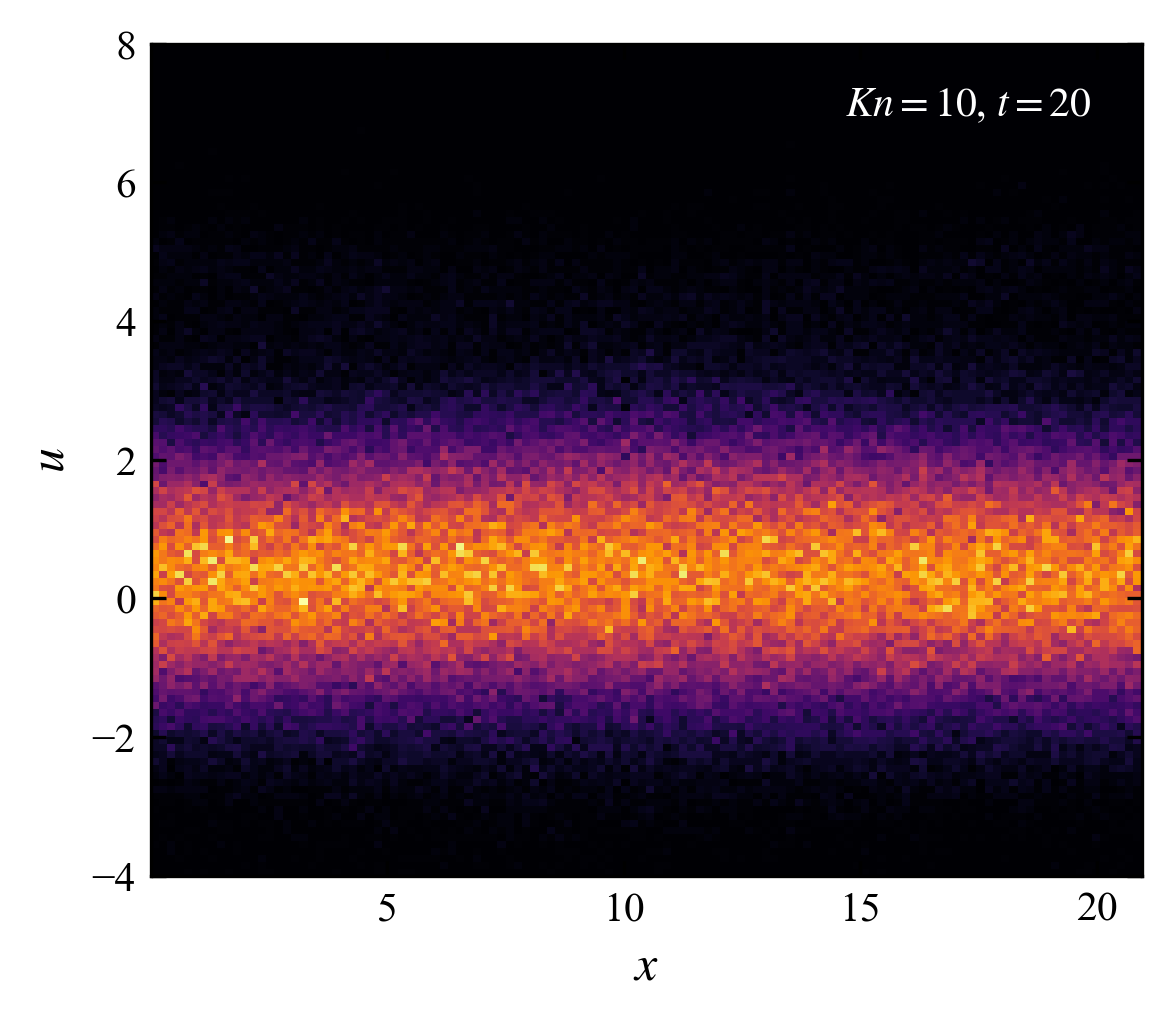}
    \end{subfigure}
	\caption{Snapshots of the phase-space distribution function for the bump-on-tail instability in the highly rarefied regime ($\text{Kn}=10$). The rows, from top to bottom, correspond to different simulation times at $t=10$ and $20$. The columns, from left to right, display the results obtained using the PIC-FP, UGKWP-FP, and PIC-BGK methods, respectively.}
    \label{fig:bti-phase-kn10}
\end{figure}

At $\text{Kn}=10$, the disparity between the FP and BGK collision models emerges at an earlier stage ($t \approx 8$), as illustrated in Fig.~\ref{fig:bti-kn10-E}. Throughout the entire evolution, the proposed UGKWP-FP method maintains close agreement with the reference PIC-FP solution, whereas the PIC-BGK solution exhibits a markedly higher electric field energy. This divergence is directly reflected in the phase-space distributions at $t=10$ (Fig.~\ref{fig:bti-phase-kn10}): the secondary bump structure remains partially preserved in the PIC-BGK model, providing a weak yet persistent drive for the instability. In contrast, the continuous velocity-space diffusion inherent to the FP operator rapidly smooths out this localized non-equilibrium feature. Beyond $t \approx 10$, the bump is progressively eliminated in both models as the velocity distributions relax toward a near-equilibrium state. Nevertheless, the PIC-BGK solution maintains a higher overall energy level into the late stage, this occurs because the prolonged survival of the velocity bump during the preceding phase sustains the field growth, leaving a larger baseline energy amplitude at the onset of the damping-dominated regime.

\begin{figure}
    \centering
    \begin{subfigure}[b]{0.48\textwidth}
    \centering
    \includegraphics[width=1.0\linewidth]{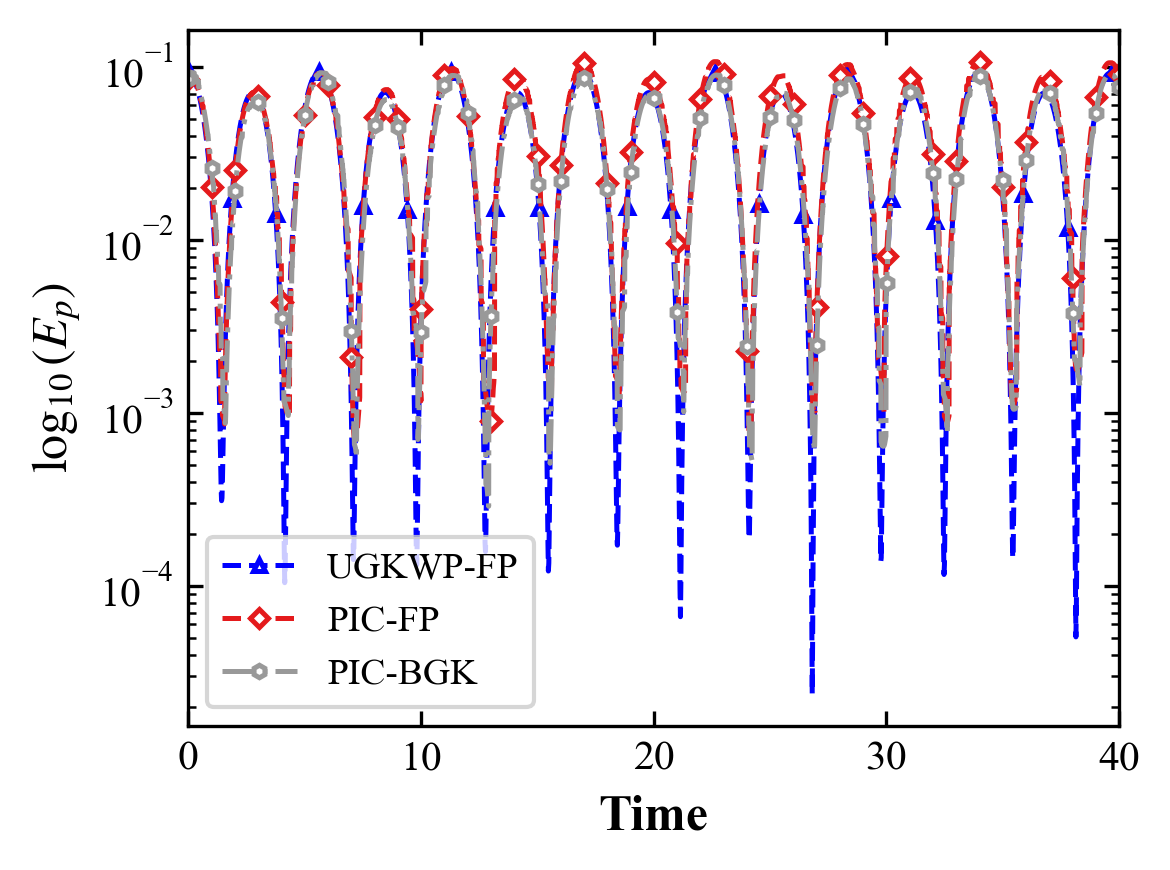}
    \caption{}
    \label{fig:bti-kn0.001-a}
    \end{subfigure}
    \hfill
    \begin{subfigure}[b]{0.48\textwidth}
    \centering
    \includegraphics[width=1.0\linewidth]{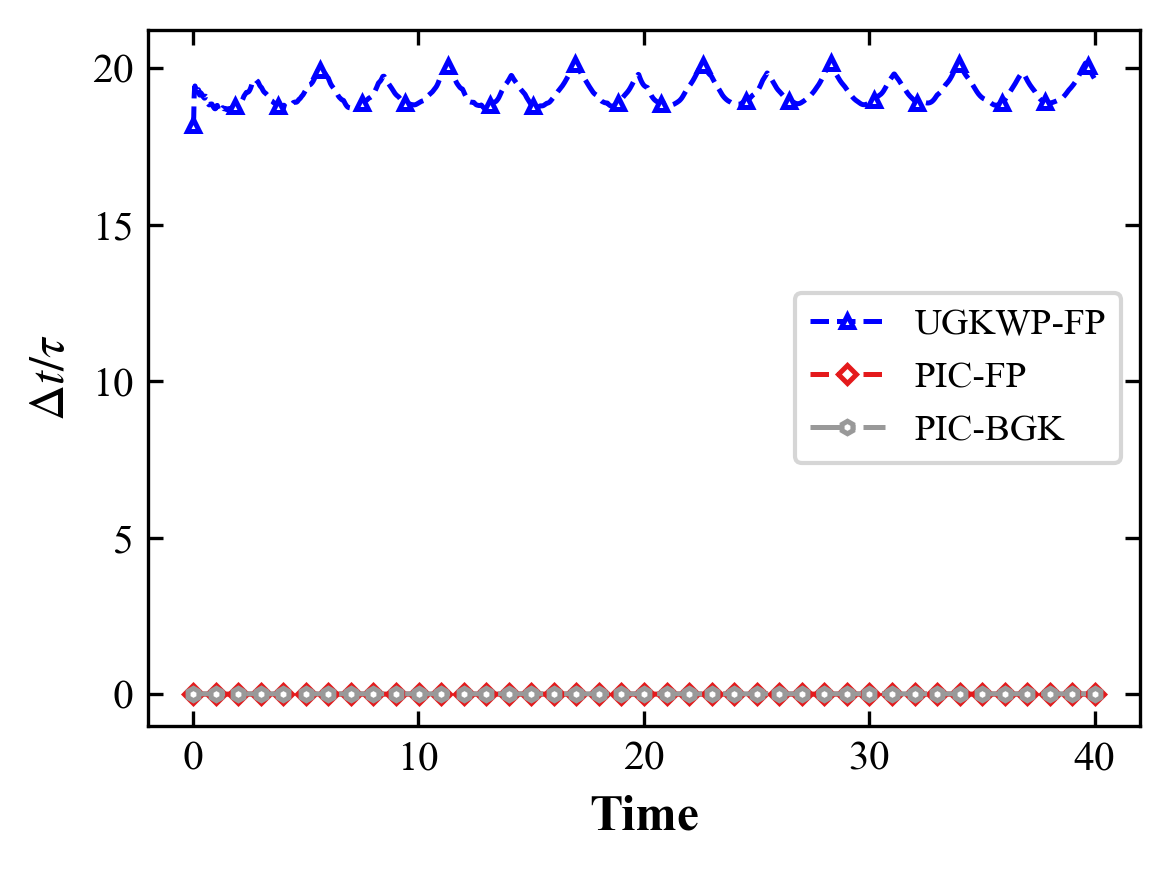}
    \caption{}
    \label{fig:bti-kn0.001-b}
    \end{subfigure}
    \caption{Temporal evolution of (a) the electric field energy and (b) the normalized time step $\Delta t/\tau$ for bump-on-tail instability in the near-continuum regime at $\text{Kn} = 0.001$. Results are compared among the UGKWP-FP, PIC-FP, and PIC-BGK methods.}
    \label{fig:bti-kn0.001}
\end{figure}

In the near-continuum regime $\text{Kn} = 0.001$, strong collisional relaxation
 rapidly damp the bump, completely suppressing the instability. As shown in Fig.~\ref{fig:bti-kn0.001-a}, the electric field energy no longer exhibits exponential growth or nonlinear particle trapping; instead, it transitions into oscillations around a stable amplitude, similar to the behavior observed in collisional Landau damping. The energy profile predicted by the UGKWP-FP scheme agrees well with the reference PIC-FP solution, confirming that the multiscale framework accurately captures the collision-dominated plasma dynamics.

Fig.~\ref{fig:bti-kn0.001-b} highlights the computational efficiency and asymptotic-preserving property of the UGKWP-FP method through the normalized time step, $\Delta t/\tau$. Conventional particle solvers (PIC-FP and PIC-BGK) are severely restricted by the stiff kinetic scale, requiring $\Delta t \le \tau \sim 0.001$ and grid resolutions fine enough to resolve the mean free path $\Delta x \le l_{mfp} \sim 0.001$, which leads to prohibitive computational costs. In contrast, UGKWP-FP operates stably with a time step up to twenty times larger than the collision time $\Delta t / \tau \approx 20$ on a coarse spatial grid ($\Delta x = 0.08 \gg l_{\text{mfp}}$), significantly reducing the computational burden while correctly recovering the continuum-limit physics.

\subsection{Sod shock tube}
\label{sec:sod}

To evaluate the asymptotic-preserving property and shock-capturing capability of the UGKWP-FP method across different flow regimes, the Sod shock tube problem is calculated. The electric field is not considered in this case. The spatial domain is set to $x \in [0, 1]$ with a spatial resolution of $\Delta x = 0.01$. In this case, $\omega=\omega_{ref} = 0.81$. The initial macroscopic state is initialized as:
\begin{equation}
(\rho, u, p) = \begin{cases} 
(1.0, \, 0, \, 1.0), & x \le 0.5, \\ 
(0.125, \, 0, \, 0.1), & x > 0.5. 
\end{cases}
\end{equation}
Three representative Knudsen numbers are simulated up to $t = 0.15$: $\text{Kn} = 0.1$, $\text{Kn} = 0.001$, and $\text{Kn} = 10^{-5}$. The reference solutions are provided by the UGKS based on the BGK model.

\begin{figure}
    \centering
    \begin{subfigure}[b]{0.32\textwidth}
    \centering
    \includegraphics[width=1.0\linewidth]{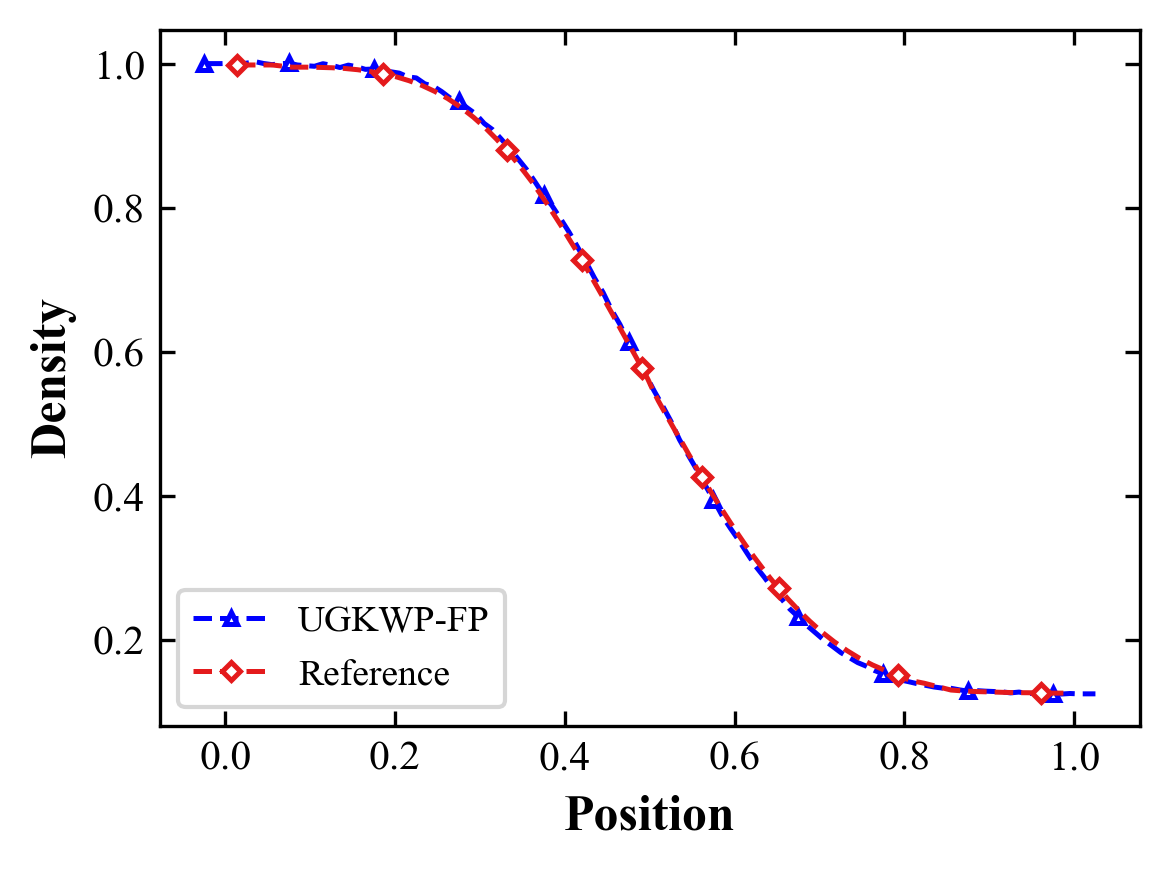}
    \caption{}
    \label{fig:sod-kn0.1-a}
    \end{subfigure}
    \hfill
    \begin{subfigure}[b]{0.32\textwidth}
    \centering
    \includegraphics[width=1.0\linewidth]{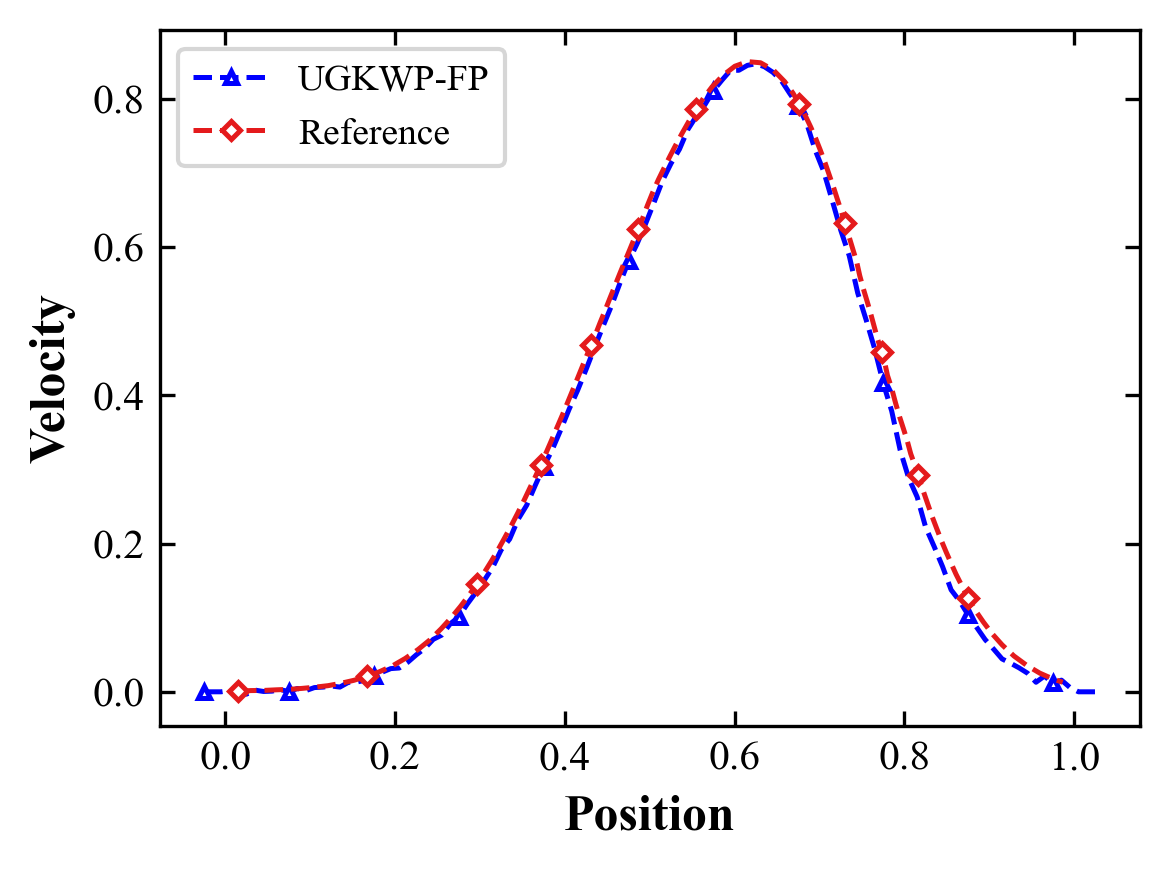}
    \caption{}
    \label{fig:sod-kn0.1-b}
    \end{subfigure}
    \hfill
    \begin{subfigure}[b]{0.32\textwidth}
    \centering
    \includegraphics[width=1.0\linewidth]{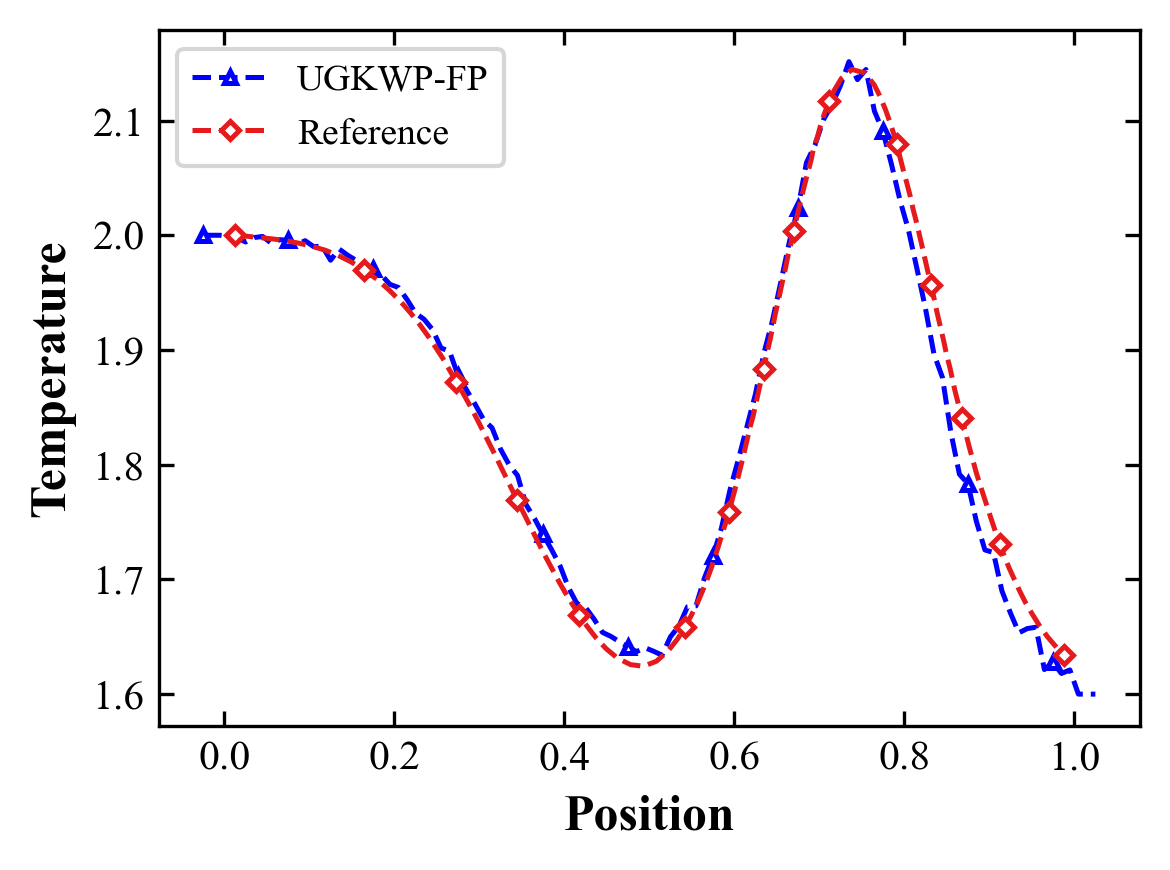}
    \caption{}
    \label{fig:sod-kn0.1-c}
    \end{subfigure}
    \caption{(a) Density, (b) velocity, (c) temperature profiles of Sod shock tube at t=0.15 with Kn=0.1. $\Delta x = 0.01$, $\Delta t = 0.005 \sim 0.01 \tau$. The reference solution is calculated by UGKS with BGK collision.}
    \label{fig:sod-kn0.1}
\end{figure}

\begin{figure}
    \centering
    \begin{subfigure}[b]{0.32\textwidth}
    \centering
    \includegraphics[width=1.0\linewidth]{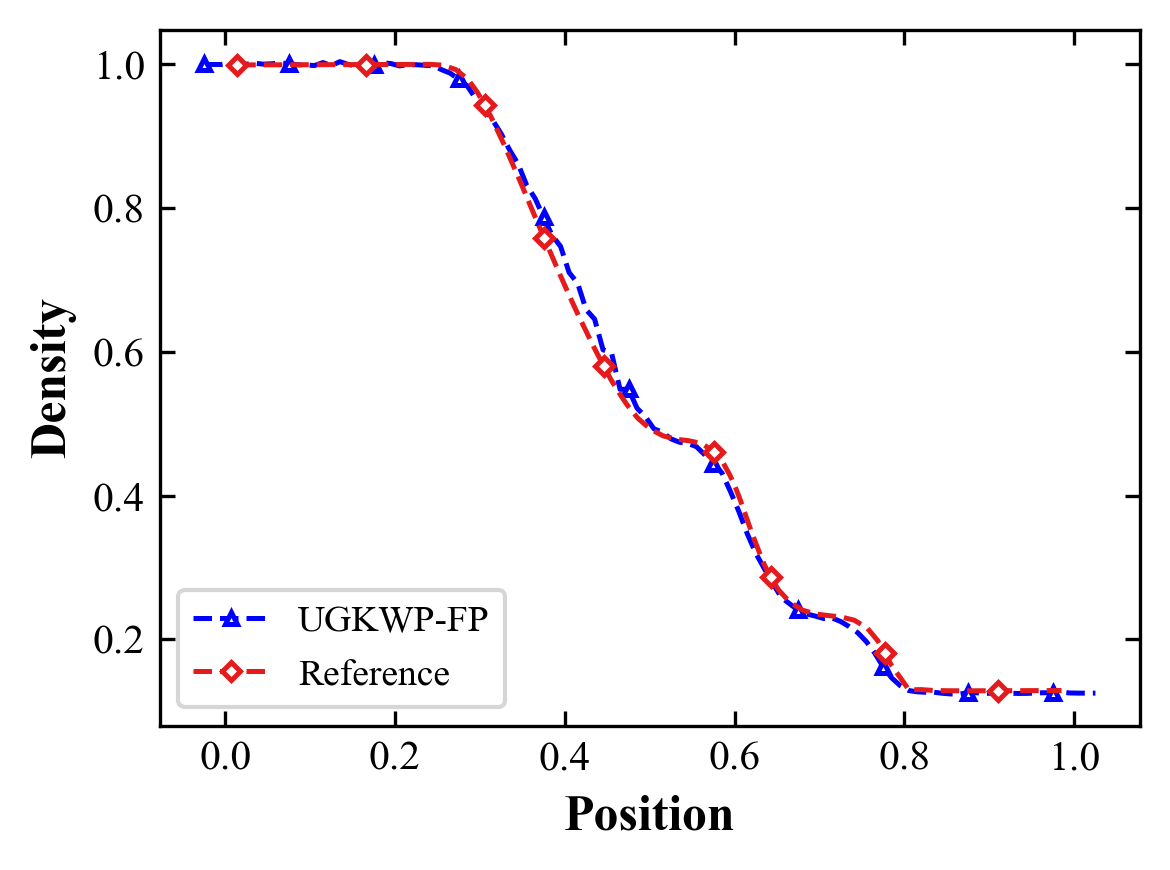}
    \caption{}
    \label{fig:sod-kn0.001-a}
    \end{subfigure}
    \hfill
    \begin{subfigure}[b]{0.32\textwidth}
    \centering
    \includegraphics[width=1.0\linewidth]{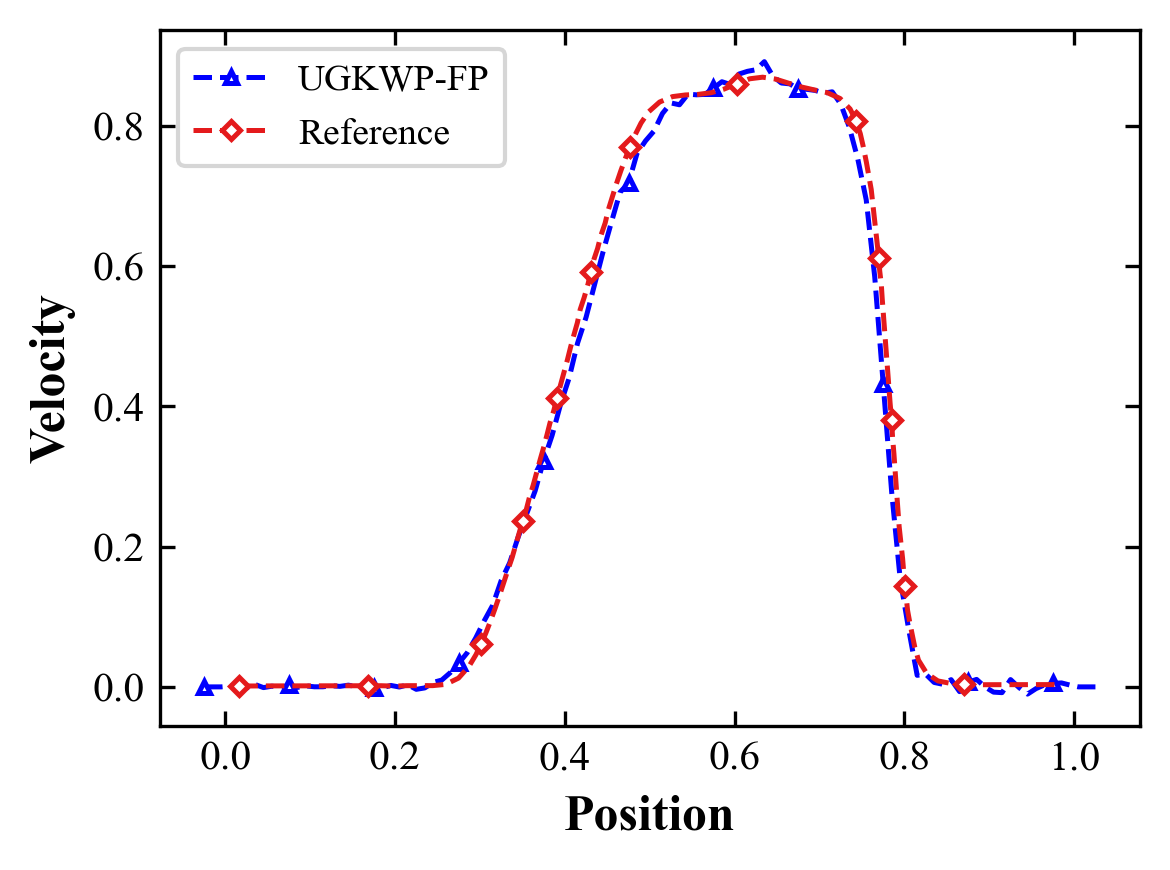}
    \caption{}
    \label{fig:sod-kn0.001-b}
    \end{subfigure}
    \hfill
    \begin{subfigure}[b]{0.32\textwidth}
    \centering
    \includegraphics[width=1.0\linewidth]{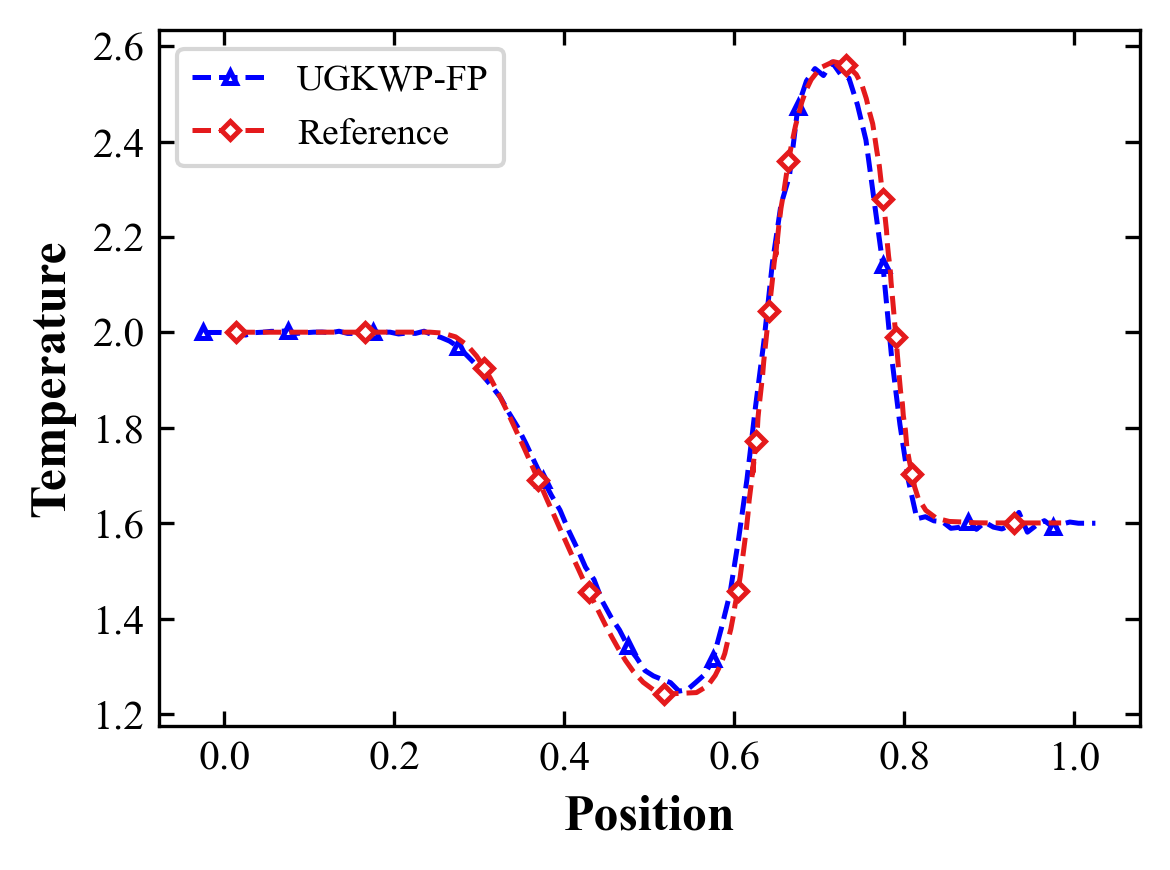}
    \caption{}
    \label{fig:sod-kn0.001-c}
    \end{subfigure}
    \caption{(a) Density, (b) velocity, (c) temperature profiles of Sod shock tube at t=0.15 with Kn=0.001. $\Delta x = 0.01$, $\Delta t = 0.005 \sim \tau$. The reference solution is calculated by UGKS with BGK collision.}
    \label{fig:sod-kn0.001}
\end{figure}

\begin{figure}
    \centering
    \begin{subfigure}[b]{0.32\textwidth}
    \centering
    \includegraphics[width=1.0\linewidth]{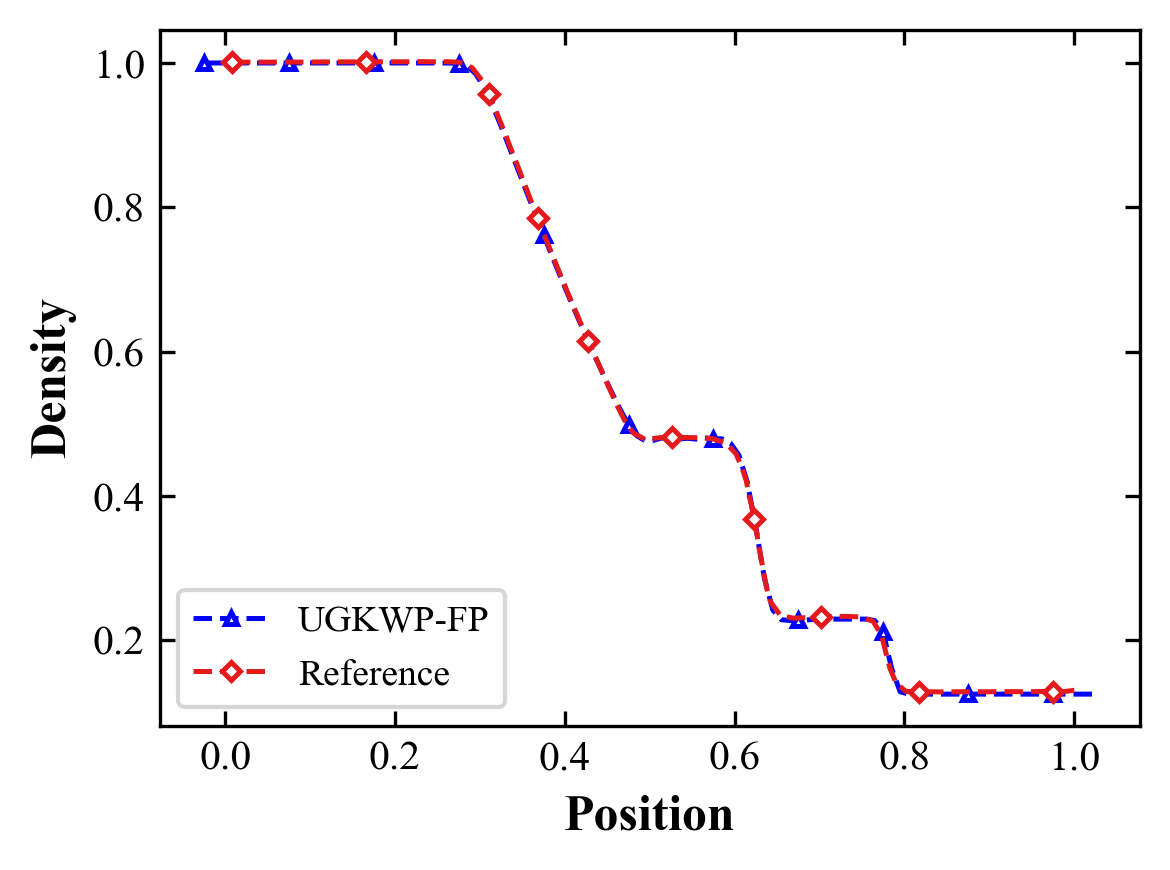}
    \caption{}
    \label{fig:sod-kn0.00001-a}
    \end{subfigure}
    \hfill
    \begin{subfigure}[b]{0.32\textwidth}
    \centering
    \includegraphics[width=1.0\linewidth]{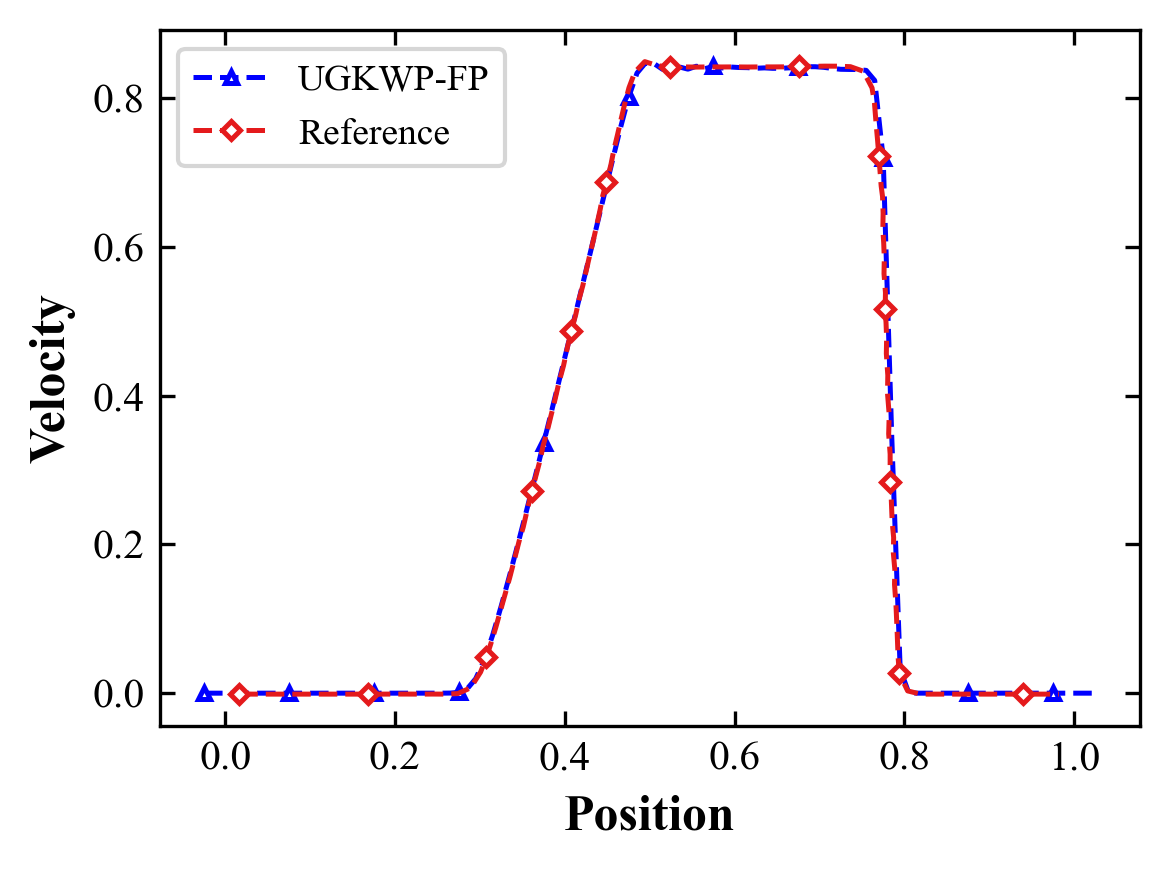}
    \caption{}
    \label{fig:sod-kn0.00001-b}
    \end{subfigure}
    \hfill
    \begin{subfigure}[b]{0.32\textwidth}
    \centering
    \includegraphics[width=1.0\linewidth]{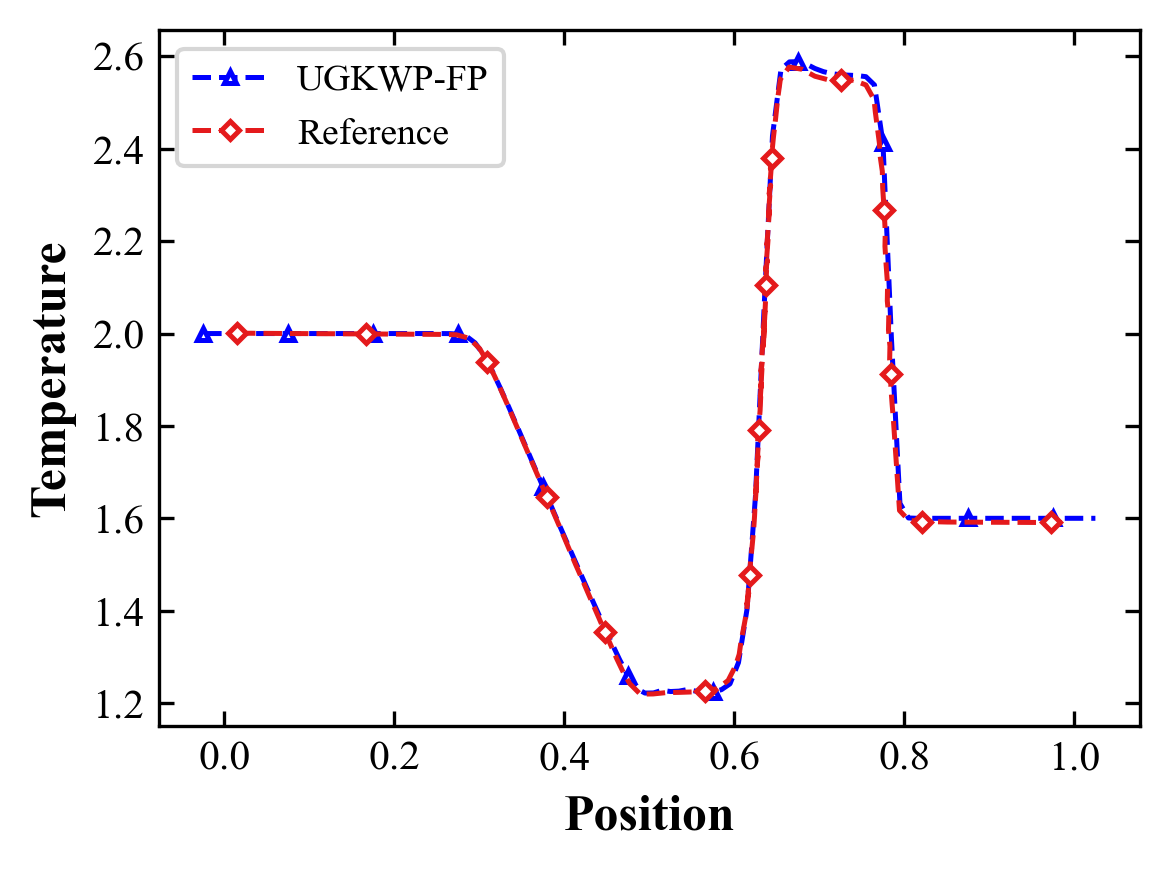}
    \caption{}
    \label{fig:sod-kn0.00001-c}
    \end{subfigure}
    \caption{(a) Density, (b) velocity, (c) temperature profiles of Sod shock tube at t=0.15 with Kn=0.00001. $\Delta x = 0.01$, $\Delta t = 0.005 \sim 100 \tau$. The reference solution is calculated by UGKS with BGK collision.}
    \label{fig:sod-kn0.00001}
\end{figure}

Figures~\ref{fig:sod-kn0.1} and \ref{fig:sod-kn0.001} present the macroscopic density, velocity, and temperature profiles at $\text{Kn} = 0.1$ and $\text{Kn} = 0.001$, respectively. The UGKWP-FP results show good overall agreement with the reference UGKS solutions across both regimes. Minor disparities appear in the temperature and velocity profiles within non-equilibrium regions. These differences are physically expected: the reference solver uses the BGK model, whereas UGKWP-FP incorporates the Fokker--Planck operator.

Figure~\ref{fig:sod-kn0.00001} illustrates the solution in the continuum limit at $\text{Kn} = 10^{-5}$. In this regime, the numerical time step is two orders of magnitude larger than the microscopic collision time ($\Delta t \sim 100\tau$). The UGKWP-FP profiles agree well with the reference UGKS solution, accurately resolving the expansion wave, contact discontinuity, and shock front. This agreement demonstrates two key capabilities of the method: first, it satisfies the asymptotic-preserving property by correctly recovering the hydrodynamic limit under under-resolved kinetic scales ($\Delta t \gg \tau$, $\Delta x \gg l_{\text{mfp}}$); second, it exhibits robust shock-capturing performance for strong nonlinear discontinuities without introducing numerical oscillations.

\section{Conclusions}
\label{sec:conclusions}

In this work, UGKWP-FP has been developed for the VPFP system. By incorporating the LB collision operator, the method explicitly captures velocity-space drift and diffusion. This capability provides a more explicit representation of velocity-space drift and diffusion than the BGK model. The scheme decomposes the FP operator into two components: a nonstiff drift--diffusion contribution, simulated via a modified Langevin stochastic process, and a stiff thermalization contribution, represented by an analytical hydrodynamic wave. A modified relaxation rate ensures that the particle dynamics recover kinetic transport in the rarefied limit while preserving the correct hydrodynamic behavior under strong collisions. Numerical validations include homogeneous relaxation, nonlinear Landau damping, and the bump-on-tail instability. These benchmarks confirm that UGKWP-FP accurately captures velocity-space gradient dissipation in rarefied flows and recovers the expected equilibrium behavior in the continuum limit. Furthermore, the Sod shock tube test confirms the robust shock-capturing performance and asymptotic recovery of the hydrodynamic limit across disparate Knudsen numbers.

Furthermore, the method demonstrates asymptotic-preserving properties in near-continuum regimes. Conventional particle solvers require time steps and spatial grids strictly bounded by microscopic collision scales ($\Delta t \le \tau$ and $\Delta x \le l_{\text{mfp}}$). In contrast, UGKWP-FP operates accurately and stably with macroscopic time steps ($\Delta t \gg \tau$) and coarse spatial meshes ($\Delta x \gg l_{\text{mfp}}$). As collisionality increases, the scheme dynamically transfers the phase-space representation from stochastic particles to analytical equilibrium waves. This adaptive approach substantially reduces computational and memory overhead in strongly collisional regions. Overall, the UGKWP-FP framework provides an efficient and physically consistent computational tool for simulating multiscale plasma dynamics across disparate collisional regimes.

\appendix
\section{Viscosity Matching for the FP--BGK Closure}
\label{app:viscosity-matching}

In the strongly collisional regime, the particle-resolved collision
frequency satisfies \(\nu_p\ll\nu\). Therefore, the BGK relaxation
time is approximated by
\begin{equation}
    \tau_{B}
    =
    \frac{C}{\nu-\nu_p}
    \simeq
    \frac{C}{\nu}
    =
    C\tau,
    \qquad
    \tau=\frac{1}{\nu}.
    \label{eq:tauc-strong-collision}
\end{equation}
For the BGK model, the Chapman--Enskog expansion gives the relation
between the dynamic viscosity \(\mu\), the thermodynamic pressure \(p\),
and the relaxation time \cite{xu2001gas},
\begin{equation}
    \mu=p\tau_{B}.
    \label{eq:bgk-viscosity}
\end{equation}
For the Fokker--Planck model, the corresponding Chapman--Enskog result \cite{mathiaud2016fokker, liu2019conservative} is
\begin{equation}
    \mu=\frac{1}{2}p\tau.
    \label{eq:fp-viscosity}
\end{equation}
Equating the two viscosity expressions yields
\begin{equation}
    \tau_{B}=\frac{1}{2}\tau.
    \label{eq:tau-matching}
\end{equation}
Comparing Eq.~\eqref{eq:tauc-strong-collision} with
Eq.~\eqref{eq:tau-matching} gives the viscosity-matching choice
\begin{equation}
    C=\frac{1}{2}.
\end{equation}

Since the LB operator relaxes the stress tensor at rate $2\nu_p$ \cite{gorji2011fokker}, and the BGK term at rate
$1/\tau_B = 2(\nu - \nu_p)$, the total stress relaxation rate is $2\nu$ for
any admissible $\nu_p$, so the viscosity of the FP-BGK model is in
fact uniformly consistent with the original LB operator. 

\section{Algorithms for PIC-FP and PIC-BGK}
\label{sec: pic}

To evaluate numerical performance, conventional PIC collision schemes are implemented using a second-order Strang-splitting strategy. Depending on the collision operator employed, these schemes are denoted as PIC-FP and PIC-BGK. For a time step $t^n \to t^{n+1} = t^n + \Delta t$, the solution procedure per step consists of a half-step velocity-space update, a full-step position update, an electric field update, and a second half-step velocity-space update:

\begin{enumerate}[label=\textbf{Step \arabic*:}, leftmargin=*]
    \item \textbf{First half-step velocity update ($\Delta t / 2$).} First, advance particle velocities due to field acceleration over $\Delta t/2$:
    \begin{equation}
        \boldsymbol{v}_k^* = \boldsymbol{v}_k^n + \boldsymbol{a}_i^n \frac{\Delta t}{2},
    \end{equation}
    where $\boldsymbol{a}_i^n$ is the local cell acceleration evaluated at $t^n$. Then, compute the local cell-averaged bulk velocity $\boldsymbol{U}_i^n$ and temperature $T_i^n$, and update velocities according to the selected collision operator over $\Delta t/2$:
    \begin{itemize}
        \item \textbf{BGK collision:} Evaluate the collision probability $P_{\mathrm{col}} = 1 - \exp(-\Delta t / (2\tau_{B,i}))$. For each particle $P_k$, sample $\xi \sim \mathcal{U}(0,1)$. If $\xi < P_{\mathrm{col}}$, the particle velocity is resampled as:
        \begin{equation}
            \boldsymbol{v}_k^{n+1/2} \sim \mathcal{N}\left(\boldsymbol{U}_i^n, \, \frac{T_i^n}{2}\mathbb{I}\right).
        \end{equation}
        Otherwise, $\boldsymbol{v}_k^{n+1/2} = \boldsymbol{v}_k^*$.

        \item \textbf{Fokker--Planck collision:} Evaluate the local collision frequency $\nu_i$ and update the velocity via the exact OU solution over $\Delta t/2$:
        \begin{equation}
            \boldsymbol{v}_k^{n+1/2} = \boldsymbol{U}_i^n + e^{-\frac{1}{2}\nu_i \Delta t} \left( \boldsymbol{v}_k^* - \boldsymbol{U}_i^n \right) + \sqrt{\frac{T_i^n}{2} \left( 1 - e^{-\nu_i \Delta t} \right)} \, \boldsymbol{\xi}_k,
        \end{equation}
        where $\boldsymbol{\xi}_k \sim \mathcal{N}(\boldsymbol{0}, \mathbb{I})$ is a vector of independent standard normal random variables.
    \end{itemize}
    
    After the collision step, the cellwise conservation correction is applied.

    \item \textbf{Full-step particle transport ($\Delta t$).} Update particle positions via streaming over the full time step $\Delta t$:
    \begin{equation}
        \boldsymbol{x}_k^{n+1} = \boldsymbol{x}_k^n + \boldsymbol{v}_k^{n+1/2} \Delta t.
    \end{equation}

    \item \textbf{Electric field update and second half-step velocity update ($\Delta t / 2$).} Deposit particle charges onto the grid based on the updated positions $\boldsymbol{x}_k^{n+1}$, and solve the Poisson equation to update the self-consistent field acceleration $\boldsymbol{a}_i^{n+1}$ at $t^{n+1}$.

    Next, advance particle velocities due to field acceleration over the remaining $\Delta t/2$:
    \begin{equation}
        \boldsymbol{v}_k^{**} = \boldsymbol{v}_k^{n+1/2} + \boldsymbol{a}_i^{n+1} \frac{\Delta t}{2}.
    \end{equation}
    Finally, re-evaluate the macroscopic properties $\boldsymbol{U}_i^{n+1}$ and $T_i^{n+1}$ at $t^{n+1}$, and apply the second collision update over $\Delta t/2$ (using either the BGK or FP operator as described in Step 1) to obtain the final particle velocity $\boldsymbol{v}_k^{n+1}$.
\end{enumerate}
\section*{Acknowledgements}
The current research is supported by National Key R\&D Program of China (Grant Nos. 2022YFA1004500), National Science Foundation of China (92371107), and Hong Kong research grant council (16208324).\\

\noindent\textbf{Declaration of interests}: The authors report no conflict of interest.

\newpage
\noindent \textbf{Reference}
\bibliographystyle{elsarticle-num}
\bibliography{ref}

\end{document}